\documentclass[prodmode]{acmsmall}
\usepackage[T1]{fontenc} 
\usepackage[utf8]{inputenc}
\usepackage{lmodern} 
\normalfont 
\DeclareFontShape{T1}{lmr}{bx}{sc} { <-> ssub * cmr/bx/sc }{}
\usepackage{dsfont} 
\usepackage{booktabs}
\usepackage{multirow}
\usepackage{graphicx}
\usepackage[numbers,sort&compress,square]{natbib}
\usepackage[hidelinks]{hyperref}
\usepackage{siunitx}
\usepackage{algorithm}
\floatname{algorithm}{ALGORITHM}
\usepackage[noend]{algpseudocode}
\algrenewcommand\alglinenumber[1]{\small #1:}

\usepackage{pifont}
\newcommand{\cmark}{\ding{51}}%
\newcommand{\xmark}{\ding{55}}%

\usepackage{subfigure} 

\usepackage{xspace}

\usepackage{mathtools}
\usepackage{amssymb}
\newtheorem{fact}[theorem]{Fact}

\newcommand*{\algoname}{tri\`est}
\newcommand*{\algonamecaps}{TRI\`EST}

\newcommand*{\algo}{\textsc{\algoname}\xspace}
\newcommand*{\algobase}{\textsc{\algoname-base}\xspace}
\newcommand*{\algofd}{\textsc{\algoname-fd}\xspace}
\newcommand*{\algoimproved}{\textsc{\algoname-impr}\xspace}
\newcommand*{\neigh}{\mathcal{N}}
\newcommand*{\variance}{\mathrm{Var}}
\newcommand*{\Sam}{\mathcal{S}}

\acmVolume{XX}
\acmNumber{X}
\acmArticle{XX}
\acmYear{2016}
\acmMonth{0}

\setcopyright{rightsretained}

\doi{0000001.0000001}

\issn{1234-56789}

\begin{document}

\markboth{Lorenzo De Stefani et al.}{\algonamecaps: Counting Triangles in
Fully-dynamic Streams with Fixed Memory Size}

\title{\algonamecaps: Counting Local and Global Triangles in Fully-dynamic
Streams with Fixed Memory Size}

\author{LORENZO DE STEFANI and ELI UPFAL
  \affil{Brown University}
  ALESSANDRO EPASTO
  \affil{Google}
  MATTEO RIONDATO
  \affil{Two Sigma Investments, LP}
}

\begin{abstract}
  \emph{``Ogni lassada xe persa''}\footnote{Any missed chance is lost forever.}
  -- Proverb from Trieste, Italy.
  \smallskip

  %
  We present \algo, a suite of one-pass streaming algorithms to compute
  unbiased, low-variance, high-quality approximations of the global and local
  (i.e., incident to each vertex) number of triangles in a fully-dynamic graph
  represented as an adversarial stream of edge insertions and deletions.

  Our algorithms use reservoir sampling and its variants to exploit the
  user-specified memory space at all times. This is in contrast with previous
  approaches, which require hard-to-choose parameters (e.g., a fixed sampling
  probability) and offer no guarantees on the amount of memory they use. We
  analyze the variance of the estimations and show novel concentration bounds
  for these quantities.

  Our experimental results on very large graphs demonstrate that \algo
  outperforms state-of-the-art approaches in accuracy and exhibits a small
  update time.
\end{abstract}

%
%

\acmformat{Lorenzo De Stefani, Alessandro Epasto, Matteo Riondato, and Eli
Upfal. 2016. \algonamecaps: Counting Local and Global Triangles in Fully-dynamic
Streams with Fixed Memory Size.}

\begin{bottomstuff}
  A preliminary report of this work appeared in the proceedings of ACM KDD'16
  as~\citep{DeStefaniERU16}.

  This work was supported in part by NSF grant IIS-1247581 and NIH grant
  R01-CA180776.

  Authors' addresses: Lorenzo De Stefani {and} Eli Upfal, Department of Computer
  Science, Brown University, email: \{lorenzo,eli\}@cs.brown.edu; Alessandro
  Epasto, Google Inc., email: aepasto@google.com; Matteo Riondato, Two
  Sigma Investments LP, email: matteo@twosigma.com.
\end{bottomstuff}

\maketitle

\section{Introduction}\label{sec:intro}
Exact computation of characteristic quantities of Web-scale networks is often
impractical or even infeasible due to the humongous size of these graphs. It is
natural in these cases to resort to \emph{efficient-to-compute approximations}
of these quantities that, when of sufficiently high quality, can be used as
proxies for the exact values.

In addition to being huge, many interesting networks are \emph{fully-dynamic}
and can be represented as a \emph{stream} whose elements are edges/nodes
insertions and deletions which occur in an \emph{arbitrary} (even adversarial)
order. Characteristic quantities in these graphs are \emph{intrinsically
volatile}, hence there is limited added value in maintaining them exactly.
The goal is rather to keep track, \emph{at all times}, of a high-quality
approximation of these quantities. For efficiency, the algorithms should
\emph{aim at exploiting the available memory space as much as possible} and they
should \emph{require only one pass over the stream}.

We introduce \algo, a suite of \emph{sampling-based, one-pass algorithms for
adversarial fully-dynamic streams to approximate the global number of triangles
and the local number of triangles incident to each vertex}. Mining local and
global triangles is a fundamental primitive with many applications (e.g.,
community detection~\citep{BerryHLVP11}, topic
mining~\citep{EckmannM02}, spam/anomaly
detection~\citep{BecchettiBCG10,LimK15},
ego-networks mining~\cite{epasto2015ego}
and protein interaction networks analysis~\citep{MiloSOIKA02}.)

Many previous works on triangle estimation in streams also employ sampling (see
Sect.~\ref{sec:relwork}), but they usually require the user to
specify \emph{in advance} an \emph{edge sampling probability $p$} that is fixed
for the entire stream. This approach presents several significant drawbacks.
First, choosing a $p$ that allows to obtain the desired approximation
quality requires to know or guess a number of properties of the input (e.g., the
size of the stream). Second, a fixed $p$ implies that the sample size grows with
the size of the stream, which is problematic when the stream size is not known
in advance: if the user specifies a large $p$, the algorithm may run out of
memory, while for a smaller $p$ it will provide a suboptimal estimation.
Third, even assuming to be able to compute a $p$ that ensures (in expectation)
full use of the available space, the memory would be fully utilized only at the
end of the stream, and the estimations computed throughout the execution would
be suboptimal.

\paragraph{Contributions}
We address all the above issues by taking a significant departure from
the fixed-probability, independent edge sampling approach taken even by
state-of-the-art methods~\citep{LimK15}. Specifically:

\begin{longitem}
  \item We introduce \algo (\emph{TRI}angle \emph{E}stimation from
    \emph{ST}reams), a suite of \emph{one-pass streaming algorithms} to
    approximate, at each time instant, the global and local number of
    triangles in a \emph{fully-dynamic} graph stream (i.e., a sequence of
    edges additions and deletions in arbitrary order) using a \emph{fixed
    amount of memory}. This is the first contribution that enjoys all these
    properties. \algo only requires the user to specify \emph{the amount of
    available memory}, an interpretable parameter that is definitively known
    to the user.
  \item Our algorithms maintain a sample of edges: they use the \emph{reservoir
    sampling}~\citep{Vitter85} and \emph{random pairing}~\citep{GemullaLH08}
    sampling schemes to exploit the available memory as much as possible. To the
    best of our knowledge, ours is the first application of these techniques to
    subgraph counting in fully-dynamic, arbitrarily long, adversarially ordered
    streams. We present an analysis of the unbiasedness and of the variance of
    our estimators, and establish strong concentration results for them. The use
    of reservoir sampling and random pairing requires additional sophistication
    in the analysis, as the presence of an edge in the sample is \emph{not
    independent} from the concurrent presence of another edge. Hence, in our
    proofs we must consider the complex dependencies in events involving sets of
    edges. The gain is worth the effort: we prove that the variance of our
    algorithms is smaller than that of state-of-the-art methods~\citep{LimK15},
    and this is confirmed by our experiments.
  \item We conduct an extensive experimental evaluation of \algo on very large
    graphs, some with billions of edges, comparing the performances of our
    algorithms to those of existing state-of-the-art
    contributions~\citep{LimK15,JhaSP15,PavanTTW13}. \emph{Our algorithms
    significantly and consistently reduce the average estimation error by up to
    $90\%$} w.r.t.~the state of the art, both in the global and local estimation
    problems, while using the same amount of memory. Our algorithms are also
    extremely scalable, showing update times in the order of hundreds of
    microseconds for graphs with billions of edges.
\end{longitem}

In this article, we extend the conference version~\citep{DeStefaniERU16} in
multiple ways. First of all, we include all proofs of our theoretical results
and give many additional details that were omitted from the conference version
due to lack of space. Secondly, we strengthen the analysis of \algo, presenting
tighter bounds to the variance of its variants. Thirdly, we show how to extend
\algo to approximate the count of triangles in multigraphs. Additionally, we
include a whole subsection of discussion of our results, highlighting their
advantages, disadvantages, and limitations. Finally, we expand our experimental
evaluation, reporting the results of additional experiments and giving
additional details on the comparison with existing state-of-the-art methods.

\paragraph{Paper organization}
We formally introduce the settings and the problem in
Sect.~\ref{sec:prelims}. In Sect.~\ref{sec:relwork} we discuss related works.
We present and analyze \algo and discuss our design choices in
Sect.~\ref{sec:algorithms}. The results of our experimental evaluation are
presented in Sect.~\ref{sec:experiments}. We draw our conclusions in
Sect.~\ref{sec:concl}. Some of the proofs of our theoretical results are
deferred to Appendix~\ref{sec:appendix}.

\section{Preliminaries}\label{sec:prelims}
We study the problem of counting global and local triangles in a fully-dynamic
undirected graph as an arbitrary (adversarial) stream of edge insertions and
deletions.

Formally, for any (discrete) time instant $t\ge 0$, let
$G^{(t)}=(V^{(t)},E^{(t)})$ be the graph observed up to and including time $t$.
At time $t=0$ we have $V^{(t)}=E^{(t)}=\emptyset$. For any $t>0$, at time $t+1$
we receive an element $e_{t+1}=(\bullet,(u,v))$  from a stream, where
$\bullet\in\{+,-\}$ and $u,v$ are two distinct vertices. The graph
$G^{(t+1)}=(V^{(t+1)},E^{(t+1)})$ is obtained by \emph{inserting a new edge or
deleting an existing edge} as follows:
\[
  E^{(t+1)}=\left\{
    \begin{array}{ll}
      E^{(t)}\cup\{(u,v)\}\mbox{ if } \bullet=\mbox{``}+\mbox{''}\\
      E^{(t)}\setminus\{(u,v)\}\mbox{ if } \bullet=\mbox{``}-\mbox{''}
    \end{array}
  \right.\enspace.
\]
If $u$ or $v$ do not belong to $V^{(t)}$, they are added to $V^{(t+1)}$. Nodes
are deleted from $V^{(t)}$ when they have degree zero.

Edges can be added and deleted in the graph in an arbitrary adversarial order,
i.e., as to cause the worst outcome for the algorithm, but we assume that the
adversary has \emph{no access to the random bits} used by the algorithm.
We assume that \emph{all operations have effect}: if $e\in E^{(t)}$
(resp.~$e\not\in E^{(t)}$), $(+,e)$ (resp.~$(-,e)$) can not be on the stream at
time $t+1$.

Given a graph $G^{(t)}=(V^{(t)},E^{(t)})$, a \emph{triangle} in $G^{(t)}$ is a
\emph{set} of three edges $\{(u,v),(v,w),(w,u)\}\subseteq E^{(t)}$, with $u$,
$v$, and $w$ being three distinct vertices. We refer to $\{u,v,w\}\subseteq
V^{(t)}$ as the \emph{corners} of the triangle. We denote with $\Delta^{(t)}$
the set of \emph{all} triangles in $G^{(t)}$, and, for any vertex $u\in
V^{(t)}$, with $\Delta^{(t)}_u$ the subset of $\Delta^{(t)}$ containing all and
only the triangles that have $u$ as a corner.

{\em Problem definition.} We study the \emph{Global} (resp.~\emph{Local})
Triangle Counting Problem in Fully-dynamic Streams, which requires to compute,
at each time $t\ge 0$ an estimation of $|\Delta^{(t)}|$ (resp.~for each $u\in V$
an estimation of $|\Delta^{(t)}_u|$).

\paragraph{Multigraphs} Our approach can be further extended to count the number
of global and local triangles on a \emph{multigraph} represented as a stream of
edges. Using a formalization analogous to that discussed for graphs, for any
(discrete) time instant $t\ge 0$, let $G^{(t)}=(V^{(t)},\mathcal{E}^{(t)})$ be
the multigraph observed up to and including time $t$, where $\mathcal{E}^{(t)}$
is now a \emph{bag} of edges between vertices of $V^{(t)}$. The multigraph
evolves through a series of edges additions and deletions according to the same
process described for graphs. The definition of triangle in a multigraph is also
the same. As before we denote with $\Delta^{(t)}$ the set of \emph{all}
triangles in $G^{(t)}$, but now this set may contain multiple triangles with the
same set of vertices, although each of these triangles will be a different set
of three edges among those vertices. For any vertex $u\in V^{(t)}$, we still
denote with $\Delta_u^{(t)}$ the subset of $\Delta^{(t)}$ containing all and
only the triangles that have $u$ as a corner, with a similar caveat as
$\Delta^{(t)}$. The problems of global and local triangle counting in multigraph
edge streams are defined exactly in the same way as for graph edge streams.

\section{Related work}\label{sec:relwork}
The literature on exact and approximate triangle counting is extremely rich,
including exact algorithms, graph
sparsifiers~\citep{TsourakakisKMF09,TsourakakisKM11}, complex-valued
sketches~\citep{ManjunathMPS11,KaneMSS12}, and MapReduce
algorithms~\citep{SuriV11,PaghT12,ParkC13,ParkSKP14,ParkMK16}. Here we restrict the discussion to
the works most related to ours, i.e., to those presenting algorithms for
counting or approximating the number of triangles from data streams. We refer
to the survey by~\citet{Latapy08} for an in-depth discussion of other works.
Table~\ref{table:comparison} presents a summary of the comparison, in terms of
desirable properties, between this work and relevant previous contributions.

Many authors presented algorithms for more restricted (i.e., less generic)
settings than ours, or for which the constraints on the computation are more
lax~\citep{BarYossefKS02,BuriolFSMSS06,JowharyG05,KutzkovP13}. For example,
\citet{BecchettiBCG10} and~\citet{KolountzakisMPT12} present algorithms for
approximate triangle counting from \emph{static} graphs by performing multiple
passes over the data. \citet{PavanTTW13} and \citet{JhaSP15} propose algorithms
for approximating only the global number of triangles from
\emph{edge-insertion-only} streams. \citet{BulteauFKP15} present a one-pass
algorithm for fully-dynamic graphs, but the triangle count estimation is
(expensively) computed only at the end of the stream and the algorithm requires,
in the worst case, more memory than what is needed to store the entire graph.
\citet{AhmedDKN14} apply the sampling-and-hold approach to insertion-only graph
stream mining to obtain, only at the end of the stream and using non-constant
space, an estimation of many network measures including triangles.

None of these works has \emph{all} the features offered by \algo: performs a
single pass over the data, handles fully-dynamic streams, uses a fixed amount of
memory space, requires a single interpretable parameter, and returns an
estimation at each time instant. Furthermore, our experimental results show that
we outperform the algorithms from~\citep{PavanTTW13,JhaSP15} on insertion-only
streams.

\citet{LimK15} present an algorithm for insertion-only streams that is based on
independent edge sampling with a fixed probability:  for each edge on the
stream, a coin with a user-specified fixed tails probability $p$ is flipped,
and, if the outcome is tails, the edge is added to the stored sample and the
estimation of local triangles is updated. Since the memory is not fully utilized
during most of the stream, the variance of the estimate is large. Our approach
handles fully-dynamic streams and makes better use of the available memory space
at each time instant, resulting in a better estimation, as shown by our
analytical and experimental results.

\citet{Vitter85} presents a detailed analysis of the reservoir sampling scheme
and discusses methods to speed up the algorithm by reducing the number of calls to
the random number generator. Random Pairing~\citep{GemullaLH08} is an extension
of reservoir sampling to handle fully-dynamic streams with insertions and
deletions. \citet{CohenCD12} generalize and extend the Random Pairing approach
to the case where the elements on the stream are key-value pairs, where the
value may be negative (and less than $-1$). In our settings, where the value is
not less than $-1$ (for an edge deletion), these generalizations do not apply
and the algorithm presented by~\citet{CohenCD12} reduces essentially to Random
Pairing.

\begin{table}[ht]
  \tbl{Comparison with previous contributions
  }{
    \begin{tabular}{cccccc}
      \toprule
      Work &  \shortstack[c]{Single\\pass} & \shortstack[c]{Fixed\\space} &
      \shortstack[c]{Local\\counts} & \shortstack[c]{Global\\counts} &
      \shortstack[c]{Fully-dynamic\\streams}\\
      \midrule
      \citep{BecchettiBCG10} & \xmark & \cmark/\xmark$^\dagger$ & \cmark & \xmark & \xmark\\
      \citep{KolountzakisMPT12} & \xmark & \xmark & \xmark & \cmark & \xmark \\
      \citep{PavanTTW13} & \cmark & \cmark & \xmark & \cmark & \xmark \\
      \citep{JhaSP15} & \cmark & \cmark & \xmark & \cmark & \xmark \\
      \citep{AhmedDKN14} & \cmark & \xmark & \xmark & \cmark & \xmark \\
      \citep{LimK15} & \cmark & \xmark & \xmark & \xmark & \xmark \\
      This work & \cmark & \cmark & \cmark & \cmark & \cmark \\
      \bottomrule
    \end{tabular}
  }
  \begin{tabnote}
    \tabnoteentry{$^\dagger$}{The required space is $O(|V^{(t)}|)$, which,
      although not dependent on the number of triangles or on the number of
      edges, is not fixed in the sense that it can be fixed a-priori.}
  \end{tabnote}
  \label{table:comparison}
\end{table}

\section{Algorithms}\label{sec:algorithms}
We present \algo, a suite of three novel algorithms for approximate global and
local triangle counting from edge streams. The first two work on insertion-only
streams, while the third can handle fully-dynamic streams where edge deletions
are allowed. We defer the discussion of the multigraph case to
Sect.~\ref{sec:multigraphs}.

\paragraph{Parameters}
Our algorithms keep an edge sample $\Sam$ containing up to $M$ edges from the
stream, where $M$ is a positive integer parameter. For ease of presentation, we
realistically assume $M\ge 6$. In Sect.~\ref{sec:intro} we motivated the design
choice of only requiring $M$ as a parameter and remarked on its advantages over
using a fixed sampling probability $p$. Our algorithms are designed to use the
available space as much as possible.

\paragraph{Counters}
\algo algorithms keep \emph{counters} to compute the estimations of the global
and local number of triangles. They \emph{always} keep one global counter $\tau$
for the estimation of the global number of triangles. Only the global counter is
needed to estimate the total triangle count. To estimate the local triangle
counts, the algorithms keep a set of local counters $\tau_u$ for a subset of the
nodes $u\in V^{(t)}$. The local counters are created on the fly as needed, and
\emph{always} destroyed as soon as they have a value of $0$. Hence our
algorithms use $O(M)$ space (with one exception, see Sect.~\ref{sec:improved}).

\paragraph{Notation}
For any $t\ge 0$, let $G^\Sam=(V^\Sam,E^\Sam)$ be the subgraph of $G^{(t)}$
containing all and only the edges in the current sample $\Sam$. We denote with
$\neigh^{\Sam}_u$ the \emph{neighborhood} of $u$ in $G^\Sam$:
$\neigh^{\Sam}_u=\{v\in V^{(t)} ~:~ (u,v)\in \Sam\}$ and with
$\neigh^{\Sam}_{u,v}= \neigh^{\Sam}_u \cap \neigh^{\Sam}_v$ the \emph{shared
neighborhood} of $u$ and $v$ in $G^\Sam$.

\paragraph{Presentation}
We only present the analysis of our algorithms for the problem of \emph{global}
triangle counting. For each presented result involving the estimation of the
global triangle count (e.g., unbiasedness, bound on variance, concentration
bound) and potentially using other global quantities (e.g., the number of
pairs of triangles in $\Delta^{(t)}$ sharing an edge), it is straightforward to
derive the correspondent variant for the estimation of the local triangle count,
using similarly defined local quantities  (e.g., the number of
pairs of triangles in $\Delta_u^{(t)}$ sharing an edge.)

\subsection{A first algorithm -- \algobase}\label{sec:algobase}
We first present \algobase, which works on insertion-only streams and uses
standard reservoir sampling~\citep{Vitter85} to maintain the edge sample $\Sam$:
\begin{itemize}
  \item If $t\le M$, then the edge $e_t=(u,v)$ on the stream at time $t$ is
    deterministically inserted in $\Sam$.
  \item If $t>M$, \algobase flips a biased coin with heads probability $M/t$. If
    the outcome is heads, it chooses an edge $(w,z)\in\Sam$ uniformly at
    random, removes $(w,z)$ from $\Sam$, and inserts $(u,v)$ in $\Sam$.
    Otherwise, $\Sam$ is not modified.
\end{itemize}
After each insertion (resp.~removal) of an edge $(u,v)$ from $\Sam$, \algobase
calls the procedure \textsc{UpdateCounters} that increments (resp.~decrements)
$\tau$, $\tau_u$ and $\tau_v$ by $|\neigh^\Sam_{u,v}|$, and $\tau_c$ by one, for
each $c\in\neigh^{\Sam}_{u,v}$.

The pseudocode for \algobase is presented in Alg.~\ref{alg:triest-base}.

\begin{algorithm}[ht]
  \small 
  \caption{\algobase}
  \label{alg:triest-base}
  \begin{algorithmic}[1]
    \Statex{\textbf{Input:} Insertion-only edge stream $\Sigma$, integer $M\ge6$}
    \State $\Sam\leftarrow\emptyset$, $t \leftarrow 0$, $\tau \leftarrow 0$
    \For{ {\bf each} element $(+,(u,v))$ from $\Sigma$}
    \State $t\leftarrow t +1$
    \If{\textsc{SampleEdge}$((u,v), t )$}
      \State $\Sam \leftarrow \Sam\cup \{(u,v)\}$
      \State \textsc{UpdateCounters}$(+,(u,v))$
    \EndIf
    \EndFor
    \Statex
    \Function{\textsc{SampleEdge}}{$(u,v),t$}
      \If {$t\leq  M$}
        \State \textbf{return} True
      \ElsIf{\textsc{FlipBiasedCoin}$(\frac{M}{t}) = $ heads}
        \State $(u',v') \leftarrow$ random edge from $\Sam$
        \State $\Sam\leftarrow \Sam\setminus \{(u',v')\}$
        \State \textsc{UpdateCounters}$(-,(u',v'))$
        \State \textbf{return} True
      \EndIf
      \State \textbf{return} False
    \EndFunction
    \Statex
    \Function{UpdateCounters}{$(\bullet, (u,v))$}
      \State $\mathcal{N}^\Sam_{u,v} \leftarrow \mathcal{N}^\Sam_u \cap \mathcal{N}^\Sam_v$
      \ForAll {$c \in \mathcal{N}^\Sam_{u,v}$}
        \State $\tau \leftarrow \tau \bullet 1$
        \State $\tau_c \leftarrow \tau_c \bullet 1$
        \State $\tau_u \leftarrow \tau_u \bullet 1$
        \State $\tau_v \leftarrow \tau_v \bullet 1$
      \EndFor
    \EndFunction
  \end{algorithmic}
\end{algorithm}

\subsubsection{Estimation}
 For any pair of positive integers $a$ and $b$
 such that $a\le\min\{M,b\}$ let
\[
  \xi_{a,b} =
  \left\{\begin{array}{ccl} & 1 & \text{if } b\le M\\
  \displaystyle\binom{b}{M}\Big/\displaystyle\binom{b-a}{M-a}&=\displaystyle\prod_{i=0}^{a-1}\frac{b-i}{M-i} & \text{otherwise}\end{array}\right.
  \enspace.
\]
As shown in the following lemma, $\xi_{k,t}^{-1}$ is the probability that $k$
edges of $G^{(t)}$ are all in $\Sam$ at time $t$, i.e., the $k$-th order
inclusion probability of the reservoir sampling scheme. The proof can be found
in App.~\ref{app:algobase}.

\begin{lemma}\label{lem:reservoirhighorder}
  For any time step $t$ and any positive integer $k\le t$, let $B$ be any
  subset of $E^{(t)}$ of size $|B|=k\le t$. Then, at the end of time step $t$,
  \[
    \Pr(B\subseteq\Sam)=\left\{\begin{array}{cl}0 & \text{if } k>M\\
    \xi_{k,t}^{-1} & \text{otherwise}\end{array} \right.\enspace.
  \]
\end{lemma}
We make use of this lemma in the analysis of \algobase.

Let, for any $t\ge 0$, $\xi^{(t)}= \xi_{3,t}$ and let $\tau^{(t)}$
(resp.~$\tau_u^{(t)}$) be the value of the counter $\tau$ at the end of time
step $t$ (i.e., after the edge on the stream at time $t$ has been processed by
\algobase) (resp.~the value of the counter $\tau_u$ at the end of time step $t$
if there is such a counter, 0 otherwise). When queried at the end of time $t$,
\algobase returns $\xi^{(t)}\tau^{(t)}$ (resp.~$\xi^{(t)}\tau_u^{(t)}$) as the
estimation for the global (resp.~local for $u\in V^{(t)}$) triangle count.

\subsubsection{Analysis}
We now present the analysis of the estimations computed by \algobase.
Specifically, we prove their unbiasedness (and their exactness for $t\le M$) and
then show an exact derivation of their variance and a concentration result. We
show the results for the global counts, but results analogous to those in
Thms.~\ref{thm:baseunbiased}, \ref{thm:basevariance},
and~\ref{thm:baseconcentration} hold for the local triangle count for any $u\in
V^{(t)}$, replacing the global quantities with the corresponding local ones. We
also compare, theoretically, the variance of \algobase with that of a
fixed-probability edge sampling approach~\citep{LimK15}, showing that \algobase
has smaller variance for the vast majority of the stream.

\subsubsection{Expectation}
We have the following result about the estimations computed by \algobase.

\begin{theorem}\label{thm:baseunbiased}
  We have
  \begin{align*}
    \xi^{(t)}\tau^{(t)}=\tau^{(t)}=|\Delta^{(t)}| &\mbox{ if } t\le M\\
    \mathbb{E}\left[\xi^{(t)}\tau^{(t)}\right]=|\Delta^{(t)}| &\mbox{ if } t>
    M\enspace.
  \end{align*}
\end{theorem}

The \algobase estimations are not only \emph{unbiased in all cases}, but
actually \emph{exact for $t\le M$}, i.e., for $t\le M$, they are the true
global/local number of triangles in $G^{(t)}$.

To prove Thm.~\ref{thm:baseunbiased}, we need to introduce a technical
lemma. Its proof can be found in Appendix~\ref{app:algobase}. We denote with
$\Delta^{\Sam}$ the set of triangles in $G^{\Sam}$.

\begin{lemma}\label{lem:baseunbiasedaux}
  After each call to \textsc{UpdateCounters}, we have $\tau=|\Delta^\Sam|$
  and $\tau_v=|\Delta_v^\Sam|$ for any $v\in V_\Sam$ s.t. $|\Delta_v^\Sam|\ge
  1$.
\end{lemma}

From here, the proof of Thm.~\ref{thm:baseunbiased} is a straightforward
application of Lemma~\ref{lem:baseunbiasedaux} for the case $t\le M$ and of that
lemma, the definition of expectation, and Lemma~\ref{lem:reservoirhighorder}
otherwise. The complete proof can be found in App.~\ref{app:algobase}.

\subsubsection{Variance}
We now analyze the variance of the estimation returned by \algobase for $t>M$
(the variance is $0$ for $t\le M$.)

Let $r^{(t)}$ be the \emph{total} number of unordered pairs of distinct
triangles from $\Delta^{(t)}$ sharing an edge,\footnote{Two distinct triangles
can share at most one edge.} and $w^{(t)}=\binom{|\Delta^{(t)}|}{2}-r^{(t)}$ be
the number of unordered pairs of distinct triangles that do not share any edge.

\begin{theorem}\label{thm:basevariance}
  For any $t>M$, let $f(t) = \xi^{(t)}-1$,
  \[
    g(t) = \xi^{(t)}\frac{(M-3)(M-4)}{(t-3)(t-4)} -1
  \]
  and
  \[
    h(t) = \xi^{(t)}\frac{(M-3)(M-4)(M-5)}{(t-3)(t-4)(t-5)}
    -1\enspace(\le 0).
  \]
  We have:
  \begin{equation}\label{eq:globvariancebase}
    \variance\left[\xi(t)\tau^{(t)}\right] = |\Delta^{(t)}|
    f(t)+r^{(t)}g(t)+w^{(t)}h(t).
  \end{equation}
\end{theorem}

In our proofs, we carefully account for the fact that, as we use reservoir
sampling~\citep{Vitter85}, the presence of an edge $a$ in $\Sam$ is \emph{not
independent} from the concurrent presence of another edge $b$ in $\Sam$. This is
not the case for samples built using fixed-probability independent edge
sampling, such as \textsc{mascot}~\citep{LimK15}. When computing the variance,
we must consider not only pairs of triangles that share an edge, as in the case
for independent edge sampling approaches, but also pairs of triangles sharing no
edge, since their respective presences in the sample are not independent events.
The gain is worth the additional sophistication needed in the analysis, as the
contribution to the variance by triangles no sharing edges is
\emph{non-positive} ($h(t)\le 0$), i.e., it reduces the variance. A comparison
of the variance of our estimator with that obtained with a fixed-probability
independent edge sampling approach, is discussed in Sect.~\ref{sec:comparison}.

\begin{proof}[of Thm.~\ref{thm:basevariance}]
  Assume $|\Delta^{(t)}|>0$, otherwise the estimation is deterministically
  correct and has variance 0 and the thesis holds. Let $\lambda\in\Delta^{(t)}$
  and $\delta_\lambda^{(t)}$ be as in the proof of Thm.~\ref{thm:baseunbiased}.
  We have $\variance[\delta_\lambda^{(t)}]=\xi^{(t)}-1$ and from this and the
  definition of variance and covariance we obtain
  \begin{align}
    \variance\left[\xi^{(t)}\tau^{(t)}\right] &= \variance\left[\sum_{\lambda\in\Delta^{(t)}}\delta_\lambda^{(t)}\right]
    = \sum_{\lambda\in \Delta^{(t)}} \sum_{\gamma \in
    \Delta^{(t)}}\textrm{Cov}\left[\delta_\lambda^{(t)} ,\delta_\gamma^{(t)}\right]
    \nonumber\\
    &=  \sum_{\lambda\in \Delta^{(t)}}
    \variance\left[\delta_\lambda^{(t)}\right] + \sum_{\substack{\lambda,
    \gamma\in \Delta^{(t)}\\ \lambda \neq \gamma}}
    \textrm{Cov}\left[\delta_\lambda^{(t)}
    ,\delta_\gamma^{(t)}\right]\nonumber\\
    &= |\Delta^{(t)}|(\xi^{(t)}-1) + \sum_{\substack{\lambda,
    \gamma\in \Delta^{(t)}\\ \lambda \neq \gamma}}
    \textrm{Cov}\left[\delta_\lambda^{(t)}
    ,\delta_\gamma^{(t)}\right]\nonumber\\
    &= |\Delta^{(t)}|(\xi^{(t)}-1) + \sum_{\substack{\lambda,
    \gamma\in \Delta^{(t)}\\ \lambda \neq \gamma}}
    \left(\mathbb{E}\left[\delta_\lambda^{(t)}\delta_\gamma^{(t)}\right] -
    \mathbb{E}\left[\delta_\lambda^{(t)}\right]\mathbb{E}\left[\delta_\gamma^{(t)}\right]\right)\nonumber\\
    &= |\Delta^{(t)}|(\xi^{(t)}-1) + \sum_{\substack{\lambda,
    \gamma\in \Delta^{(t)}\\ \lambda \neq \gamma}}\left(
    \mathbb{E}\left[\delta_\lambda^{(t)}\delta_\gamma^{(t)}\right] - 1\right)
    \enspace.\label{eq:proofvariancebase}
  \end{align}

  Assume now $|\Delta^{(t)}|\ge 2$, otherwise we have $r^{(t)}=w^{(t)}=0$ and
  the thesis holds as the second term on the
  r.h.s.~of~\eqref{eq:proofvariancebase} is 0. Let $\lambda$ and $\gamma$ be
  two distinct triangles in $\Delta^{(t)}$. If $\lambda$ and $\gamma$ do not
  share an edge, we have
  $\delta_\lambda^{(t)}\delta_\gamma^{(t)}=\xi^{(t)}\xi^{(t)}=\xi_{3,t}^2$
  if all \emph{six} edges composing $\lambda$ and $\gamma$ are in $\Sam$ at
  the end of time step $t$, and $\delta_\lambda^{(t)}\delta_\gamma^{(t)}=0$
  otherwise. From Lemma~\ref{lem:reservoirhighorder} we then have that
  \begin{align}
    \mathbb{E}\left[\delta_\lambda^{(t)}\delta_\gamma^{(t)}\right]&=\xi_{3,t}^{2}\Pr\left(\delta_\lambda^{(t)}\delta_\gamma^{(t)}=\xi_{3,t}^{2}\right)=\xi_{3,t}^{2}\frac{1}{\xi_{6,t}}=\xi_{3,t}\prod_{j=3}^{5}\frac{M-j}{t-j}\nonumber\\
    &=\xi^{(t)}\frac{(M-3)(M-4)(M-5)}{(t-3)(t-4)(t-5)}\label{eq:expectprodnoshare}\enspace.
  \end{align}
  If instead $\lambda$ and $\gamma$ share exactly an edge we have
  $\delta_\lambda^{(t)}\delta_\gamma^{(t)}=\xi_{3,t}^{2}$ if all \emph{five}
  edges composing $\lambda$ and $\gamma$ are in $\Sam$ at the end of time step
  $t$, and $\delta_\lambda^{(t)}\delta_\gamma^{(t)}=0$ otherwise. From
  Lemma~\ref{lem:reservoirhighorder} we then have that
  \begin{align}
    \mathbb{E}\left[\delta_\lambda^{(t)}\delta_\gamma^{(t)}\right]&=\xi_{3,t}^{2}\Pr\left(\delta_\lambda^{(t)}\delta_\gamma^{(t)}=\xi_{3,t}^{2}\right)=\xi_{3,t}^{2}\frac{1}{\xi_{5,t}}=\xi_{3,t}\prod_{j=3}^4\frac{M-j}{t-j}\nonumber\\
    &=\xi^{(t)}\frac{(M-3)(M-4)}{(t-3)(t-4)}\enspace.\label{eq:expectprodshare}
  \end{align}
  The thesis follows by combining~\eqref{eq:proofvariancebase},
  \eqref{eq:expectprodnoshare}, \eqref{eq:expectprodshare}, recalling the
  definitions of $r^{(t)}$ and $w^{(t)}$, and slightly reorganizing the terms.
\end{proof}

\subsubsection{Concentration}
We have the following concentration result on the estimation returned by
\algobase. Let $h^{(t)}$ denote the maximum number of triangles sharing a single
edge in $G^{(t)}$.

\begin{theorem}\label{thm:baseconcentration}
  Let $t\ge 0$ and assume $|\Delta^{(t)}| >0$.\footnote{If $|\Delta^{(t)}|=0$,
  our algorithms correctly estimate $0$ triangles.}
  For any $\varepsilon,\delta\in(0,1)$, let
  \[
    \Phi = \sqrt[3]{8\varepsilon^{-2}\frac{3h^{(t)}+1}{|\Delta^{(t)}|}\ln
    \left(\frac{(3h^{(t)}+1)e}{\delta}\right)}\enspace.
  \]
  If
  \[
    M \ge \max \left\{t \Phi \left(1+\frac{1}{2}\ln^{2/3} \left (t \Phi
    \right)\right), 12\varepsilon^{-1}+e^2, 25\right\},
  \]
  then $|\xi^{(t)}\tau^{(t)}-|\Delta^{(t)}||<\varepsilon|\Delta^{(t)}|$ with probability
  $>1-\delta$.
\end{theorem}

The roadmap to proving Thm.~\ref{thm:baseconcentration} is the following:
\begin{longenum}
  \item we first define two simpler algorithms, named \textsc{indep} and
    \textsc{mix}. The algorithms use, respectively, fixed-probability
    independent sampling of edges and reservoir sampling (but with a different
    estimator than the one used by \algobase);
  \item we then prove concentration results on the estimators of
    \textsc{indep} and \textsc{mix}. Specifically, the concentration result for
    \textsc{indep} uses a result by~\citet{hajnal1970proof} on graph coloring,
    while the one for \textsc{mix} will depend on the concentration result for
    \textsc{indep} and on a Poisson-approximation-like technical result stating
    that probabilities of events when using reservoir sampling are close to the
    probabilities of those events when using fixed-probability independent
    sampling;
  \item we then show that the estimates returned by \algobase are close to
    the estimates returned by \textsc{mix};
  \item finally, we combine the above results and show that, if $M$ is large
    enough, then the estimation provided by \textsc{mix} is likely to be
    close to $|\Delta^{(t)}|$ and since the estimation computed by \algobase
    is close to that of \textsc{mix}, then it must also be close to
    $|\Delta^{(t)}|$.
\end{longenum}
\emph{Note:} for ease of presentation, in the following we use $\phi^{(t)}$ to
denote the estimation returned by \algobase, i.e.,
$\phi^{(t)}=\xi^{(t)}\tau^{(t)}$.

\paragraph{The \textsc{indep} algorithm}
The \textsc{indep} algorithm works as follows: it creates a sample
$\Sam_\textsc{in}$ by sampling edges in $E^{(t)}$ independently with a fixed
probability $p$. It estimates the global number of triangles in $G^{(t)}$ as
\[
  \phi_\textsc{in}^{(t)}= \frac{\tau_\textsc{in}^{(t)}}{p^3},
\]
where $\tau_\text{\sc in}^{(t)}$ is the number of triangles in
$\Sam_\textsc{in}$. This is for example the approach taken
by \textsc{mascot-c}~\citep{LimK15}.

\paragraph{The \textsc{mix} algorithm}
The \textsc{mix} algorithm works as follows: it uses reservoir sampling (like
\algobase) to create a sample $\Sam_\textsc{mix}$ of $M$ edges from $E^{(t)}$,
but uses a different estimator than the one used by \algobase. Specifically,
\textsc{mix} uses
\[
  \phi_\textsc{mix}^{(t)}=\left(\frac{t}{M}\right)^3\tau^{(t)}
\]
as an estimator for $|\Delta^{(t)}|$, where $\tau^{(t)}$ is, as in \algobase,
the number of triangles in $G^\Sam$ (\algobase uses
$\phi^{(t)}=\frac{t(t-1)(t-2)}{M(M-1)(M-2)}\tau^{(t)}$ as an estimator.)

We call this algorithm \textsc{mix} because it uses reservoir sampling to create
the sample, but computes the estimate as if it used fixed-probability
independent sampling, hence in some sense it ``mixes'' the two approaches.

\paragraph{Concentration results for \textsc{indep} and \textsc{mix}}
We now show a concentration result for \textsc{indep}. Then we show a technical
lemma (Lemma~\ref{lem:equivMp}) relating the probabilities of events when using
reservoir sampling to the probabilities of those events when using
fixed-probability independent sampling.  Finally, we use these results to show
that the estimator used by \textsc{mix} is also concentrated
(Lemma~\ref{lem:concentrationmix}).

\begin{lemma}\label{lem:concentrationindependent}
  Let $t\ge 0$ and assume $|\Delta^{(t)}|>0$.\footnote{For $|\Delta^{(t)}|=0$,
    \textsc{indep} correctly and deterministically returns $0$ as the
  estimation.} For any $\varepsilon,\delta\in(0,1)$, if
  \begin{equation}\label{eq:requirementp}
    p \ge \sqrt[3]{2\varepsilon^{-2}\ln \left(\frac{3h^{(t)}+1}{\delta}\right)\frac{3h^{(t)}+1}{|\Delta^{(t)}|}}
  \end{equation}
  then
  \[
    \Pr\left(|\phi_\textsc{in}^{(t)}-\Delta^{(t)}||<\varepsilon|\Delta^{(t)}|\right)> 1-\delta\enspace.
  \]
\end{lemma}

\begin{proof}
  Let $H$ be a graph built as follows: $H$ has one node for each triangle in
  $G^{(t)}$ and there is an edge between two nodes in $H$ if the corresponding
  triangles in $G^{(t)}$ share an edge. By this construction, the maximum degree
  in $H$ is $3h^{(t)}$. Hence by the Hajanal-Szem\'eredi's
  theorem~\citep{hajnal1970proof} there is a proper coloring of $H$ with at most
  $3h^{(t)}+1$ colors such that for each color there are at least $L =
  \frac{|\Delta^{(t)}|}{3h^{(t)}+1}$ nodes with that color.

  Assign an arbitrary numbering to the triangles of $G^{(t)}$
  (and, therefore, to the nodes of $H$) and let $X_i$ be a Bernoulli random
  variable, indicating whether the triangle $i$ in $G^{(t)}$ is in the sample at time $t$.
  From the properties of independent sampling of edges we have
  $\Pr(X_i=1)=p^3$ for any triangle $i$.  For any color $c$ of the coloring of
  $H$, let $\mathcal{X}_c$ be the set of r.v.'s $X_i$ such that the node $i$
  in $H$ has color $c$. Since the coloring of $H$ which we are considering is
  proper, the r.v.'s in $\mathcal{X}_c$ are independent, as they correspond to
  triangles which do not share any edge and edges are sampled independent of
  each other. Let $Y_c$ be the sum of the r.v.'s in $\mathcal{X}_c$. The
  r.v.~$Y_c$ has a binomial distribution with parameters $|\mathcal{X}_c|$ and
  $p_t^3$. By the Chernoff bound for binomial r.v.'s, we have that
  \begin{align*}
    \Pr\left(|p^{-3}Y_c - |\mathcal{X}_c||>\varepsilon
    |\mathcal{X}_c|\right)&<
      2\exp\left(-\varepsilon^2p^3|\mathcal{X}_c|/2\right)\\
      &<2\exp\left(-\varepsilon^2p^3L/2\right)\\
      &\le\frac{\delta}{3h^{(t)}+1},
    \end{align*}
    where the last step comes from the requirement
    in~\eqref{eq:requirementp}.Then by applying the union bound over all the (at
    most) $3h^{(t)}+1$ colors we get
    \[
      \Pr(\exists \mbox{ color } c \mbox{ s.t. } |p^{-3}Y_c -
      |\mathcal{X}_c||>\varepsilon |\mathcal{X}_c| ) < \delta\enspace.
    \]
    Since $\phi_\textsc{in}{(t)}=p^{-3}\displaystyle\sum_{\mbox{\tiny color } c} Y_c$,
    from the above equation we have that, with probability at least $1-\delta$,
    \begin{align*}
      |\phi_\textsc{in}^{(t)}-|\Delta^{(t)}||&\le \left|\sum_{\mbox{\tiny color } c}
      p^{-3}Y_c-\sum_{\mbox{\tiny color } c} |\mathcal{X}_c|\right| \\
      &\le \sum_{\mbox{\tiny color } c} |p^{-3}Y_c - |\mathcal{X}_c||
      \le \sum_{\mbox{\tiny color } c}\varepsilon|\mathcal{X}_c|\le\varepsilon
      |\Delta^{(t)}|\enspace.
    \end{align*}
\end{proof}

The above result is of independent interest and can be used, for example, to
give concentration bounds to the estimation computed by
\textsc{mascot-c}~\citep{LimK15}.

We remark that we can not use the same approach from
Lemma~\ref{lem:concentrationindependent} to show a concentration result for
\algobase because it uses reservoir sampling, hence the event of having a
triangle $a$ in $\Sam$ and the event of having another triangle $b$ in $\Sam$
are not independent.

We can however show the following general result, similar in spirit to the
well-know Poisson approximation of balls-and-bins
processes~\citep{mitzenmacher2005probability}. Its proof can be found in
App.~\ref{app:algobase}.

Fix the parameter $M$ and a time $t>M$. Let $\Sam_\textsc{mix}$ be a sample of
$M$ edges from $E^{(t)}$ obtained through reservoir sampling (as \textsc{mix}
would do), and let $\Sam_\textsc{in}$ be a sample of the edges in $E^{(t)}$
obtained by sampling edges independently with probability $M/t$ (as
\textsc{indep} would do). We remark that the size of $\Sam_\textsc{in}$ is in
$[0,t]$ but not necessarily $M$.

\begin{lemma}\label{lem:equivMp}
  Let $f~:~2^{E^{(t)}}\to \{0,1\}$ be an arbitrary binary function from the
  powerset of $E^{(t)}$ to $\{0,1\}$ . We have
  \[
    \Pr\left(f(\Sam_\textsc{mix}) = 1\right) \le e\sqrt{M}
    \Pr\left(f(\Sam_\textsc{in}) = 1\right)\enspace.
  \]
\end{lemma}

We now use the above two lemmas to show that the estimator
$\phi_\textsc{mix}^{(t)}$ computed by \textsc{mix} is concentrated. We will
first need the following technical fact.

\begin{fact}\label{fact:loglog}
  For any $x\ge 5$, we have
  \[
    \ln\left(x(1+\ln^{2/3}x)\right) \le \ln^2 x\enspace.
  \]
\end{fact}

\begin{lemma}\label{lem:concentrationmix}
  Let $t\ge 0$ and assume $|\Delta^{(t)}|<0$. For any
  $\varepsilon,\delta\in(0,1)$, let
  \[
    \Psi = 2\varepsilon^{-2}\frac{3h^{(t)}+1}{|\Delta^{(t)}|} \ln
    \left(e\frac{3h^{(t)}+1}{\delta}\right) \enspace.
  \]
  If
  \[
    M \ge\max\left\{ t \sqrt[3]{\Psi} \left(1 + \frac{1}{2}\ln^{2/3} \left (t \sqrt[3]{\Psi}
    \right)\right), 25\right\}
  \]
  then
  \[
    \Pr\left(|\phi_\textsc{mix}^{(t)}-|\Delta^{(t)}||<\varepsilon|\Delta^{(t)}|\right)\ge
    1-\delta\enspace.
  \]
\end{lemma}
\begin{proof}
  For any $S\subseteq E^{(t)}$ let $\tau(S)$ be the number of triangles in $S$,
  i.e., the number of triplets of edges in $S$ that compose a triangle in
  $G^{(t)}$. Define the function $g ~:~ 2^{E^{(t)}}\to \mathbb{R}$ as
  \[
    g(S) =\left(\frac{t}{M}\right)^3 \tau(S)\enspace.
  \]
  Assume that we run \textsc{indep} with $p=M/t$, and let
  $\Sam_\textsc{in}\subseteq E^{(t)}$ be the sample built by \textsc{indep}
  (through independent sampling with fixed probability $p$). Assume also that
  we run \textsc{mix} with parameter $M$, and let $\Sam_\textsc{mix}$ be the
  sample built by \textsc{mix} (through reservoir sampling with a reservoir of
  size $M$). We have that $\phi_\textsc{in}^{(t)}=g(\Sam_\textsc{in})$ and
  $\phi_\textsc{mix}^{(t)}=g(\Sam_\textsc{mix})$. Define now the binary
  function $f ~:~ 2^{E^{(t)}}\to \{0,1\}$ as
  \[
    f(S)=\left\{\begin{array}{cl}1 & \text{if }
  |g(S)-|\Delta^{(t)}||>\varepsilon|\Delta^{(t)}|\\ 0 & \text{otherwise} \end{array}\right.\enspace.
  \]
  We now show that, for $M$ as in the hypothesis, we have
  \begin{equation}\label{eq:whattoshow}
    p \ge \sqrt[3]{2\varepsilon^{-2}\frac{3h^{(t)}+1}{|\Delta^{(t)}|} \ln \left(e\sqrt{M}\frac{3h^{(t)}+1}{\delta}\right)}\enspace.
  \end{equation}
  Assume for now that the above is true. From this, using
  Lemma~\ref{lem:concentrationindependent} and the above fact about $g$ we get
  that
  \[
    \Pr\left(|\phi_\textsc{in}^{(t)}-|\Delta^{(t)}||>\varepsilon|\Delta^{(t)}|\right)
    = \Pr\left(f(\Sam_\textsc{in})=1\right) < \frac{\delta}{e\sqrt{M}}\enspace.
  \]
  From this and Lemma~\ref{lem:equivMp}, we get that
  \[
    \Pr\left(f(\Sam_\textsc{mix})=1\right) \le \delta
  \]
  which, from the definition of $f$ and the properties of $g$, is equivalent to
  \[
    \Pr\left(|\phi_\textsc{mix}^{(t)}-|\Delta^{(t)}||>\varepsilon|\Delta^{(t)}|\right) \le \delta
  \]
  and the proof is complete. All that is left is to show
  that~\eqref{eq:whattoshow} holds for $M$ as in the hypothesis.

  Since $p=M/t$, we have that~\eqref{eq:whattoshow} holds for
  \begin{align}
    M^3 &\ge t^3 2\varepsilon^{-2}\frac{3h^{(t)}+1}{|\Delta^{(t)}|}  \ln
    \left(\sqrt{M}e\frac{3h^{(t)}+1}{\delta}\right)\nonumber\\
    & =  t^3 2\varepsilon^{-2}\frac{3h^{(t)}+1}{|\Delta^{(t)}|}  \left (\ln
    \left(e\frac{3h^{(t)}+1}{\delta}\right) + \frac{1}{2}\ln M
  \right)\enspace.\label{eq:M}
\end{align}
We now show that~\eqref{eq:M} holds.

Let $A=t\sqrt[3]{\Psi}$ and let $B=t\sqrt[3]{\Psi}\ln^{2/3} \left (t
\sqrt[3]{\Psi} \right)$. We now show that $A^3+B^3$ is greater or equal to the
r.h.s.~of~\eqref{eq:M}, hence $M^3=(A+B)^3> A^3 + B^3$ must also be greater or
equal to the r.h.s.~of~\eqref{eq:M}, i.e., \eqref{eq:M} holds. This really
reduces to show that
\begin{equation}\label{eq:bcube}
  B^3\ge t^3 2\varepsilon^{-2}\frac{3h^{(t)}+1}{|\Delta^{(t)}|}\frac{1}{2}\ln M
\end{equation}
as the r.h.s.of~\eqref{eq:M} can be written as
\[
  A^3+ t^3
  2\varepsilon^{-2}\frac{3h^{(t)}+1}{|\Delta^{(t)}|}\frac{1}{2}\ln
  M\enspace.
\]
We actually show that
\begin{equation}\label{eq:bcubesecond}
  B^3 \ge t^3\Psi\frac{1}{2}\ln M
\end{equation}
which implies~\eqref{eq:bcube} which, as discussed, in turn
implies~\eqref{eq:M}. Consider the ratio
\begin{equation}\label{eq:ratio}
  \frac{B^3}{t^3\Psi\frac{1}{2}\ln M} =
  \frac{\frac{1}{2}t^3\Psi\ln^2(t\sqrt[3]{\Psi})}{t^3\Psi\frac{1}{2}\ln M} =
  \frac{\ln^2(t\sqrt[3]{\Psi})}{\ln M} \ge
  \frac{\ln^2(t\sqrt[3]{\Psi})}{\ln\left(t \sqrt[3]{\Psi}
  \left(1 + \ln^{2/3} \left (t \sqrt[3]{\Psi} \right)\right)\right)
  }\enspace.
\end{equation}
We now show that  $t\sqrt[3]{\Psi} \ge 5$. By the assumptions $t> M \ge 25$ and by
\[
  t \sqrt[3]{\Psi}\ge \frac{t}{\sqrt[3]{|\Delta^{(t)}|}} \ge \sqrt{t}
\]
which holds because $|\Delta^{(t)}| \le t^{3/2}$ (in a graph with $t$ edges there can
not be more than $t^{3/2}$ triangles) we have that $t\sqrt[3]{\Psi} \ge 5$.
Hence Fact~\ref{fact:loglog} holds and we can write, from~\eqref{eq:ratio}:
\[
  \frac{\ln^2(t\sqrt[3]{\Psi})}{\ln\left(t \sqrt[3]{\Psi}
  \left(1 + \ln^{2/3} \left (t \sqrt[3]{\Psi} \right)\right)\right)}\ge
  \frac{\ln^2(t\sqrt[3]{\Psi})}{\ln^2\left(t
  \sqrt[3]{\Psi}\right)}\ge 1,
\]
which proves~\eqref{eq:bcubesecond}, and in cascade~\eqref{eq:bcube},
\eqref{eq:M}, \eqref{eq:whattoshow}, and the thesis.
\end{proof}

\paragraph{Relationship between \algobase and \textsc{mix}} When both \algobase
and \textsc{mix} use a sample of size $M$, their respective estimators
$\phi^{(t)}$ and $\phi_\textsc{mix}^{(t}$ are related as discussed in the
following result, whose straightforward proof is deferred to App.~\ref{app:algobase}.

\begin{lemma}\label{lem:estimatorsrelationship}
  For any $t>M$ we have
  \[
    \left|\phi^{(t)}-\phi_\textsc{mix}^{(t)}\right|\le
    \phi_\textsc{mix}^{(t)}\frac{4}{M-2}\enspace.
  \]
\end{lemma}

\paragraph{Tying everything together} Finally we can use the previous lemmas
to prove our concentration result for \algobase.

\begin{proof}[of Thm.~\ref{thm:baseconcentration}]
  For $M$ as in the hypothesis we have, from Lemma~\ref{lem:concentrationmix},
  that
  \[
    \Pr\left(\phi_\textsc{mix}^{(t)}\le
    (1+\varepsilon/2)|\Delta^{(t)}|\right)\ge
    1-\delta\enspace.
  \]
  Suppose the event $\phi_\textsc{mix}^{(t)}\le
  (1+\varepsilon/2)|\Delta^{(t)}|$ (i.e., $|\phi_\textsc{mix}^{(t)} -
  |\Delta^{(t)}| |\le \varepsilon|\Delta^{(t)}|/2$) is indeed verified. From this and
  Lemma~\ref{lem:estimatorsrelationship} we have
  \[
    |\phi^{(t)}-\phi_\textsc{mix}^{(t)}|\le
    \left(1+\frac{\varepsilon}{2}\right)|\Delta^{(t)}|\frac{4}{M-2}\le
    |\Delta^{(t)}|\frac{6}{M-2},
  \]
  where the last inequality follows from the fact that $\varepsilon<1$. Hence,
  given that $M\ge 12\varepsilon^{-1}+e^2\ge 12\varepsilon^{-1}+2$, we have
  \[
    |\phi^{(t)}-\phi_\textsc{mix}^{(t)}|\le
    |\Delta^{(t)}|\frac{\varepsilon}{2}\enspace.
  \]
  Using the above, we can then write:
  \begin{align*}
    |\phi^{(t)} - |\Delta^{(t)}| | &=
    |\phi^{(t)} - \phi_\textsc{mix}^{(t)} +  \phi_\textsc{mix}^{(t)} - |\Delta^{(t)}| | \\
    &\le |\phi^{(t)} - \phi_\textsc{mix}^{(t)}| +  |\phi_\textsc{mix}^{(t)} - |\Delta^{(t)}| | \\
    &\le \frac{\varepsilon}{2} |\Delta^{(t)}|  +
    \frac{\varepsilon}{2}|\Delta^{(t)}|  = \varepsilon |\Delta^{(t)}|
  \end{align*}
  which completes the proof.
\end{proof}

\subsubsection{Comparison with fixed-probability
approaches}\label{sec:comparison}
We now compare the variance of \algobase to the variance of the fixed
probability sampling approach \textsc{mascot-c}~\citep{LimK15}, which samples
edges \emph{independently} with a fixed probability $p$ and uses
$p^{-3}|\Delta_\Sam|$ as the estimate for the global number of triangles at time
$t$. As shown by~\citet[Lemma 2]{LimK15}, the variance of this estimator is
\[
  \variance[p^{-3}|\Delta_\Sam|] = |\Delta^{(t)}|\bar{f}(p) + r^{(t)}
  \bar{g}(p),
\]
where $\bar{f}(p) = p^{-3} - 1$ and $\bar{g}(p) = p^{-1} - 1$.

Assume that we give \textsc{mascot-c} the additional information that the stream
has finite length $T$, and assume we run \textsc{mascot-c} with $p=M/T$
so that the expected sample size at the end of the stream is
$M$.\footnote{We are giving \textsc{mascot-c} a significant advantage: if
only space $M$ were available, we should run \textsc{mascot-c} with a sufficiently
smaller $p'<p$, otherwise there would be a constant probability that
\textsc{mascot-c} would run out of memory.} Let $\mathbb{V}^{(t)}_\text{fix}$ be
the resulting variance of the \textsc{mascot-c} estimator at time $t$, and let
$\mathbb{V}^{(t)}$ be the variance of our estimator at time $t$
(see~\eqref{eq:globvariancebase}).  For $t\le M$, $\mathbb{V}^{(t)}=0$, hence
$\mathbb{V}^{(t)}\le \mathbb{V}^{(t)}_\text{fix}$.

For $M< t<T$, we can show the following result. Its proof is more tedious than
interesting so we defer it to App.~\ref{app:algobase}.
\begin{lemma}\label{lem:variancecomparison}
  Let $0<\alpha<1$ be a constant. For any constant
  $M>\max(\frac{8\alpha}{1-\alpha}, 42)$ and any $t \le \alpha T$ we have
  $\mathbb{V}^{(t)} < \mathbb{V}^{(t)}_\mathrm{fix}$.
\end{lemma}
For example, if we set $\alpha = 0.99$ and run \algobase with $M\ge 400$ and
\textsc{mascot-c} with $p=M / T$, we have that \algobase has strictly
smaller variance than \textsc{mascot-c} for $99\%$ of the stream.

What about $t=T$? The difference between the definitions of
$\mathbb{V}^{(t)}_\text{fix}$ and $\mathbb{V}^{(t)}$ is in the presence of
$\bar{f}(M/T)$ instead of $f(t)$ (resp.~$\bar{g}(M/T)$ instead of $g(t)$) as
well as the additional term $w^{(t)}h(M,t)\le 0$ in our $\mathbb{V}^{(t)}$. Let
$M(T)$ be an arbitrary slowly increasing function of $T$. For $T \to \infty$ we
can show that $\lim_{T \to \infty}\frac{\bar{f}(M(T)/T)}{f(T)} = \lim_{T \to
\infty}\frac{\bar{g}(M(T)/T)}{g(T)} = 1$, hence, \emph{informally},
$\mathbb{V}^{(T)}\to\mathbb{V}^{(T)}_\text{fix}$, for $T\to\infty$.

A similar discussion also holds for the method we present in
Sect.~\ref{sec:improved}, and explains the results of our experimental
evaluations, which shows that our algorithms have strictly lower (empirical)
variance than fixed probability approaches for most of the stream.

\subsubsection{Update time}
The time to process an element of the stream is dominated by
the computation of the shared neighborhood $\mathcal{N}_{u,v}$ in
\textsc{UpdateCounters}. A \textsc{Mergesort}-based algorithm for the
intersection requires $O\left(\deg(u) + \deg(v)\right)$ time, where the degrees
are w.r.t.~the graph $G_\Sam$. By storing the neighborhood of each vertex in a
Hash Table (resp.~an AVL tree), the update time can be reduced to
$O(\min\{\deg(v),\deg(u)\})$ (resp.~amortized time
$O(\min\{\deg(v),\deg(u)\}+\log\max\{\deg(v),\deg(u)\})$).

\subsection{Improved insertion algorithm -- \algoimproved}\label{sec:improved}
\algoimproved is a variant of \algobase with small modifications that result in
higher-quality (i.e., lower variance) estimations. The changes are:
\begin{longenum}
  \item \textsc{UpdateCounters} is called \emph{unconditionally for each element
    on the stream}, before the algorithm decides whether or not to insert the
    edge into $\Sam$. W.r.t.~the pseudocode in Alg.~\ref{alg:triest-base}, this
    change corresponds to moving the call to \textsc{UpdateCounters} on line 6
    to \emph{before} the {\bf if} block. \textsc{mascot}~\citep{LimK15} uses a
    similar idea, but \algoimproved is significantly different as \algoimproved
    allows edges to be removed from the sample, since it uses reservoir sampling.
  \item \algoimproved \emph{never} decrements the counters when an edge is
    removed from $\Sam$. W.r.t.~the pseudocode in Alg.~\ref{alg:triest-base}, we
    remove the call to \textsc{UpdateCounters} on line 13.
  \item \textsc{UpdateCounters} performs a \emph{weighted}
    increase of the counters using $\eta^{(t)}=\max\{1,(t-1)(t-2)/(M(M-1))\}$ as
    weight. W.r.t.~the pseudocode in Alg.~\ref{alg:triest-base}, we replace
    ``$1$'' with $\eta^{(t)}$ on lines 19--22 (given change 2 above, all the
    calls to \textsc{UpdateCounters} have $\bullet=+$).
\end{longenum}
The resulting pseudocode for \algoimproved is presented in
Alg.~\ref{alg:triest-impr}.

\begin{algorithm}[ht]
  \small 
  \caption{\algoimproved}
  \label{alg:triest-impr}
  \begin{algorithmic}[1]
    \Statex{\textbf{Input:} Insertion-only edge stream $\Sigma$, integer $M\ge6$}
    \State $\Sam\leftarrow\emptyset$, $t \leftarrow 0$, $\tau \leftarrow 0$
    \For{ {\bf each} element $(+,(u,v))$ from $\Sigma$}
    \State $t\leftarrow t +1$
    \State \textsc{UpdateCounters}$(u,v)$
    \If{\textsc{SampleEdge}$((u,v), t )$}
      \State $\Sam \leftarrow \Sam\cup \{(u,v)\}$
    \EndIf
    \EndFor
    \Statex
    \Function{\textsc{SampleEdge}}{$(u,v),t$}
      \If {$t\leq  M$}
        \State \textbf{return} True
      \ElsIf{\textsc{FlipBiasedCoin}$(\frac{M}{t}) = $ heads}
        \State $(u',v') \leftarrow$ random edge from $\Sam$
        \State $\Sam\leftarrow \Sam\setminus \{(u',v')\}$
        \State \textbf{return} True
      \EndIf
      \State \textbf{return} False
    \EndFunction
    \Statex
    \Function{UpdateCounters}{$u,v$}
      \State $\mathcal{N}^\Sam_{u,v} \leftarrow \mathcal{N}^\Sam_u \cap \mathcal{N}^\Sam_v$
      \State $\eta=\max\{1,(t-1)(t-2)/(M(M-1))\}$
      \ForAll {$c \in \mathcal{N}^\Sam_{u,v}$}
        \State $\tau \leftarrow \tau + \eta$
        \State $\tau_c \leftarrow \tau_c + \eta$
        \State $\tau_u \leftarrow \tau_u + \eta$
        \State $\tau_v \leftarrow \tau_v + \eta$
      \EndFor
    \EndFunction
  \end{algorithmic}
\end{algorithm}

\paragraph{Counters} If we are interested only in estimating the global number of
triangles in $G^{(t)}$, \algoimproved needs to maintain only the counter $\tau$
and the edge sample $\Sam$ of size $M$, so it still requires space $O(M)$. If
instead we are interested in estimating the local triangle counts, at any time $t$
\algoimproved maintains (non-zero) local counters \emph{only} for the nodes $u$
such that at least one triangle with a corner $u$  has been detected by the
algorithm up until time $t$. The number of such nodes may be greater than
$O(M)$, but this is the price to pay to obtain estimations with lower variance
(Thm.~\ref{thm:improvedvariance}).

\subsubsection{Estimation}
When queried for an estimation, \algoimproved returns the value of the
corresponding counter, unmodified.

\subsubsection{Analysis}
We now present the analysis of the estimations computed by \algoimproved,
showing results involving their unbiasedness, their variance, and their
concentration around their expectation. Results analogous to those in
Thms.~\ref{thm:improvedunbiased}, \ref{thm:improvedvariance},
and~\ref{thm:improvedconcentration} hold for the local triangle count for any
$u\in V^{(t)}$, replacing the global quantities with the corresponding local
ones.

\subsubsection{Expectation}
As in \algobase, the estimations by \algoimproved are \emph{exact} at time
$t\le M$ and unbiased for $t>M$. The proof of the following theorem follows the
same steps as the one for Thm~\ref{thm:baseunbiased}, and can be found in
App.~\ref{app:algoimproved}.
\begin{theorem}\label{thm:improvedunbiased}
  We have $\tau^{(t)}=|\Delta^{(t)}|$ if $t\le M$ and
  $\mathbb{E}\left[\tau^{(t)}\right]$ $=|\Delta^{(t)}|$ if $t> M$.
\end{theorem}

\subsubsection{Variance}
We now show an \emph{upper bound} to the variance of the \algoimproved
estimations for $t>M$. The proof relies on a very careful analysis of the
covariance of two triangles which depends on the order of arrival of the edges
in the stream (which we assume to be adversarial). For any
$\lambda\in\Delta^{(t)}$ we denote as $t_\lambda$ the time at which the last
edge of $\lambda$ is observed on the stream. Let $z^{(t)}$ be the number of
unordered pairs $(\lambda,\gamma)$ of distinct triangles in $G^{(t)}$ that share
an edge $g$ and are such that:
\begin{enumerate}
  \item $g$ is neither the last edge of $\lambda$ nor $\gamma$ on the stream;
    and
  \item $\min\{t_\lambda,t_\gamma\}>M+1$.
\end{enumerate}

\begin{theorem}\label{thm:improvedvariance}
  Then, for any time $t>M$, we have
  \begin{equation}\label{eq:improvedvariance}
    \variance\left[\tau^{(t)}\right]\le
    |\Delta^{(t)}|(\eta^{(t)}-1)+z^{(t)}\frac{t-1-M}{M}\enspace.
  \end{equation}
\end{theorem}
The bound to the variance presented in~\eqref{eq:improvedvariance} is extremely
pessimistic and loose. Specifically, it does not contain the \emph{negative}
contribution to the variance given by the $\binom{|\Delta^{(t)}|}{2}-z^{(t)}$
triangles that do not satisfy the requirements in the definition of $z^{(t)}$.
Among these pairs there are, for example, all pairs of triangles not sharing any
edge, but also many pairs of triangles that share an edge, as the definition of
$z^{(t)}$ consider only a subsets of these. All these pairs would give a
negative contribution to the variance, i.e., decrease the
r.h.s.~of~\eqref{eq:improvedvariance}, whose more correct form would be
\[
  |\Delta^{(t)}|(\eta^{(t)}-1)+z^{(t)}\frac{t-1-M}{M}+\left(\binom{|\Delta^{(t)}|}{2}-z^{(t)}\right)\omega_{M,t}
\]
where $\omega_{M,t}$ is (an upper bound to) the minimum \emph{negative}
contribution of a pair of triangles that do not satisfy the requirements in the
definition of $z^{(t)}$. Sadly, computing informative upper bounds to
$\omega_{M,t}$ is not straightforward, even in the restricted setting where only
pairs of triangles not sharing any edge are considered.

To prove Thm.~\ref{thm:improvedvariance} we first need
Lemma~\ref{lem:negativedependence}, whose proof is deferred to
App.~\ref{app:algoimproved}.

For any time step $t$ and any edge $e\in E^{(t)}$, we denote with $t_e$ the time
step at which $e$ is on the stream. For any $\lambda\in\Delta^{(t)}$, let $\lambda=(\ell_1,\ell_2,\ell_3)$, where the edges are numbered in
order of appearance on the stream. We define the event $D_\lambda$ as the event
that $\ell_1$ and $\ell_2$ are both in the edge sample $\Sam$ at the end of time
step $t_\lambda-1$.

\begin{lemma}\label{lem:negativedependence}
  Let $\lambda=(\ell_1,\ell_2,\ell_3)$ and $\gamma=(g_1,g_2,g_3)$ be two
  \emph{disjoint} triangles, where the edges are numbered in order of
  appearance on the stream, and assume, w.l.o.g., that the last edge of
  $\lambda$ is on the stream \emph{before} the last edge of $\gamma$. Then
  \[
    \Pr(D_\gamma~|~D_\lambda) \le \Pr(D_\gamma)\enspace.
  \]
\end{lemma}

We can now prove Thm.~\ref{thm:improvedvariance}.

\begin{proof}[of Thm.~\ref{thm:improvedvariance}]
  Assume $|\Delta^{(t)}|>0$, otherwise \algoimproved estimation is
  deterministically correct and has variance 0 and the thesis holds. Let
  $\lambda\in\Delta^{(t)}$ and let $\delta_\lambda$ be a random variable that
  takes value $\xi_{2,t_\lambda-1}$ if both $\ell_1$ and $\ell_2$ are in $\Sam$
  at the end of time step $t_\lambda-1$, and $0$ otherwise. Since
  \[
    \variance\left[\delta_{\lambda}\right] = \xi_{2,t_\lambda-1}-1\le
    \xi_{2,t-1},
  \]
  we have:
  \begin{align}
    \variance\left[\tau^{(t)}\right] &=
    \variance\left[\sum_{\lambda\in\Delta^{(t)}} \delta_{\lambda}\right]=
    \sum_{\lambda\in \Delta^{(t)}} \sum_{\gamma \in
    \Delta^{(t)}}\text{Cov}\left[\delta_{\lambda},\delta_{\gamma}\right]
    \nonumber\\
    &= \sum_{\lambda\in \Delta^{(t)}}
    \variance\left[\delta_{\lambda}\right] + \sum_{\substack{\lambda,
    \gamma\in \Delta^{(t)}\\\lambda \neq \gamma}}
    \text{Cov}\left[\delta_{\lambda},\delta_{\gamma}\right]\nonumber\\
    &\le |\Delta^{(t)}|(\xi_{2,t-1}-1) + \sum_{\substack{\lambda,
    \gamma\in \Delta^{(t)}\\\lambda \neq \gamma}}\left(
    \mathbb{E}[\delta_\lambda\delta_\gamma]-\mathbb{E}[\delta_\lambda]\mathbb{E}[\delta_\gamma]\right)\nonumber\\
    &\le |\Delta^{(t)}|(\xi_{2,t-1}-1) + \sum_{\substack{\lambda,
    \gamma\in \Delta^{(t)}\\\lambda \neq \gamma}}
    \left(\mathbb{E}[\delta_\lambda\delta_\gamma] - 1\right) \enspace.\label{eq:variance2improved}
  \end{align}

  For any $\lambda\in\Delta^{(t)}$ define $q_\lambda=\xi_{2,t_\lambda-1}$.
  Assume now $|\Delta^{(t)}|\ge 2$, otherwise we have $r^{(t)}=w^{(t)}=0$ and
  the thesis holds as the second term on the
  r.h.s.~of~\eqref{eq:variance2improved} is 0. Let now $\lambda$ and $\gamma$
  be two distinct triangles in $\Delta^{(t)}$
  (hence $t\ge 5$). We have
  \begin{equation*}
    \mathbb{E}\left[\delta_{\lambda}\delta_{\gamma}\right]=
    q_\lambda q_\gamma\Pr\left(\delta_\lambda\delta_\gamma=q_\lambda q_\gamma\right)
  \end{equation*}
  The event ``$\delta_\lambda\delta_\gamma=q_\lambda q_\gamma$'' is the
  intersection of events $D_\lambda\cap D_\gamma$, where $D_\lambda$ is the
  event that the first two edges of $\lambda$ are in $\Sam$ at the end of time
  step $t_\lambda-1$, and similarly for $D_\gamma$. We now look at
  $\Pr(D_\lambda\cap D_\gamma)$ in the various possible cases.

  Assume that $\lambda$ and $\gamma$ do not share any edge, and, w.l.o.g.,
  that the third (and last) edge of $\lambda$ appears on the stream before the
  third (and last) edge of $\gamma$, i.e., $t_\lambda< t_\gamma$. From
  Lemma~\ref{lem:negativedependence} and Lemma~\ref{lem:reservoirhighorder} we
  then have
  \[
    \Pr(D_\lambda\cap D_\gamma)=\Pr(D_\gamma|D_\lambda)\Pr(D_\lambda)\le
    \Pr(D_\gamma)\Pr(D_\lambda)\le \frac{1}{q_\lambda q_\gamma}\enspace.
  \]

  Consider now the case where $\lambda$ and $\gamma$ share an edge $g$.
  W.l.o.g., let us assume that $t_\lambda\le t_\gamma$ (since the shared edge
  may be the last on the stream both for $\lambda$ and for $\gamma$, we may
  have $t_\lambda=t_\gamma$). There are the following possible sub-cases :
  \begin{description}
    \item[$g$ is the last on the stream among all the edges of $\lambda$ and
      $\gamma$] In this case we have $t_\lambda=t_\gamma$. The event
      ``$D_\lambda\cap D_\gamma$'' happens if and only if the \emph{four}
      edges that, together with $g$, compose $\lambda$ and $\gamma$ are
      all in $\Sam$ at the end of time step $t_\lambda-1$. Then, from
      Lemma~\ref{lem:reservoirhighorder} we have
      \[
        \Pr(D_\lambda\cap D_\gamma)= \frac{1}{\xi_{4,t_\lambda-1}}
        \le
        \frac{1}{q_\lambda}\frac{(M-2)(M-3)}{(t_\lambda-3)(t_\lambda-4)}\le\frac{1}{q_\lambda}\frac{M(M-1)}{(t_\lambda-1)(t_\lambda-2)}\le
        \frac{1}{q_\lambda q_\gamma}\enspace.
      \]
    \item[$g$ is the last on the stream among all the edges of $\lambda$ and
      the first among all the edges of $\gamma$] In this case, we
      have  that $D_\lambda$ and $D_\gamma$ are independent. Indeed
      the fact that the first two edges of $\lambda$ (neither of which
      is $g$) are in $\Sam$ when $g$ arrives on the stream has no
      influence on the probability that $g$ and the second edge of
      $\gamma$ are inserted in $\Sam$ and are not evicted until
      the third edge of $\gamma$ is on the stream. Hence we have
      \[
        \Pr(D_\lambda\cap D_\gamma)=
        \Pr(D_\gamma)\Pr(D_\lambda)=\frac{1}{q_\lambda q_\gamma}\enspace.
      \]
    \item[$g$ is the last on the stream among all the edges of $\lambda$ and
      the second among all the edges of $\gamma$]
      In this case we can follow an approach similar to the one in the
      proof for Lemma~\ref{lem:negativedependence} and have that
      \[
        \Pr(D_\lambda\cap D_\gamma) \le
        \Pr(D_\gamma)\Pr(D_\lambda)\le \frac{1}{q_\lambda q_\gamma}\enspace.
      \]
      The intuition behind this is that if the first two edges of
      $\lambda$ are in $\Sam$ when $g$ is on the stream, their presence
      lowers the probability that the first edge of $\gamma$ is in $\Sam$
      at the same time, and hence that the first edge of $\gamma$ and $g$
      are in $\Sam$ when the last edge of $\gamma$ is on the stream.
    \item[$g$ is not the last on the stream among all the edges of
      $\lambda$]
      There are two situations to consider, or better, one situation and
      all other possibilities. The situation we consider is that
      \begin{enumerate}
        \item $g$ is the first edge of $\gamma$ on the stream; and
        \item the second edge of $\gamma$ to be on the stream is on the
          stream at time $t_2>t_\lambda$.
      \end{enumerate}
      Suppose this is the case. Recall that if $D_\lambda$ is verified,
      than we know that $g$ is in $\Sam$ at the
      beginning of time step $t_\lambda$. Define the following events:
      \begin{itemize}
        \item $E_1$: ``the set of edges evicted from $\Sam$ between the
          beginning of time step $t_\lambda$ and the beginning of time step
          $t_2$ does not contain $g$.''
        \item $E_2$: ``the second edge of $\gamma$, which is on the stream at
          time $t_2$, is inserted in $\Sam$ and the edge that is evicted is not
          $g$.''
        \item $E_3$: ``the set of edges evicted from $\Sam$ between the
          beginning of time step $t_2+1$ and the beginning of time
          step $t_\gamma$ does not contain either $g$ or the second edge of
          $\gamma$.''
      \end{itemize}
      We can then write
      \[
        \Pr(D_\gamma~|~D_\lambda)= \Pr(E_1~|~D_\lambda)
        \Pr(E_2~|~E_1\cap D_\lambda)
        \Pr(E_3~|~E_2\cap E_1\cap D_\lambda)\enspace.
      \]
      We now compute the probabilities on the r.h.s., where we denote with
      $\mathds{1}_{t_2>M}(1)$ the function that has value $1$ if $t_2>M$,
      and value $0$ otherwise:
      \begin{align*}
        \Pr(E_1~|~D_\lambda) &=
        \prod_{j=\max\{t_\lambda,M+1\}}^{t_2-1}\left(\left(1-\frac{M}{j}\right)+\frac{M}{j}\left(\frac{M-1}{M}\right)\right)\\
        &=\prod_{j=\max\{t_\lambda,M+1\}}^{t_2-1}\frac{j-1}{j}=
        \frac{\max\{t_\lambda-1, M\}}{\max\{M,t_2-1\}}\enspace;\\
        \Pr(E_2~|~E_1\cap D_\lambda) &=
        \frac{M}{\max\{t_2,M\}}\frac{M-\mathds{1}_{t_2>M}(1)}{M}=
        \frac{M-\mathds{1}_{t_2>M}(1)}{\max\{t_2,M\}}\enspace;\\
        \Pr(E_3~|~E_2\cap E_1\cap D_\lambda)&=
        \prod_{j=\max\{t_2+1,M+1\}}^{t_\gamma-1}\left(\left(1-\frac{M}{j}\right)+\frac{M}{j}\left(\frac{M-2}{M}\right)\right)\\
        &= \prod_{j=\max\{t_2+1,M+1\}}^{t_\gamma-1}\frac{j-2}{j}
        =\frac{\max\{t_2,M\}\max\{t_2-1,M-1\}}{\max\{t_\gamma-2,M-1\}\max\{t_\gamma-1,M\}}\enspace.
      \end{align*}
      Hence, we have
      \[
        \Pr(D_\gamma~|~D_\lambda) =
        \frac{\max\{t_\lambda-1,M\} (M-\mathds{1}_{t_2>M}(1))
        \max\{t_2-1,M-1\}}{ \max\{M,t_2-1\}
        \max\{(t_\gamma-2)(t_\gamma-1),M(M-1)\}}\enspace.
      \]
      With a (somewhat tedious) case analysis we can verify that
      \[
        \Pr(D_\gamma~|~D_\lambda)\le
        \frac{1}{q_\gamma}\frac{\max\{M,t_\lambda-1\}}{M}\enspace.
      \]

      Consider now the complement of the situation we just analyzed. In
      this case, two edges of $\gamma$, that is, $g$ and another edge $h$ are on
      the stream before time $t_\lambda$, in some non-relevant order (i.e., $g$
      could be the first or the second edge of $\gamma$ on the stream). Define
      now the following events:
      \begin{itemize}
        \item $E_1$: ``$h$ and $g$ are both in $\Sam$ at the beginning
          of time step $t_\lambda$.''
        \item $E_2$: ``the set of edges evicted from $\Sam$ between the
          beginning of time step $t_\lambda$ and the beginning of time
          step $t_\gamma$ does not contain either $g$ or $h$.''
      \end{itemize}
      We can then write
      \[
        \Pr(D_\gamma~|~D_\lambda) =
        \Pr(E_1~|~D_\lambda)\Pr(E_2~|~E_1\cap D_\lambda)\enspace.
      \]
      If $t_\lambda \le M+1$, we have that $\Pr(E_1~|~D_\lambda)=1$.
      Consider instead the case $t_\lambda > M+1$. If $D_\lambda$ is
      verified, then both $g$ and the other edge of $\lambda$ are in
      $\Sam$ at the beginning of time step $t_\lambda$. At this time,
      all subsets of $E^{(t_\lambda-1)}$ of size $M$ and containing both
      $g$ and the other edge of $\lambda$ have an equal probability of
      being $\Sam$, from Lemma~\ref{lem:reservoir}. There are
      $\binom{t_\lambda-3}{M-2}$ such sets. Among these, there are
      $\binom{t_\lambda-4}{M-3}$ sets that also contain $h$. Therefore, if
      $t_\lambda > M+1$, we have
      \[
        \Pr(E_1~|~D_\lambda)=\frac{\binom{t_\lambda-4}{M-3}}{\binom{t_\lambda-3}{M-2}}=\frac{M-2}{t_\lambda-3}\enspace.
      \]
      Considering what we said before for the case $t_\lambda\le M+1$, we
      then have
      \[
        \Pr(E_1~|~D_\lambda)=\min\left\{1,\frac{M-2}{t_\lambda-3}\right\}\enspace.
      \]
      We also have
      \[
        \Pr(E_2~|~E_1\cap D_\lambda)=
        \prod_{j=\max\{t_\lambda,M+1\}}^{t_\gamma-1}\frac{j-2}{j}=
        \frac{\max\{(t_\lambda-2)(t_\lambda-1),
        M(M-1)\}}{\max\{(t_\gamma-2)(t_\gamma-1), M(M-1)\}}\enspace.
      \]
      Therefore,
      \[
        \Pr(D_\gamma~|~D_\lambda)=
        \min\left\{1,\frac{M-2}{t_\lambda-3}\right\}\frac{\max\{(t_\lambda-2)(t_\lambda-1), M(M-1)\}}{\max\{(t_\gamma-2)(t_\gamma-1), M(M-1)\}}\enspace.
      \]
      With a case analysis, one can show that
      \[
        \Pr(D_\gamma~|~D_\lambda)\le\frac{1}{q_\gamma}\frac{\max\{M,t_\lambda-1\}}{M}\enspace.
      \]
  \end{description}
  To recap we have the following two scenarios when considering two distinct
  triangles $\gamma$ and $\lambda$:
  \begin{longenum}
    \item if $\lambda$ and $\gamma$ share an edge and, assuming
      w.l.o.g.~that the third edge of $\lambda$ is on the stream no later
      than the third edge of $\gamma$, and the shared edge is neither the
      last among all edges of $\lambda$ to appear on the stream nor the
      last among all edges of $\gamma$ to appear on the stream, then we
      have
      \begin{align*}
        \mathrm{Cov}[\delta_\lambda,\delta_\gamma] &\le
        \mathbb{E}[\delta_\lambda\delta_\gamma]-1
        =q_\lambda q_\gamma\Pr(\delta_\lambda\delta_\gamma
        =q_\lambda q_\gamma)-1\\
        &\le q_\lambda q_\gamma \frac{1}{q_\lambda q_\gamma}\frac{\max\{M,t_\lambda-1\}}{M} -1\le
        \frac{\max\{M,t_\lambda-1\}}{M} -1\le \frac{t-1-M}{M};
      \end{align*}where the last inequality follows from the fact that $t_\lambda\le
      t$ and $t-1\ge M$.

      For the pairs $(\lambda,\gamma)$ such that $t_\lambda\le M+1$, we have
      $\max\{M,t_\lambda-1\}/M=1$ and therefore
      $\mathrm{Cov}[\delta_\lambda,\delta_\gamma]\leq 0$. We should therefore
      only consider the pairs for which $t_\lambda> M+1$. Their number is given
      by $z^{(t)}$.
    \item in all other cases, including when $\gamma$ and $\lambda$ do not
      share an edge, we have
      $\mathbb{E}[\delta_\lambda\delta_\gamma]\le 1$, and since
      $\mathbb{E}[\delta_\lambda]\mathbb{E}[\delta_\gamma]=1$, we have
      \[
        \mathrm{Cov}[\delta_\lambda,\delta_\gamma] \le 0\enspace.
      \]
  \end{longenum}
  Hence, we can bound
  \[
    \sum_{\substack{\lambda,\gamma\in\Delta^{(t)}\\\lambda\neq\gamma}}\mathrm{Cov}[\delta_\lambda,\delta_\gamma]\le
    z^{(t)}\frac{t-1-M}{M}
  \]
  and the thesis follows by combining this into~\eqref{eq:variance2improved}.
\end{proof}

\subsubsection{Concentration}
We now show a concentration result on the estimation of \algoimproved, which
relies on Chebyshev's
inequality~\citep[Thm.~3.6]{mitzenmacher2005probability} and
Thm.~\ref{thm:improvedvariance}.
\begin{theorem}\label{thm:improvedconcentration}
  Let $t\ge 0$ and assume $|\Delta^{(t)}| >0$. For any
  $\varepsilon,\delta\in(0,1)$, if
  \[
    M> \max\left\{\sqrt{\frac{2(t-1)(t-2)}{\delta
      \varepsilon^2|\Delta^{(t)}|+2}+\frac{1}{4}}+\frac{1}{2},\frac{2z^{(t)}(t-1)}{\delta
    \varepsilon^2 |\Delta^{(t)}|^2+2z^{(t)}}\right\}
  \]
  then $|\tau^{(t)}-|\Delta^{(t)}||<\varepsilon|\Delta^{(t)}|$ with
  probability $>1-\delta$.
\end{theorem}

\begin{proof}
  By Chebyshev’s inequality it is sufficient to prove that
  $$\frac{\variance[\tau^{(t)}]}{\varepsilon^2|\Delta^{(t)}|^2} <\delta\enspace.$$
  We can write
  $$\frac{\variance[\tau^{(t)}]}{\varepsilon^2|\Delta^{(t)}|^2}\le
  \frac{1}{\varepsilon^2|\Delta^{(t)}|}\left((\eta(t)-1)+z^{(t)}\frac{t-1-M}{M|\Delta^{(t)}|}\right)\enspace.
  $$
  Hence it is sufficient to impose the following two conditions:
  \begin{description}
    \item[Condition 1]
      \begin{align}
        \frac{\delta}{2} &> \frac{\eta(t)-1}{\varepsilon^2|\Delta^{(t)}|}\label{eq:cond1}\\
        &>\frac{1}{\varepsilon^2|\Delta^{(t)}|}\frac{(t-1)(t-2)-M(M-1)}{M(M-1)}\nonumber,
      \end{align}
      which is verified for:
      \begin{equation*}
        M(M-1) > \frac{2(t-1)(t-2)}{\delta \varepsilon^2|\Delta^{(t)}|+2}.
      \end{equation*}
      As $t> M$, we have $\frac{2(t-1)(t-2)}{\delta \varepsilon^2|\Delta^{(t)}|+2}>0$.
      The condition in~\eqref{eq:cond1} is thus verified for:
      \begin{equation*}
        M>\frac{1}{2}\left(\sqrt{4\frac{2(t-1)(t-2)}{\delta \varepsilon^2|\Delta^{(t)}|+2}+1}+1\right)
      \end{equation*}
    \item[Condition 2]
      \begin{align*}
        \frac{\delta}{2} &> z^{(t)}\frac{t-1-M}{M\varepsilon^2|\Delta^{(t)}|^2},
      \end{align*}
      which is verified for:
      \begin{equation*}
        M > \frac{2z^{(t)}(t-1)}{\delta \varepsilon^2 |\Delta^{(t)}|^2+2z^{(t)}}.
      \end{equation*}
  \end{description}
  The theorem follows.
\end{proof}

In Thms.~\ref{thm:improvedvariance} and~\ref{thm:improvedconcentration}, it is
possible to replace the value $z^{(t)}$ with the more interpretable $r^{(t)}$,
which is agnostic to the order of the edges on the stream but gives a looser
upper bound to the variance.

\subsection{Fully-dynamic algorithm -- \algofd}\label{sec:fulldyn}
\algofd computes unbiased estimates of the global and local triangle counts in a
\emph{fully-dynamic stream where edges are inserted/deleted in any arbitrary,
adversarial order}. It is based on \emph{random pairing}
(RP)~\citep{GemullaLH08}, a sampling scheme that extends reservoir sampling and
can handle deletions. The idea behind the RP scheme is that edge deletions seen
on the stream will be ``compensated'' by future edge insertions. Following RP,
\algofd keeps a counter $d_\mathrm{i}$ (resp.~$d_\mathrm{o}$) to keep track of
the number of uncompensated edge deletions involving an edge $e$ that was
(resp.~was \emph{not}) in $\Sam$ at the time the deletion for $e$ was on the
stream.

When an edge deletion for an edge $e\in E^{(t-1)}$ is on the stream at the
beginning of time step $t$, then, if $e\in\Sam$ at this time, \algofd
removes $e$ from $\Sam$ (effectively decreasing the number of edges stored in
the sample by one) and increases $d_\mathrm{i}$ by one. Otherwise, it just
increases $d_\mathrm{o}$ by one. When an edge insertion for an edge $e\not\in\
E^{(t-1)}$ is on the stream at the beginning of time step $t$, if
$d_\mathrm{i}+d_\mathrm{o}=0$, then \algofd follows the standard reservoir
sampling scheme. If $|\Sam|<M$, then $e$ is deterministically inserted in
$\Sam$ without removing any edge from $\Sam$ already in $\Sam$, otherwise it
is inserted in $\Sam$ with probability $M/t$, replacing an uniformly-chosen edge
already in $\Sam$. If instead $d_\mathrm{i}+d_\mathrm{o}>0$, then $e$ is
inserted in $\Sam$ with probability $d_\mathrm{i}/(d_\mathrm{i}+d_\mathrm{o})$;
since it must be $d_\mathrm{i}>0$, then it must be $|\Sam|<M$ and no edge
already in $\Sam$ needs to be removed. In any case, after having handled the
eventual insertion of $e$ into $\Sam$, the algorithm decreases $d_\mathrm{i}$ by
$1$ if $e$ was inserted in $\Sam$, otherwise it decreases $d_\mathrm{o}$ by $1$.
\algofd also keeps track of $s^{(t)}=|E^{(t)}|$ by appropriately incrementing or
decrementing a counter by $1$ depending on whether the element on the stream is
an edge insertion or deletion. The pseudocode for \algofd is presented in
Alg.~\ref{algo: TRIEST-FD} where the {\sc UpdateCounters} procedure is the one
from Alg.~\ref{alg:triest-base}.

\begin{algorithm}[t]
  \small 
  \caption{\algofd}
  \label{algo: TRIEST-FD}
  \begin{algorithmic}[1]
    \Statex{\textbf{Input:} Fully-dynamic edge stream $\Sigma$, integer $M\ge 6$}
    \State $\Sam \leftarrow \emptyset$, $d_\mathrm{i} \leftarrow 0$, $d_\mathrm{o} \leftarrow 0$, $t\leftarrow
    0$, $s\leftarrow 0$
    \For{\textbf{each} element $\left(\bullet, \left(u,v\right)\right)$ from $\Sigma$}
      \State{$t\leftarrow t +1$}
      \State{$s\leftarrow s \bullet 1$}
      \If{$\bullet = +$}
        \If{$\textsc{SampleEdge}\left(u,v\right)$}
          \State \textsc{UpdateCounters}$\left(+,\left(u,v\right)\right)$
          \Comment \textsc{UpdateCounters} is defined as in
          Alg.~\ref{alg:triest-base}.
        \EndIf
      \ElsIf{$\left(u,v\right)\in \Sam$}
        \State \textsc{UpdateCounters}$\left(-,\left(u,v\right)\right)$
        \State $\Sam\leftarrow \Sam\setminus\{(u,v)\}$
        \State $d_\mathrm{i} \leftarrow d_\mathrm{i} +1$
      \Else
        $\quad d_\mathrm{o} \leftarrow d_\mathrm{o} +1$
      \EndIf
    \EndFor
    \Statex
    \Function{\textsc{SampleEdge}}{$u,v$}
      \If{$d_\mathrm{o}+d_\mathrm{i}=0$}
        \If{$|\Sam|<M$}
          \State{$\Sam \leftarrow \Sam \cup \{\left(u,v\right)\}$}
          \State{\textbf{return}} True
          \ElsIf{\textsc{FlipBiasedCoin}$(\frac{M}{t}) = $ heads}
          \State Select $(z,w)$ uniformly at random from $\Sam$
          \State \textsc{UpdateCounters}$\left(-,(z,w)\right)$
          \State{$\Sam\leftarrow \left(\Sam \setminus \{(z,w)\}\right) \cup \{ \left(u,v\right)\}$}
          \State{\textbf{return}} True
        \EndIf
        \ElsIf{\textsc{FlipBiasedCoin}$\left(\frac{d_\mathrm{i}}{d_\mathrm{i}+d_\mathrm{o}}\right) = $ heads}
        \State{$\Sam \leftarrow \Sam \cup \{\left(u,v\right)\}$}
        \State $d_\mathrm{i} \leftarrow d_\mathrm{i} -1$
        \State{\textbf{return}} True
      \Else
        \State $d_\mathrm{o} \leftarrow d_\mathrm{o} -1$
        \State{\textbf{return}} False
      \EndIf
    \EndFunction
  \end{algorithmic}
\end{algorithm}

\subsubsection{Estimation}
We denote as $M^{(t)}$ the size of $\Sam$ at the end of time $t$ (we always have
$M^{(t)}\le M$). For any time $t$, let $d_\mathrm{i}^{(t)}$ and
$d_\mathrm{o}^{(t)}$ be the value of the counters $d_\mathrm{i}$ and
$d_\mathrm{o}$ at the end of time $t$ respectively, and let
$\omega^{(t)}=\min\{M, s^{(t)}+d_\mathrm{i}^{(t)}+d_\mathrm{o}^{(t)}\}$. Define
\begin{equation}\label{eq:kappa}
  \kappa^{(t)}=1-\sum_{j=0}^2\binom{s^{(t)}}{j}\binom{d_\mathrm{i}^{(t)}+d_\mathrm{o}^{(t)}}{\omega^{(t)}-j}\bigg/\binom{s^{(t)}+d_\mathrm{i}^{(t)}+d_\mathrm{o}^{(t)}}{\omega^{(t)}}\enspace.
\end{equation}
For any three positive integers $a,b,c$ s.t.~$a\le b\le c$, define\footnote{We
follow the convention that $\binom{0}{0}=1$.}
\[
  \psi_{a,b,c}=\binom{c}{b}\Big/\binom{c-a}{b-a}= \prod_{i=0}^{a-1}\frac{c-i}{b-i}\enspace.
\]
When queried at the end of time $t$, for an estimation of the global
number of triangles, \algofd returns
\[
  \rho^{(t)}=\left\{\begin{array}{l} 0 \text{ if } M^{(t)}<3 \\
  \frac{\tau^{(t)}}{\kappa^{(t)}}\psi_{3,M^{(t)},s^{(t)}}=\frac{\tau^{(t)}}{\kappa^{(t)}}\frac{s^{(t)}(s^{(t)}-1)(s^{(t)}-2)}{M^{(t)}(M^{(t)}-1)(M^{(t)}-2)}
  \text{ othw.}
\end{array}\right.
\]
\algofd can keep track of $\kappa^{(t)}$ during the execution, each update of
$\kappa^{(t)}$ taking time $O(1)$. Hence the time to return the estimations is
still $O(1)$.

\subsubsection{Analysis}
We now present the analysis of the estimations computed by \algoimproved,
showing results involving their unbiasedness, their variance, and their
concentration around their expectation. Results analogous to those in
Thms.~\ref{thm:fdunbiased}, \ref{thm:fdvariance}, and~\ref{thm:fdconcentration}
hold for the local triangle count for any $u\in V^{(t)}$, replacing the global
quantities with the corresponding local ones.

\subsubsection{Expectation}
Let $t^*$ be the first $t\ge M+1$ such that $|E^{(t)}|=M+1$, if such a time step
exists (otherwise $t^*=+\infty$).
\begin{theorem}\label{thm:fdunbiased}
  We have $\rho^{(t)}=|\Delta^{(t)}|$ for all $t<t^*$, and
  $\mathbb{E}\left[\rho^{(t)}\right] = |\Delta^{(t)}|$ for $t\ge t^*$.
\end{theorem}
The proof, deferred to App.~\ref{app:algobase}, relies on properties of RP and
on the definitions of $\kappa^{(t)}$ and $\rho^{(t)}$. Specifically, it uses
Lemma~\ref{lem:gemullahighorder}, which is the correspondent of
Lemma~\ref{lem:reservoirhighorder} but for RP, and some additional technical
lemmas (including an equivalent of Lemma~\ref{lem:baseunbiasedaux} but for RP)
and combine them using the law of total expectation by conditioning on the value
of $M^(t)$.

\subsubsection{Variance}
\begin{theorem}\label{thm:fdvariance}
  Let $t>t^*$ s.t.~$|\Delta^{(t)}| >0$ and $s^{(t)}\geq M$. Suppose we have
  $d^{(t)}=d_o^{(t)}+d_i^{(t)}\leq \alpha s^{(t)}$ total unpaired deletions at
  time $t$, with $0\leq\alpha< 1$. If $M\geq
  \frac{1}{2\sqrt{\alpha'-\alpha}}7\ln s^{(t)}$ for some $\alpha< \alpha'<1$, we
  have:
  \begin{align*}
    \mathrm{Var}\left[\rho^{(t)}\right] &\leq
    (\kappa^{(t)})^{-2}|\Delta^{(t)}|\left(\psi_{3,M(1-\alpha'),s^{(t)}}-1 \right) +(\kappa^{(t)})^{-2}2\\
    &+(\kappa^{(t)})^{-2}r^{(t)}\left(\psi_{3,M(1-\alpha'),s^{(t)}}^2\psi_{5,M(1-\alpha'),s^{(t)}}^{-1}-1\right)
  \end{align*}
\end{theorem}

The proof of Thm.~\ref{thm:fdvariance} is deferred to App.~\ref{app:algofd}. It
uses two results on the variance of $\rho^{(t)}$ conditioned on a specific value of
$M^{(t)}$ (Lemmas~\ref{lem:confdvariance} and~\ref{lem:boundvarfd}), and an
analysis of the probability distribution of $M^{(t)}$ (Lemma~\ref{lem:mtprops}
and Corollary~\ref{corol:samplesizebound}). These results are then combined
using the law of total variance.

\subsubsection{Concentration}
The following result relies on Chebyshev's inequality and on
Thm.~\ref{thm:fdvariance}, and the proof (reported in App.~\ref{app:algofd})
follows the steps similar to those in the proof for Thm.~\ref{thm:improvedvariance}.
\begin{theorem}\label{thm:fdconcentration}
  Let $t\ge t^*$ s.t.~$|\Delta^{(t)}| >0$ and $s^{(t)}\geq M$. Let
  $d^{(t)}=d_o^{(t)}+d_i^{(t)}\leq \alpha s^{(t)}$ for some $0\leq\alpha< 1$.
  For any $\varepsilon,\delta\in(0,1)$, if for some $\alpha< \alpha'<1$
  \begin{align*}
    M >&\max\Bigg\{\frac{1}{\sqrt{a'-\alpha}}7\ln s^{(t)},\\
    &(1-\alpha')^{-1}\left(\sqrt[3]{\frac{2s^{(t)}(s^{(t)}-1)(s^{(t)}-2)}{\delta \varepsilon^2|\Delta^{(t)}|(\kappa^{(t)})^{2}+2\frac{|\Delta^{(t)}|-2}{|\Delta^{(t)}|}}}+2 \right),\\
  &\frac{(1-\alpha')^{-1}}{3} \left( \frac{r^{(t)}s^{(t)}}{\delta \varepsilon^2 |\Delta^{(t)}|^2(\kappa^{(t)})^{-2}+2r^{(t)}}\right)\Bigg\}
  \end{align*}
  then $|\rho^{(t)}-|\Delta^{(t)}||<\varepsilon|\Delta^{(t)}|$ with
  probability $>1-\delta$.
\end{theorem}

\subsection{Counting global and local triangles in
multigraphs}\label{sec:multigraphs}
We now discuss how to extend \algo to approximate the local and global triangle
counts in multigraphs.

\subsubsection{\algobase on multigraphs}
\algobase can be adapted to work on multigraphs as follows. First of all, the
sample $\Sam$ should be considered a \emph{bag}, i.e., it may contain multiple
copies of the same edge. Secondly, the function \textsc{UpdateCounters} must be
changed as presented in Alg.~\ref{alg:triest-base-multi}, to take into account
the fact that inserting or removing an edge $(u,v)$ from $\Sam$ respectively
increases or decreases the global number of triangles in $G^\Sam$ by a quantity
that depends on the product of the number of edges $(c,u)\in\Sam$ and
$(c,v)\in\Sam$, for $c$ in the shared neighborhood (in $G^\Sam$) of $u$ and $v$
(and similarly for the local number of triangles incidents to $c$).

\begin{algorithm}[ht]
  \small 
  \caption{\textsc{UpdateCounters} function for \algobase on multigraphs}
  \label{alg:triest-base-multi}
  \begin{algorithmic}[1]
    \Statex
    \Function{UpdateCounters}{$(\bullet, (u,v))$}
      \State $\mathcal{N}^\Sam_{u,v} \leftarrow \mathcal{N}^\Sam_u \cap \mathcal{N}^\Sam_v$
      \ForAll {$c \in \mathcal{N}^\Sam_{u,v}$}
        \State $y_{c,u} \leftarrow $ number of edges between $c$ and $u$ in
        $\Sam$
        \State $y_{c,v} \leftarrow $ number of edges between $c$ and $v$ in
        $\Sam$
        \State $y_{c} \leftarrow y_{c,u}\cdot y_{c,v}$
        \State $\tau \leftarrow \tau \bullet y_c$
        \State $\tau_c \leftarrow \tau_c \bullet y_c$
        \State $\tau_u \leftarrow \tau_u \bullet y_c$
        \State $\tau_v \leftarrow \tau_v \bullet y_c$
      \EndFor
    \EndFunction
  \end{algorithmic}
\end{algorithm}

For this modified version of \algobase, that we call \textsc{\algobase-m}, an
equivalent version of Lemma~\ref{lem:baseunbiasedaux} holds. Therefore, we can
prove a result on the unbiasedness of \textsc{\algobase-m} equivalent (i.e.,
with the same statement) as Thm.~\ref{thm:baseunbiased}. The proof of such
result is also the same as the one for Thm.~\ref{thm:baseunbiased}.

To analyze the variance of \textsc{\algobase-m}, we need to take into
consideration the fact that, in a multigraph, a pair of triangles may share
\emph{two} edges, and the variance depends (also) on the number of such pairs.
Let $r_1^{(t)}$ be the number of unordered pairs of distinct triangles from
$\Delta^{(t)}$ sharing an edge and let $r_2^{(t)}$ be the number of unordered
pairs of distinct triangles from $\Delta^{(t)}$ sharing \emph{two} edges (such
pairs may exist in a multigraph, but not in a simple graph). Let
$q^{(t)}=\binom{|\Delta^{(t)}|}{2}-r_1^{(t)}-r_2^{(t)}$ be the number of
unordered pairs of distinct triangles that do not share any edge.

\begin{theorem}\label{thm:basevariance-multi}
  For any $t>M$, let $f(t) = \xi^{(t)}-1$,
  \[
    g(t) = \xi^{(t)}\frac{(M-3)(M-4)}{(t-3)(t-4)} -1
  \]
  and
  \[
    h(t) = \xi^{(t)}\frac{(M-3)(M-4)(M-5)}{(t-3)(t-4)(t-5)}
    -1\enspace(\le 0),
  \]
  and
  \[
    j(t) = \xi^{(t)}\frac{M-3}{t-3}
    -1\enspace.
  \]
  We have:
  \[
    \variance\left[\xi(t)\tau^{(t)}\right] = |\Delta^{(t)}|
    f(t)+r_1^{(t)}g(t)+r_2^{(t)}j(t)+q^{(t)}h(t).
  \]
\end{theorem}
The proof follows the same lines as the one for Thm.~\ref{thm:basevariance},
with the additional steps needed to take into account the contribution of the
$r_2^{(t)}$ pairs of triangles in $G^{(t)}$ sharing two edges.

\subsubsection{\algoimproved on multigraphs}
A variant \textsc{\algoimproved-m} of \algoimproved for multigraphs can be
obtained by using the function \textsc{UpdateCounters} defined in
Alg.~\ref{alg:triest-base-multi}, modified to increment\footnote{As in
\algoimproved, all calls to \textsc{UpdateCounters} in \textsc{\algoimproved-m}
have $\bullet=+$. See also Alg.~\ref{alg:triest-impr}.} the counters by
$\eta^{(t)}y_c^{(t)}$, rather than $y_c^{(t)}$, where
$\eta^{(t)}=\max\{1,(t-1)(t-2)/(M(M-1))\}$. The result stated in
Thm.~\ref{thm:improvedunbiased} holds also for the estimations computed by
\textsc{\algoimproved-m}. An upper bound to the variance of the estimations,
similar to the one presented in Thm.~\ref{thm:improvedvariance} for
\algoimproved, could potentially be obtained, but its derivation would involve a
high number of special cases, as we have to take into consideration the order of
the edges in the stream.

\subsubsection{\algofd on multigraphs}
\algofd can be modified in order to provide an approximation of the number of
global and local triangles on multigraphs observed as a stream of edge deletions
and deletions. It is however necessary to clearly state the data model. We
assume that for all pairs of vertices $u,v \in V^{(t)}$, each edge connecting
$u$ and $v$ is assigned a label that is unique among the edges connecting $u$
and $v$. An edge is therefore uniquely identified by its endpoints and its label
as $(u,v), label)$. Elements of the stream are now in the form $(\bullet, (u,v),
label)$ (where $\bullet$ is either $+$ or $-$). This assumption, somewhat
strong, is necessary in order to apply the \emph{random pairing} sampling
scheme~\citep{GemullaLH08} to fully-dynamic multigraph edge streams.

Within this model, we can obtain an algorithm \textsc{\algofd-m} for multigraphs
by adapting \algofd as follows. The sample $\Sam$ is a \emph{set} of elements
$((u,v), label)$. When a deletion $(-, (u,v), label)$ is on the stream, the
sample $\Sam$ is modified if and only if $((u,v), label)$ belongs to $\Sam$.
This change can be implemented in the pseudocode from Alg.~\ref{algo: TRIEST-FD}
by modifying line 8 to be
\[
  \mbox{``\bf else if } ((u,v), label)\in \Sam \mbox{ {\bf then}''}\enspace.
\]
Additionally, the function \textsc{UpdateCounters} to be used is the one
presented in Alg.~\ref{alg:triest-base-multi}.

We can prove a result on the unbiasedness of \textsc{\algofd-m} equivalent
(i.e., with the same statement) as Thm.~\ref{thm:fdunbiased}. The proof of such
result is also the same as the one for Thm.~\ref{thm:fdunbiased}. An upper bound
to the variance of the estimations, similar to the one presented in
Thm.~\ref{thm:fdvariance} for \algofd, could be obtained by considering the fact
that in a multigraph two triangles can share two edges, in a fashion similar to
what we discussed in Thm.~\ref{thm:basevariance-multi}.

\subsection{Discussion}\label{sec:discussion}
We now briefly discuss over the algorithms we just presented, the techniques
they use, and the theoretical results we obtained for \algo, in order to
highlight advantages, disadvantages, and limitations of our approach.

\paragraph{On reservoir sampling}
Our approach of using reservoir sampling to keep a random sample of edges can be
extended to many other graph mining problems, including approximate counting of
other subgraphs more or less complex than triangles (e.g., squares, trees with a
specific structure, wedges, cliques, and so on). The estimations of such counts
would still be unbiased, but as the number of edges composing the subgraph(s) of
interest increases, the variance of the estimators also increases, because the
probability that all edges composing a subgraph are in the sample (or all but
the last one when the last one arrives, as in the case of \algoimproved),
decreases as their number increases. Other works in the triangle counting
literature~\citep{PavanTTW13,JhaSP15} use samples of wedges, rather than edges.
They perform worse than \algo in both accuracy and runtime (see
Sect.~\ref{sec:experiments}), but the idea of sampling and storing more complex
structures rather than simple edges could be a potential direction for
approximate counting of larger subgraphs.

\paragraph{On the analysis of the variance}
We showed an exact analysis of the variance of \algobase but for the other
algorithms we presented \emph{upper bounds} to the variance of the estimates.
These bounds can still be improved as they are not currently tight. For example,
we already commented on the fact that the bound in~\eqref{eq:improvedvariance}
does not include a number of negative terms that would tighten it
(i.e., decrease the bound), and that could potentially be no smaller than the
term depending on $z^{(t)}$. The absence of such terms is due to the fact that
it seems very challenging to obtain non-trivial \emph{upper bounds} to them
that are valid for every $t>M$. Our proof for this bound uses a careful case-by-case
analysis, considering the different situations for pair of triangles (e.g.,
sharing or not sharing an edge, and considering the order of edges on the
stream). It may be possible to obtain tighter bounds to the variance by
following a more holistic approach that takes into account the fact that the
sizes of the different classes of triangle pairs are highly dependent on each
other.

Another issue with the bound to the variance from~\eqref{eq:improvedvariance} is
that the quantity $z^{(t)}$ depends on the order of edges on the stream. As
already discussed, the bound can be made independent of the order by loosening
it even more. Very recent developments in the sampling theory
literature~\citep{DevilleT04} presented sampling schemes and estimators whose
second-order sampling probabilities do not depend on the order of the stream, so
it should be possible to obtain such bounds also for the triangle counting
problem, but a sampling scheme different than reservoir sampling would have to
be used, and a careful analysis is needed to establish its net advantages in
terms of performances and scalability to billion-edges graphs.

\paragraph{On the trade-off between speed and accuracy}
We concluded both previous paragraphs in this subsection by mentioning
techniques different than reservoir sampling of edges as potential directions to
improve and extend our results. In both cases these techniques are more complex
not only in their analysis but also \emph{computationally}. Given that the main
goal of algorithms like \algo is to make it possible to analyze graphs with
billions (and possibly more) nodes, the gain in accuracy
need to be weighted against expected slowdowns in execution. As we show in our
experimental evaluation in the next section, \algo, especially in the
\algoimproved variant, actually seems to strike the right balance between
accuracy and tradeoff, when compared with existing contributions.

\section{Experimental evaluation}\label{sec:experiments}

We evaluated \algo on several real-world graphs with up to a billion edges. The
algorithms were implemented in \verb C++, and ran on the Brown University CS
department
cluster.\footnote{\url{https://cs.brown.edu/about/system/services/hpc/grid/}}
Each run employed a single core and used at most $4$ GB of RAM. The code is
available from \url{http://bigdata.cs.brown.edu/triangles.html}.

\paragraph{Datasets} We created the streams from the following publicly
available graphs (properties in Table~\ref{table:graphs}).

\begin{description}
  \item[Patent (Co-Aut.) and Patent (Cit.)] The {\it Patent (Co-Aut.)} and {\it
    Patent (Cit.)} graphs are obtained from a dataset of $\approx 2$ million
    U.S.~patents granted between '75 and '99~\citep{hjt01}. In {\it Patent
    (Co-Aut.)}, the nodes represent inventors and there is an edge with
    timestamp $t$ between two co-inventors of a patent if the patent was granted
    in year $t$. In {\it Patent (Cit.)}, nodes are patents and there is an edge
    $(a,b)$ with timestamp $t$ if patent $a$ cites $b$ and $a$ was granted in
    year $t$.
  \item[LastFm] The LastFm graph is based on a dataset~\citep{celma2009music,
    koblenz} of $\approx20$ million \url{last.fm} song listenings, $\approx1$
    million songs and $\approx1000$ users. There is a node for each song and an
    edge between two songs if $\ge3$ users listened to both on day $t$.
  \item[Yahoo!-Answers] The Yahoo!~Answers graph is obtained from a sample of
    $\approx 160$ million answers to $\approx25$ millions questions posted on
    Yahoo!~Answers~\citep{yahoo-ans}.  An edge connects two users at time
    $max(t_1, t_2)$ if they both answered the same question at times $t_1$,
    $t_2$ respectively. We removed $6$ outliers questions with more than $5000$
    answers.
  \item[Twitter] This is a snapshot~\citep{kwak2010twitter,brsllp} of the
    Twitter followers/following network with $\approx 41$ million nodes and
    $\approx 1.5$ billions edges. We do not have time information for the edges,
    hence we assign a random timestamp to the edges (of which we ignore the
    direction).
\end{description}

\paragraph{Ground truth}
To evaluate the accuracy of our algorithms, we computed
the \emph{ground truth} for our smaller graphs (i.e., the exact number of
global and local triangles for each time step), using an exact algorithm. The
entire current graph is stored in memory and when an edge $u,v$ is inserted (or
deleted) we update the current count of local and global triangles by checking
how many triangles are completed (or broken). As exact algorithms are not
scalable, computing the exact triangle count is feasible only for small graphs
such as Patent (Co-Aut.), Patent (Cit.) and LastFm. Table~\ref{table:graphs}
reports the exact total number of triangles at the end of the stream for those
graphs (and an estimate for the larger ones using \algoimproved with
$M=1000000$).

\begin{table}[ht]
  \tbl{Properties of the dynamic graph streams analyzed. $|V|$, $|E|$, $|E_u|$,
    $|\Delta|$ refer respectively to the number of nodes in the graph, the
    number of edge addition events, the number of distinct edges additions, and
    the maximum number of triangles in the graph (for Yahoo!~Answers and Twitter
    estimated with \algoimproved with $M=1000000$, otherwise computed exactly with
    the na\"ive algorithm).
  }{
    \begin{tabular}{cccccc}
    \toprule
    Graph  & $|V|$ & $|E|$ &$|E_u|$ & $|\Delta|$ \\
    \midrule
    Patent (Co-Aut.) & $1{,}162{,}227$ & $3{,}660{,}945$ & 2{,}724{,}036 & \num{3.53e+06} \\
    \midrule
    Patent (Cit.) & $2{,}745{,}762$ & $13{,}965{,}410$  & 13{,}965{,}132 &  \num{6.91e+06}\\
    \midrule
    LastFm & $681{,}387$ & $43{,}518{,}693$  & 30{,}311{,}117 & \num{1.13e+09}\\
    \midrule
    Yahoo!~Answers & $2{,}432{,}573$ & $\num{1.21e9}$&\num{1.08e9} & \num{7.86e10} \\
    \midrule
    Twitter &  $41{,}652{,}230$& $\num{1.47e9}$ & \num{1.20e9}& \num{3.46e10}\\
    \end{tabular}
  } 
  \label{table:graphs}
\end{table}

\subsection{Insertion-only case}
We now evaluate \algo on insertion-only streams and compare its performances
with those of state-of-the-art approaches~\citep{LimK15,JhaSP15,PavanTTW13},
showing that \algo has an average estimation error significantly smaller than
these methods both for the global and local estimation problems, while using the
same amount of memory.

\paragraph{Estimation of the global number of triangles}
Starting from an empty graph we add one edge at a time, in timestamp order.
Figure~\ref{fig:add-only} illustrates the evolution, over time, of the
estimation computed by \algoimproved with $M=1{,}000{,}000$. For smaller graphs
for which the ground truth can be computed exactly, the curve of the exact count
is practically indistinguishable from \algoimproved estimation, showing the
precision of the method. The estimations have very small variance even on the
very large Yahoo!\ Answers and Twitter graphs (point-wise max/min estimation
over ten runs is almost coincident with the average estimation). These results
show that \algoimproved is very accurate even when storing less than a $0.001$
fraction of the total edges of the graph.

\begin{figure}[ht]
\subfigure[Patent
(Cit.)]{\includegraphics[width=0.49\textwidth,keepaspectratio]{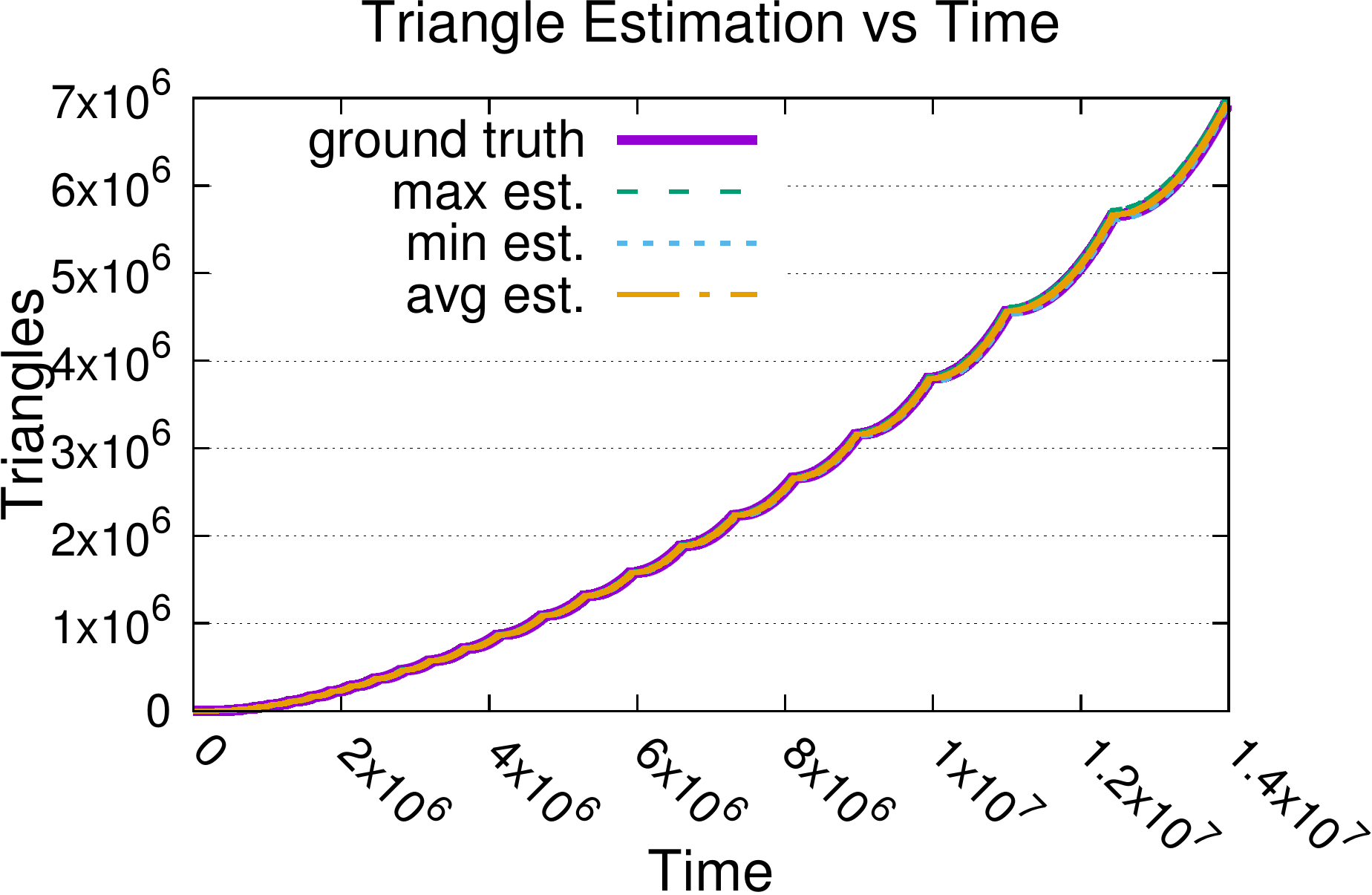}}
\subfigure[LastFm]{\includegraphics[width=0.49\textwidth,keepaspectratio]{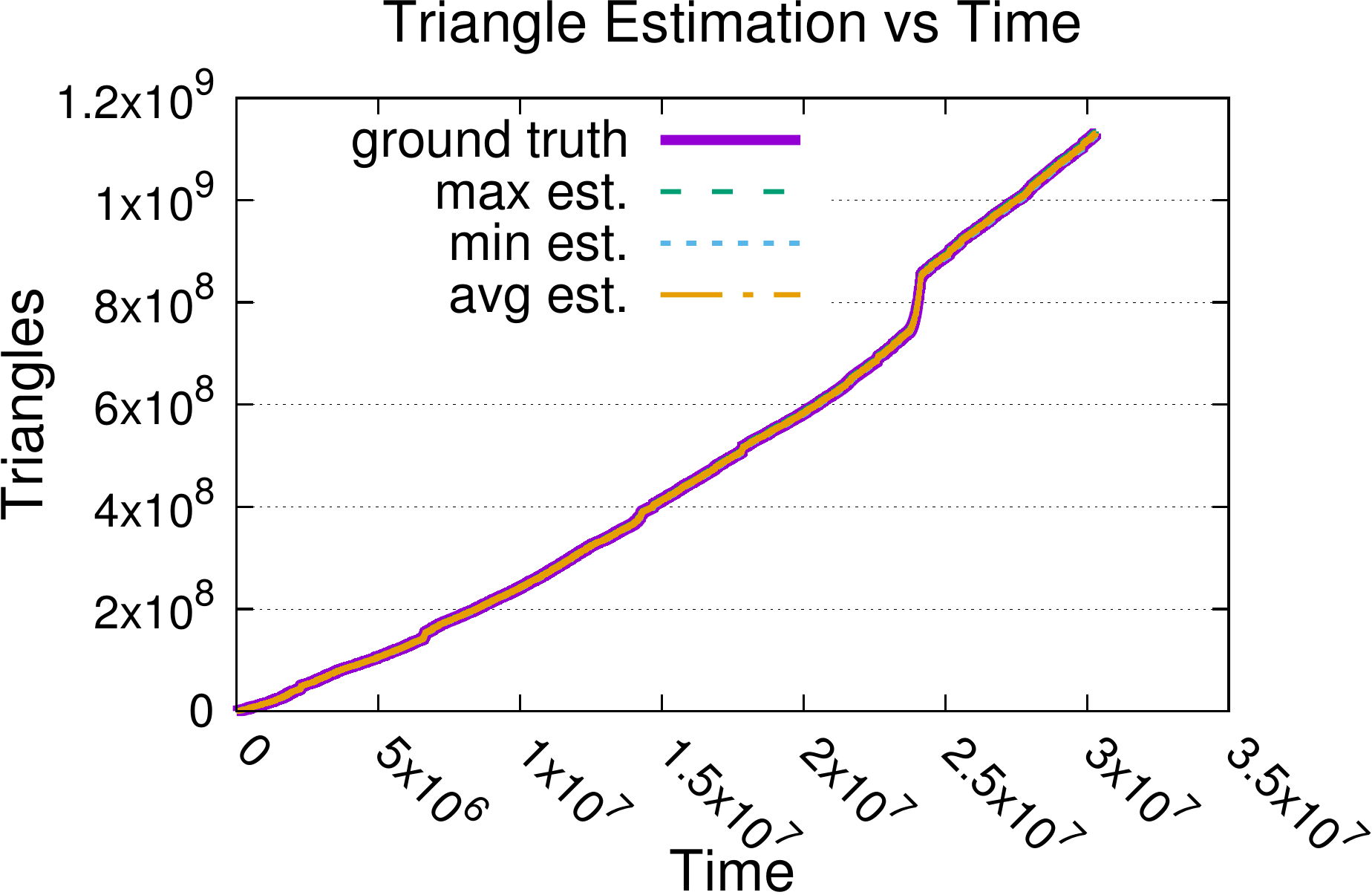}}
\subfigure[Yahoo!\
Answers]{\includegraphics[width=0.49\textwidth,keepaspectratio]{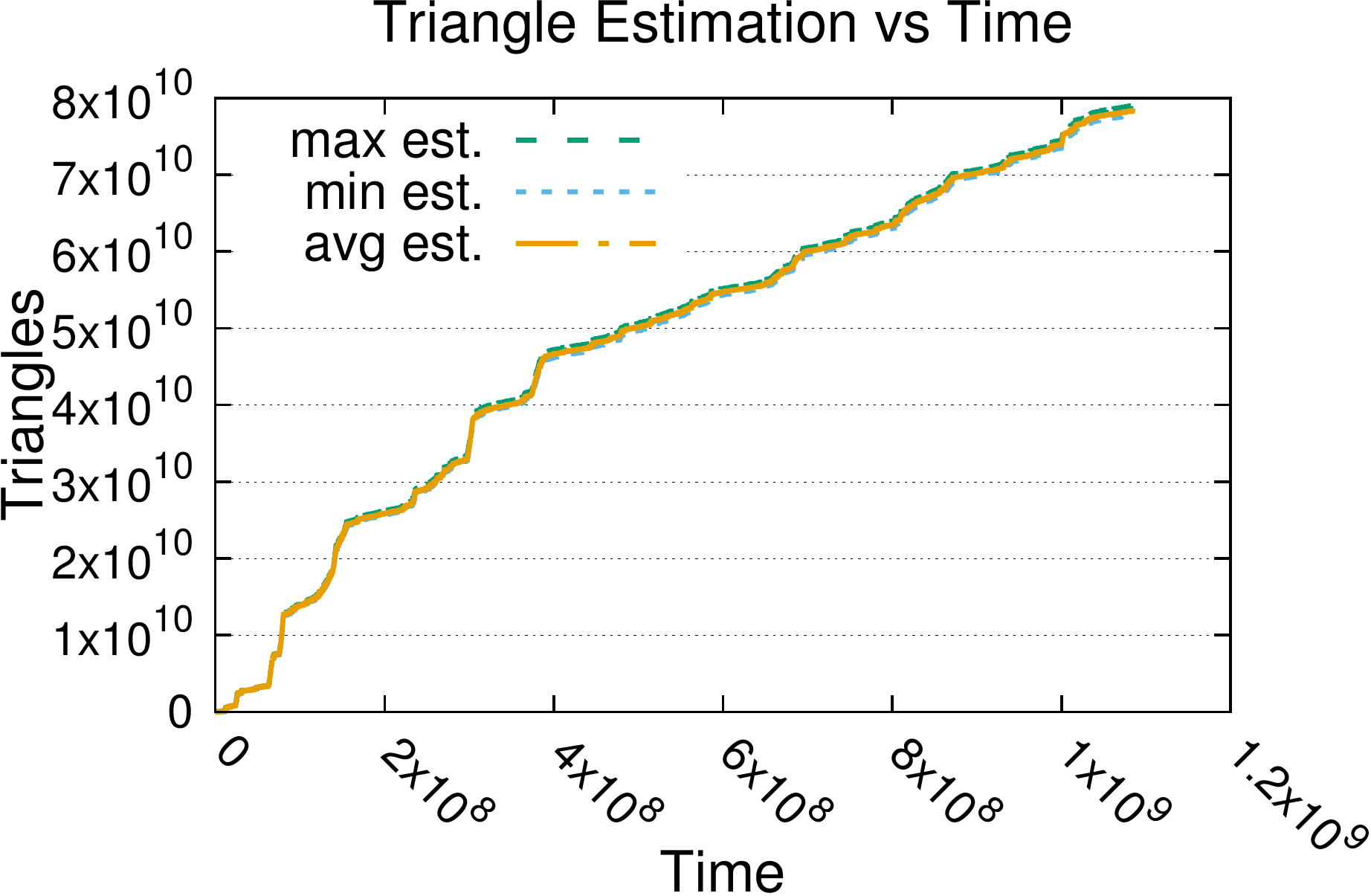}}
\subfigure[Twitter]{\includegraphics[width=0.49\textwidth,keepaspectratio]{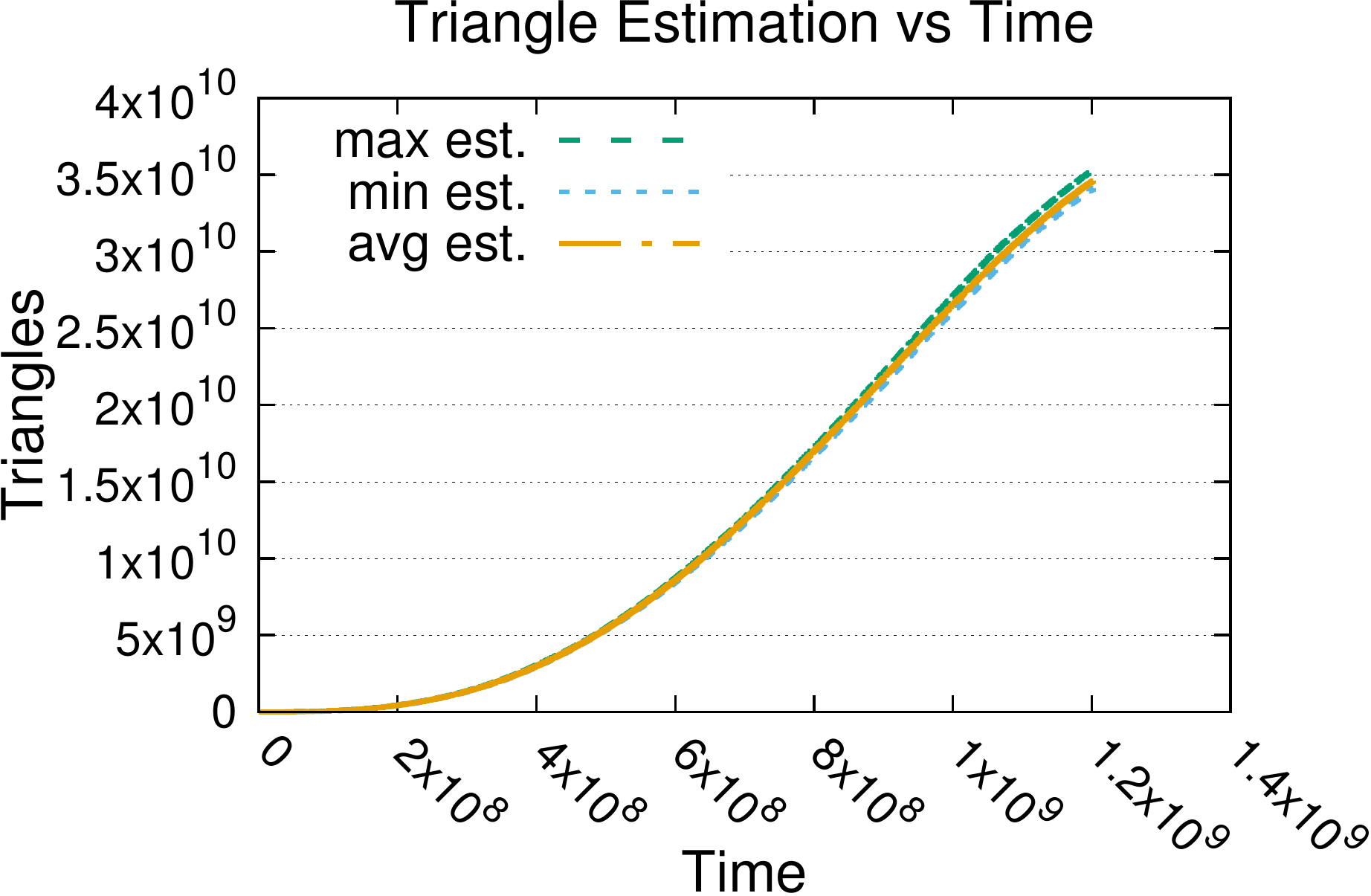}}
\caption{Estimation by \algoimproved of the global number of triangles over
  time (intended as number of elements seen on the stream). The max, min, and
  avg are taken over 10 runs. The curves are \emph{indistinguishable on
  purpose}, to highlight the fact that \algoimproved estimations have very small
  error and variance. For example, the ground truth (for graphs for which it is
available) is indistinguishable even from the max/min point-wise estimations
over ten runs. For graphs for which the ground truth is not available, the small
deviations from the avg suggest that the estimations are also close to the true
value, given that our algorithms gives unbiased estimations.}
\label{fig:add-only}
\end{figure}

\paragraph{Comparison with the state of the art}
We compare quantitatively with three state-of-the-art methods:
\textsc{mascot}~\citep{LimK15}, \textsc{Jha et al.}~\citep{JhaSP15} and
\textsc{Pavan et al.}~\citep{PavanTTW13}. \textsc{mascot} is a suite of local
triangle counting methods (but provides also a global estimation). The other two
are global triangle counting approaches. None of these can handle fully-dynamic
streams, in contrast with \algofd. We first compare the three methods to \algo
for the global triangle counting estimation. \textsc{mascot} comes in two memory
efficient variants: the basic \textsc{mascot-c} variant and an improved
\textsc{mascot-i} variant.\footnote{In the original work~\citep{LimK15}, this
variant had no suffix and was simply called \textsc{mascot}. We add the
\textsc{-i} suffix to avoid confusion. Another variant \textsc{mascot-A}
can be forced to store the entire graph with probability $1$ by appropriately
selecting the edge order (which we assume to be adversarial) so we do not
consider it here.} Both variants sample edges with fixed probability $p$, so
there is no guarantee on the amount of memory used during the execution. To
ensure fairness of comparison, we devised the following experiment. First, we
run both \textsc{mascot-c} and \textsc{mascot-i} for $\ell=10$ times with a
fixed $p$ using the same random bits for the two algorithms run-by-run (i.e. the
same coin tosses used to select the edges) measuring each time the number of
edges $M'_i$ stored in the sample at the end of the stream (by construction this
the is same for the two variants run-by-run). Then, we run our algorithms using
$M=M'_i$ (for $i\in[\ell]$). We do the same to fix the size of the edge memory
for \textsc{Jha et al.}~\citep{JhaSP15} and \textsc{Pavan et
al.}~\citep{PavanTTW13}.\footnote{More precisely, we use $M'_i/2$ estimators in
\textsc{Pavan et al.} as each estimator stores two edges. For \textsc{Jha et
al.} we set the two reservoirs in the algorithm to have each size $M'_i/2$. This
way, all algorithms use $M'_i$ cells for storing (w)edges.} This way, all
algorithms use the same amount of memory for storing edges (run-by-run).

We use the \emph{MAPE} (Mean Average Percentage Error) to assess the accuracy of
the global triangle estimators over time. The MAPE measures the average
percentage of the prediction error with respect to the ground truth, and is
widely used in the prediction literature~\citep{hk06}. For $t=1,\dotsc,T$, let
$\overline{\Delta}^{(t)}$ be the estimator of the number of triangles at time
$t$, the MAPE is defined as
$ \frac{1}{T}{\sum_{t=1}^T \left|\frac{|\Delta^{(t)}| -
\overline{\Delta}^{(t)}}{|\Delta^{(t)}|}\right|}$.\footnote{The MAPE is not
defined for $t$ s.t.~$\Delta^{(t)}=0$ so we compute it only for $t$
s.t.~$|\Delta^{(t)}| > 0$. All algorithms we consider are guaranteed to output the
correct answer for $t$ s.t.~$|\Delta^{(t)}|=0$.}

In Fig.~\ref{fig:mape-all-add}, we compare the average MAPE of
\algobase and \algoimproved as well as the two \textsc{mascot} variants and the other two streaming
algorithms for the Patent (Co-Aut.) graph, fixing $p=0.01$.  \algoimproved has
the smallest error of all the algorithms compared.

We now turn our attention to the efficiency of the methods. Whenever we refer to
one operation, we mean handling one element on the stream, either one edge
addition or one edge deletion. The average update time per operation is obtained
by dividing the total time required to process the entire stream by the number
of operations (i.e., elements on the streams).

Figure~\ref{fig:time-all-add} shows the average update time per operation in
Patent (Co-Aut.) graph, fixing $p=0.01$. Both \textsc{Jha et
al.}~\citep{JhaSP15} and \textsc{Pavan et al.}~\citep{PavanTTW13}
are up to $\approx 3$ orders of magnitude slower than the \textsc{mascot} variants and \algo.
This is expected as both algorithms have an update complexity of $\Omega(M)$
(they have to go through the entire reservoir graph at each step),
while both \textsc{mascot} algorithms and \algo need only to access the
neighborhood of the nodes involved in the edge addition.\footnote{We observe
that \textsc{Pavan et al.}~\citep{PavanTTW13} would be more efficient with batch
updates. However, we want to estimate the triangles continuously at each update.
In their experiments they use batch sizes of million of updates for efficiency.}
This allows both algorithms to efficiently exploit larger memory sizes. We can
use efficiently $M$ up to $1$ million edges in our experiments, which only
requires few megabytes of RAM.\footnote{The experiments by~\citet{JhaSP15}
use $M$ in the order of $10^3$, and in those by~\citet{PavanTTW13}, large $M$
values require large batches for efficiency.}
\textsc{mascot} is one order of magnitude faster than \algo
(which runs in $\approx 28$ micros/op), because it does not have to handle edge
removal from the sample, as it offers no guarantees on the used memory. As we
will show, \algo has much higher precision and scales well on billion-edges
graphs.

\begin{table}[ht]
  \tbl{Global triangle estimation MAPE for \algo and \textsc{mascot}. The
    rightmost column shows the reduction in terms of the avg.~MAPE obtained by
    using \algo.  Rows with $Y$ in column ``Impr.'' refer to improved algorithms
    (\algoimproved and \textsc{mascot-i}) while those with $N$ to basic
    algorithms (\algobase and \textsc{mascot-c}).
  }{
    \begin{tabular}{lcp{5pt}ccccc}
      \toprule
      \multicolumn{3}{c}{}&\multicolumn{2}{c}{Max.~MAPE}&\multicolumn{3}{c}{Avg.~MAPE}\\
      \cmidrule(l{2pt}r{2pt}){4-5}\cmidrule(l{2pt}r{2pt}){6-8}
      Graph & Impr. &$p$ & \textsc{mascot} & \algo & \textsc{mascot} &
      \algo & Change\\
      \midrule
      \multirow{4}{.7cm}{Patent (Cit.)} & N &0.01& 0.9231 & 0.2583 & 0.6517 & 0.1811&-72.2\%\\
      & Y &0.01 & 0.1907&0.0363&	0.1149&	0.0213&	-81.4\%\\
      &N&	0.1& 0.0839&0.0124 &0.0605&0.0070&-88.5\%\\
      & Y &0.1& 0.0317&0.0037& 0.0245&0.0022&	-91.1\%\\
      \midrule
      \multirow{4}{.7cm}{Patent (Co-A.)}	&N&0.01&	2.3017&	0.3029	&0.8055	&0.1820&-77.4\%\\
      & Y &0.01&0.1741 &	0.0261&	0.1063&	0.0177&	-83.4\%\\
      &N&0.1&	0.0648&	0.0175&	0.0390&	0.0079&	-79.8\%\\
      & Y &0.1&0.0225	& 0.0034&0.0174	& 0.0022	&	-87.2\%\\
      \midrule
      \multirow{4}{.7cm}{LastFm} &N&	0.01&	0.1525&	0.0185& 0.0627&	 0.0118&-81.2\%\\
      & Y &0.01&	0.0273&	0.0046 &0.0141	&0.0034&	-76.2\%\\
      &N&	0.1	&	0.0075&	0.0028&0.0047& 0.0015&		-68.1\%\\
      & Y &0.1&0.0048		&0.0013 & 0.0031	&0.0009&	-72.1\%\\
      \bottomrule
    \end{tabular}
  } 
  \label{table:fix-p-vs-res-add}
\end{table}

Given the slow execution of the other algorithms on the larger datasets we
compare in details \algo only with \textsc{mascot}.\footnote{We attempted to run
the other two algorithms but they did not complete after $12$ hours for the
larger datasets in Table~\ref{table:fix-p-vs-res-add} with the prescribed $p$
parameter setting.} Table~\ref{table:fix-p-vs-res-add} shows the average MAPE of
the two approaches. The results confirm the pattern observed in
Figure~\ref{fig:mape-all-add}: \algobase and \algoimproved both have an average
error significantly smaller than that of the basic \textsc{mascot-c} and
improved \textsc{mascot} variant respectively. We achieve up to a 91\% (i.e.,
$9$-fold) reduction in the MAPE while using the same amount of memory. This
experiment confirms the theory: reservoir sampling has overall lower or equal
variance in all steps for the same expected total number of sampled edges.

\begin{figure}[ht]
\subfigure[MAPE]{\label{fig:mape-all-add}\includegraphics[width=0.49\textwidth]{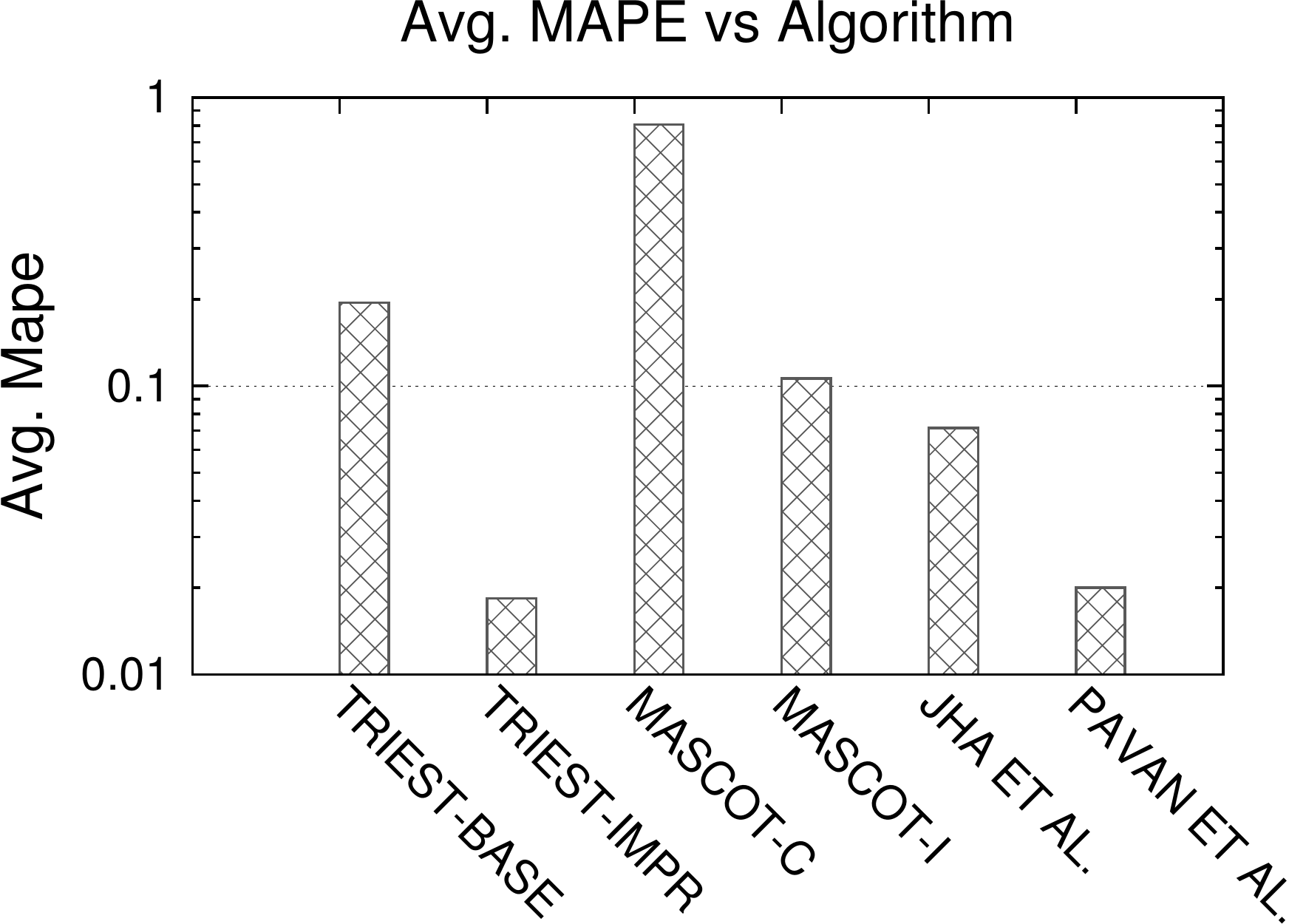}}
\subfigure[Update
Time]{\label{fig:time-all-add}\includegraphics[width=0.49\textwidth]{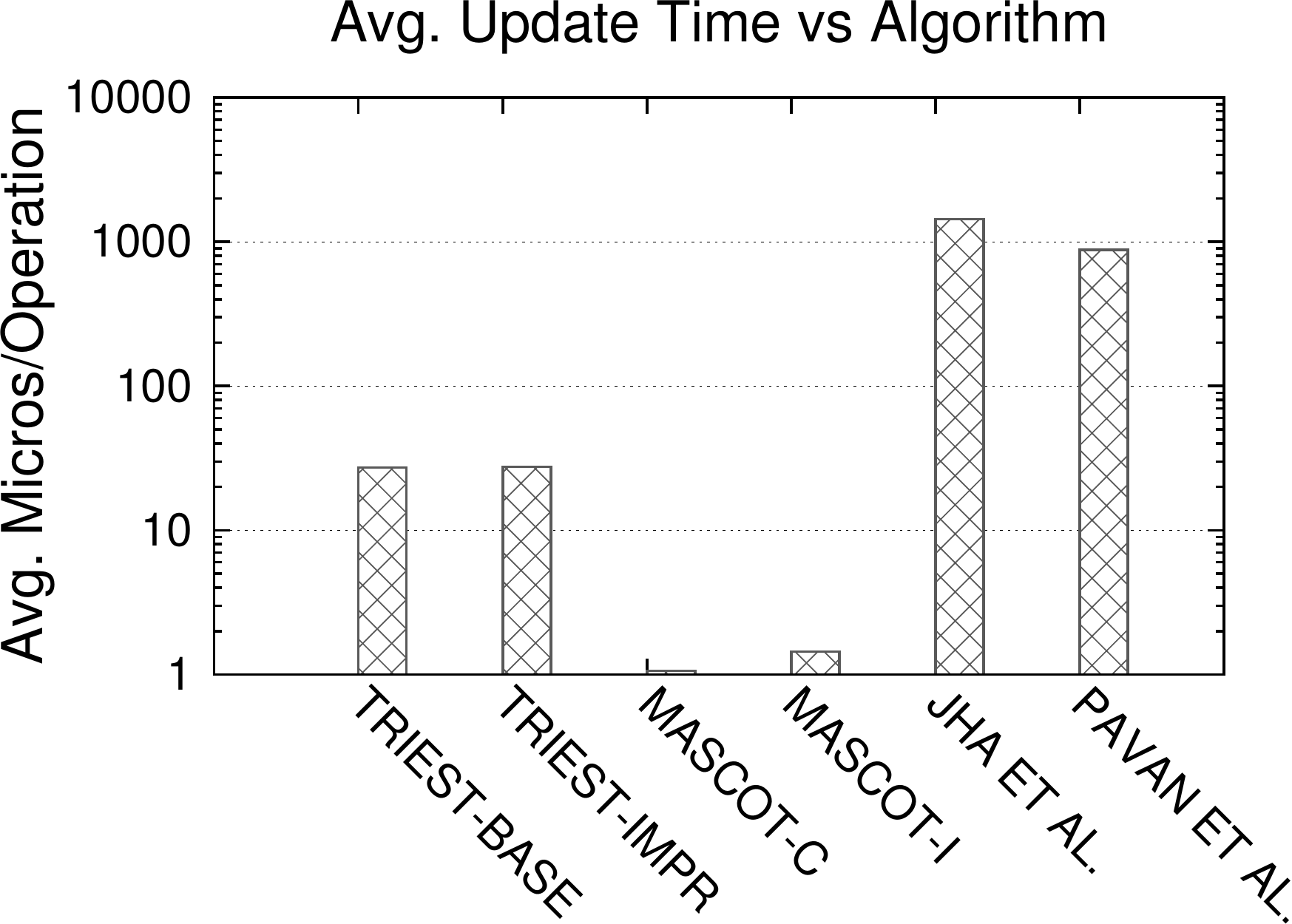}}
\caption{Average MAPE and average update time of the various methods on
the Patent (Co-Aut.) graph with $p=0.01$ (for \textsc{mascot}, see the main text
for how we computed the space used by the other algorithms) -- insertion only.
\algoimproved has the lowest error. Both \textsc{Pavan et al.} and \textsc{Jha
et al.} have very high update times compared to our method and the two
\textsc{mascot} variants.}
\label{fig:mape-time-all-add}
\end{figure}

To further validate this observation we run \algoimproved and the improved
\textsc{mascot-i} variant using the same (expected memory) $M=10000$.
Figure~\ref{fig:variance-improv-vs-mascot-sh} shows the max-min estimation over
$10$ runs and the standard deviation of the estimation over those runs.
\algoimproved shows significantly lower standard deviation (hence variance) over
the evolution of the stream, and the max and min lines are also closer to the
ground truth.  This confirms our theoretical observations in the previous
sections. Even with very low $M$ (about $2/10000$ of the size of the graph)
\algo gives high-quality estimations.

\begin{figure}[ht]
  \subfigure[Ground truth, max, and min]{\includegraphics[width=0.49\textwidth]{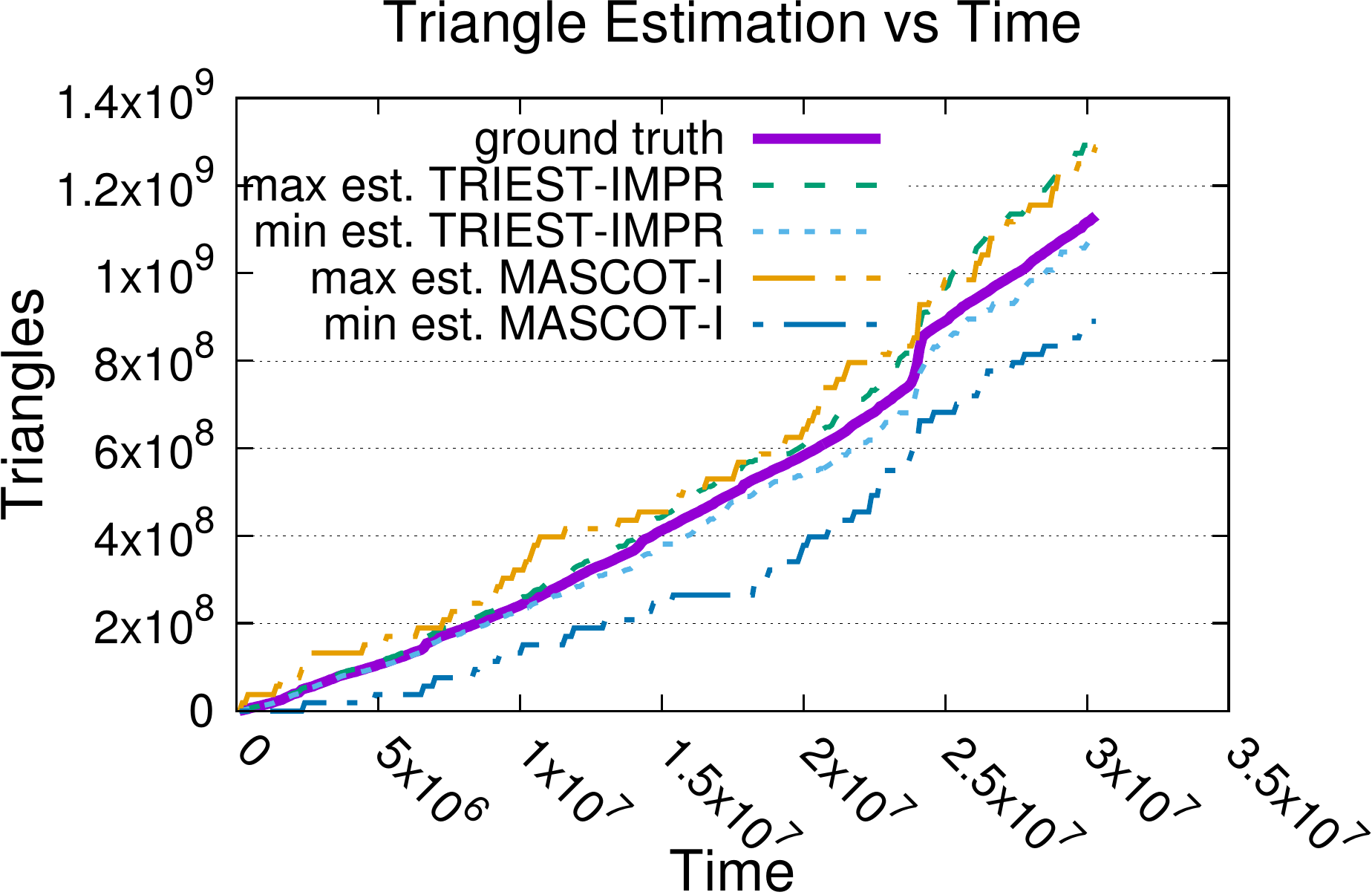}}
  \subfigure[Standard deviation]{\includegraphics[width=0.49\textwidth]{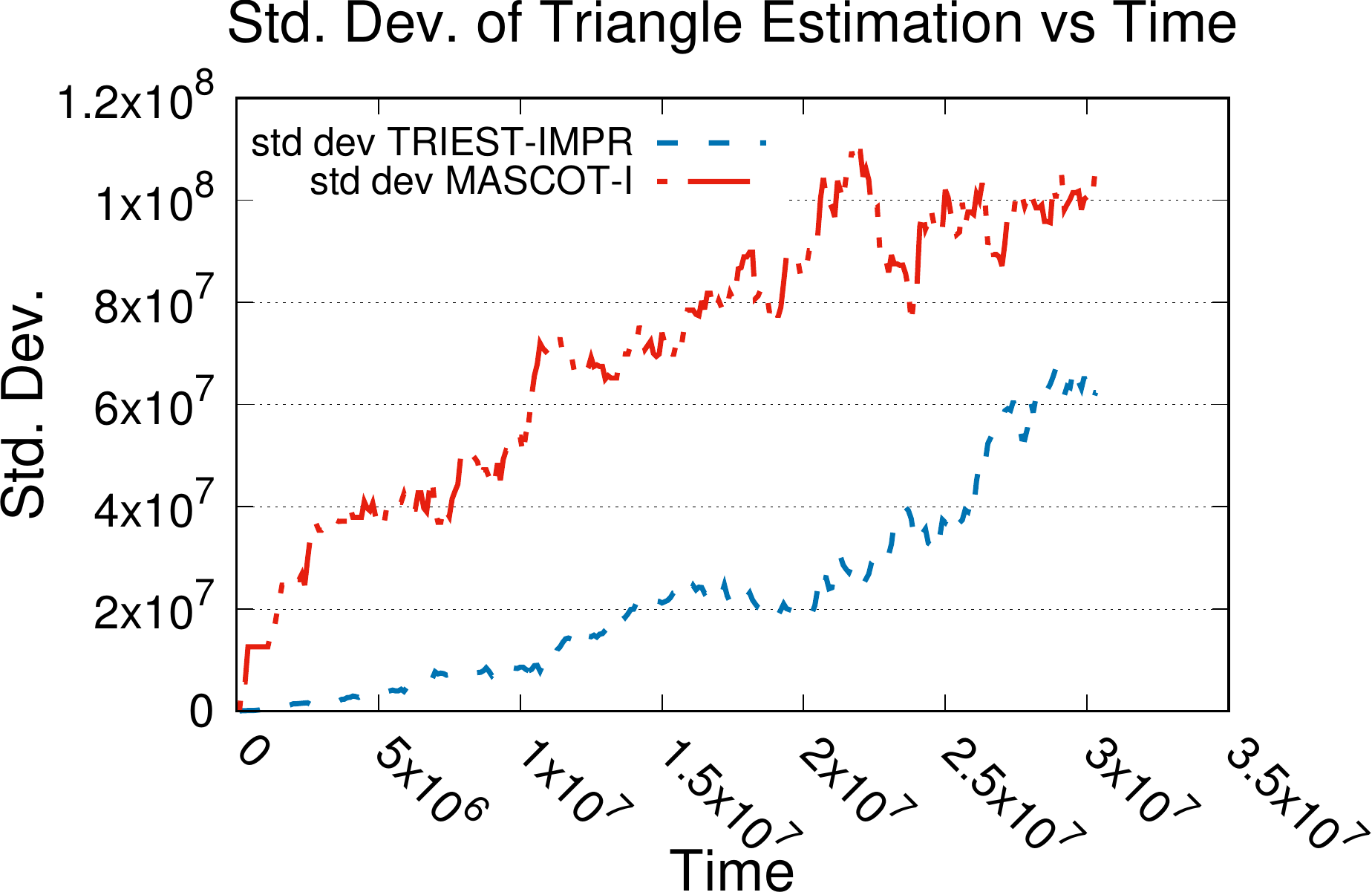}}
  \caption{Accuracy and stability of the estimation of \algoimproved with
    $M=10000$ and of \textsc{mascot-i} with same expected memory, on LastFM,
    over 10 runs. \algoimproved has a smaller standard deviation and moreover
    the max/min estimation lines are closer to the ground truth. Average
    estimations not shown as they are qualitatively similar.}
  \label{fig:variance-improv-vs-mascot-sh}
\end{figure}

\paragraph{Local triangle counting}
 We compare the precision in local triangle count estimation of \algo
 with that of \textsc{mascot}~\citep{LimK15} using the same approach of the
 previous experiment. We can not compare with \textsc{Jha et al.} and
 \textsc{Pavan et al.} algorithms as they provide only global estimation. As
 in~\citep{LimK15}, we measure the Pearson coefficient and the average
 $\varepsilon$ error (see~\citep{LimK15} for definitions). In
 Table~\ref{table:fix-p-vs-res-local-add} we report the Pearson coefficient and
 average $\varepsilon$ error over all timestamps for the smaller
 graphs.\footnote{For efficiency, in this test we evaluate the local number of
 triangles of all nodes every $1000$ edge updates.} \algo (significantly)
 improves (i.e., has higher correlation and lower error) over the
 state-of-the-art \textsc{mascot}, using the same amount of memory.

\begin{table}[ht]
  \tbl{Comparison of the quality of the local triangle estimations between our
    algorithms and the state-of-the-art approach in~\citep{LimK15}. Rows with
    $Y$ in column ``Impr.'' refer to improved algorithms (\algoimproved and
    \textsc{mascot-i}) while those with $N$ to basic algorithms (\algobase and
    \textsc{mascot-c}). In virtually all cases we significantly outperform
    \textsc{mascot} using the same amount of memory.
  }{
    \begin{tabular}{ccccccccc}
      \toprule
      \multicolumn{3}{c}{}& \multicolumn{3}{c}{Avg.~Pearson} &
      \multicolumn{3}{c}{Avg.~$\varepsilon$ Err.}\\
      \cmidrule(l{2pt}r{2pt}){4-6} \cmidrule(l{2pt}r{2pt}){7-9}
      Graph & Impr. &$p$ & \textsc{mascot} & \algo & Change &
      \textsc{mascot} & \algo &Change\\
      \midrule
    \multirow{6}{*}{LastFm}	&\multirow{3}{*}{Y}	&0.1	&0.99&	1.00&	+1.18\%&	0.79&	0.30&	-62.02\%\\
        &		&0.05	&0.97	&1.00&	+2.48\%	&0.99&	0.47&	-52.79\%\\
      &		&0.01&	0.85&	0.98&	+14.28\%	& 1.35&	0.89&	-34.24\%\\
      \cmidrule{2-9}
      &	\multirow{3}{*}{N}	&0.1&	0.97&	0.99&	+2.04\%&	1.08&	0.70&	-35.65\%\\
      &	&	0.05&	0.92	&0.98&	+6.61\%	&1.32	&0.97&	-26.53\%\\
      &		&0.01&	0.32&	0.70&	+117.74\% &1.48&	1.34&		-9.16\%\\
      \midrule
      \multirow{6}{*}{Patent (Cit.)}	&\multirow{3}{*}{Y} 	&0.1&	0.41&	0.82& +99.09\%	&0.62&	0.37&-39.15\%\\
      &&0.05	&0.24	&0.61&	+156.30\%&0.65&	0.51&		-20.78\%\\
      &&	0.01&	0.05&	0.18&	+233.05\%	&0.65&	0.64&	-1.68\%\\
      \cmidrule{2-9}
      &\multirow{3}{*}{N} &	0.1	&0.16	&0.48&	+191.85\%	&0.66&	0.60&	-8.22\%\\
      &&	0.05&	0.06&	0.24&	+300.46\%&	0.67&	0.65&	-3.21\%\\
      &&	0.01&	0.00&	0.003&	+922.02\%&	0.86&	0.68&	-21.02\%\\
      \midrule
      \multirow{6}{*}{Patent (Co-aut.)}	&\multirow{3}{*}{Y} &	0.1	&0.55	&0.87& +58.40\%	&0.86&	0.45&	-47.91\%\\
      &&	0.05&	0.34&	0.71&	+108.80\%&	0.91&	0.63&	-31.12\%\\
      &&0.01	&0.08	&0.26&	+222.84\% &0.96	&0.88&		-8.31\%\\
      \cmidrule{2-9}
      &\multirow{3}{*}{N} &	0.1&	0.25&	0.52&	+112.40\%&	0.92&	0.83&	-10.18\%\\
      &&	0.05&	0.09&	0.28&	+204.98\%	&0.92&	0.92&	0.10\%\\
      &&	0.01&	0.01&	0.03&	+191.46\%&	0.70&	0.84&	20.06\%\\
    \bottomrule
    \end{tabular}
  } 
  \label{table:fix-p-vs-res-local-add}
\end{table}

\paragraph{Memory vs accuracy trade-offs}

We study the trade-off between the sample size $M$ vs the running time and
accuracy of the estimators. Figure~\ref{fig:mape-vs-m-add} shows the tradeoffs
between the accuracy of the estimation (as MAPE) and the size $M$ for the
smaller graphs for which the ground truth number of triangles can be computed
exactly using the na\"ive algorithm. Even with small $M$, \algoimproved achieves
very low MAPE value. As expected, larger $M$ corresponds to higher accuracy and
for the same $M$ \algoimproved outperforms \algobase.

Figure~\ref{fig:time-vs-m-add} shows the average time per update in microseconds
($\mu$s) for \algoimproved as function of $M$. Some considerations on the
running time are in order. First, a larger edge sample (larger $M$) generally
requires longer average update times per operation. This is expected as a larger
sample corresponds to a larger sample graph on which to count triangles. Second,
on average a few hundreds microseconds are sufficient for handling any update
even in very large graphs with billions of edges. Our algorithms can handle
hundreds of thousands of edge updates (stream elements) per second, with very
small error (Fig.~\ref{fig:mape-vs-m-add}), and therefore \algo can be used
efficiently and effectively in high-velocity contexts. The larger average time
per update for Patent (Co-Auth.) can be explained by the fact that the graph is
relatively dense and has a small size (compared to the larger Yahoo!\ and
Twitter graphs). More precisely, the average time per update (for a fixed $M$)
depends on two main factors: the average degree and the length of the stream.
The denser the graph is, the higher the update time as more operations are
needed to update the triangle count every time the sample is modified. On the
other hand, the longer the stream, for a fixed $M$, the lower is the frequency of
updates to the reservoir (it can be show that the expected number of updates
to the reservoir is $O(M(1+\log(\frac{t}{M})))$ which grows sub-linearly in the
size of the stream $t$). This explains why the average update time for the large
and dense Yahoo!\ and Twitter graphs is so small, allowing the algorithm to
scale to billions of updates.

\begin{figure}[ht]
  \subfigure[$M$ vs
  MAPE]{\label{fig:mape-vs-m-add}\includegraphics[width=0.49\textwidth]{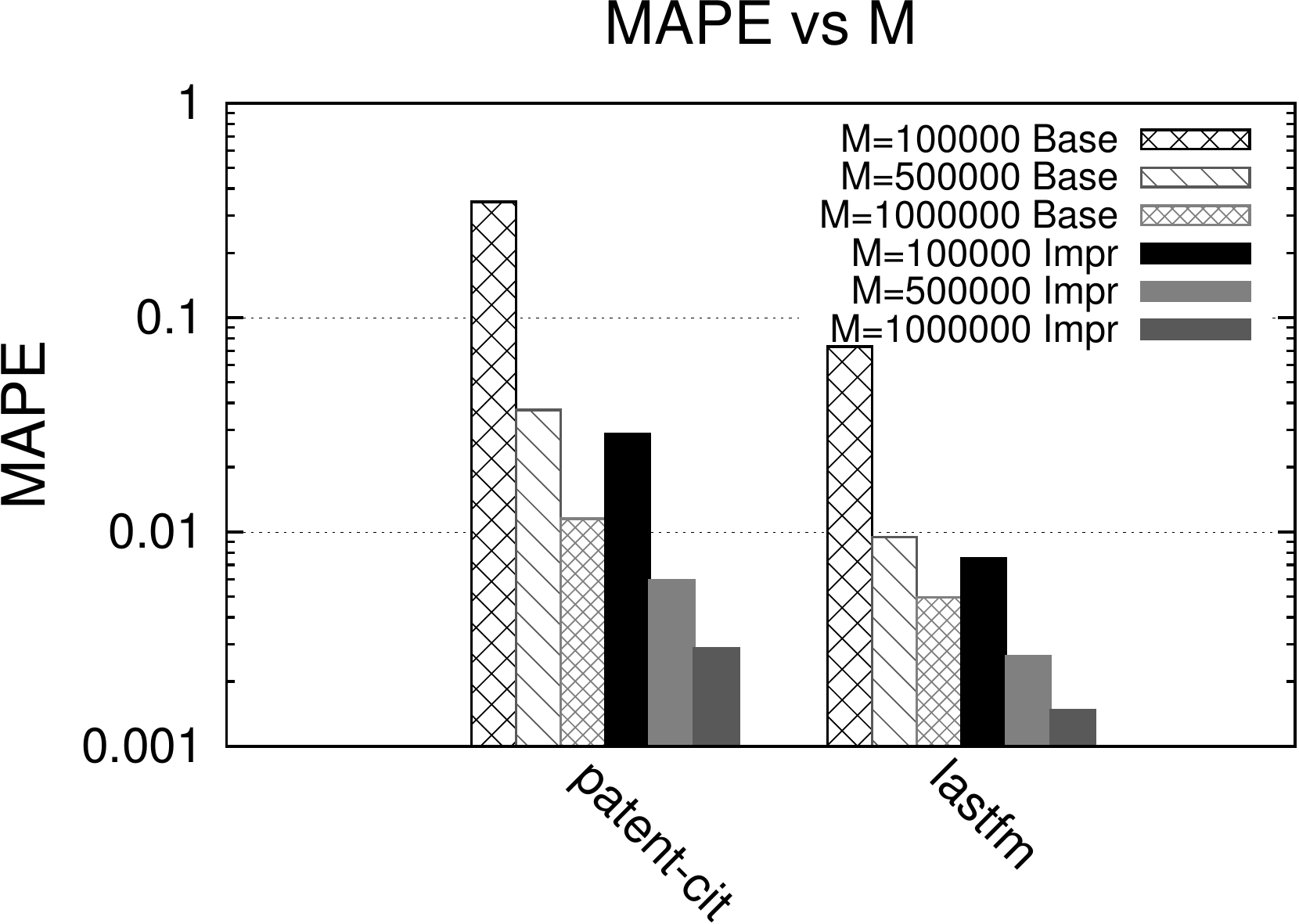}}
  \subfigure[$M$ vs Update
  Time (\algoimproved)]{\label{fig:time-vs-m-add}\includegraphics[width=0.49\textwidth]{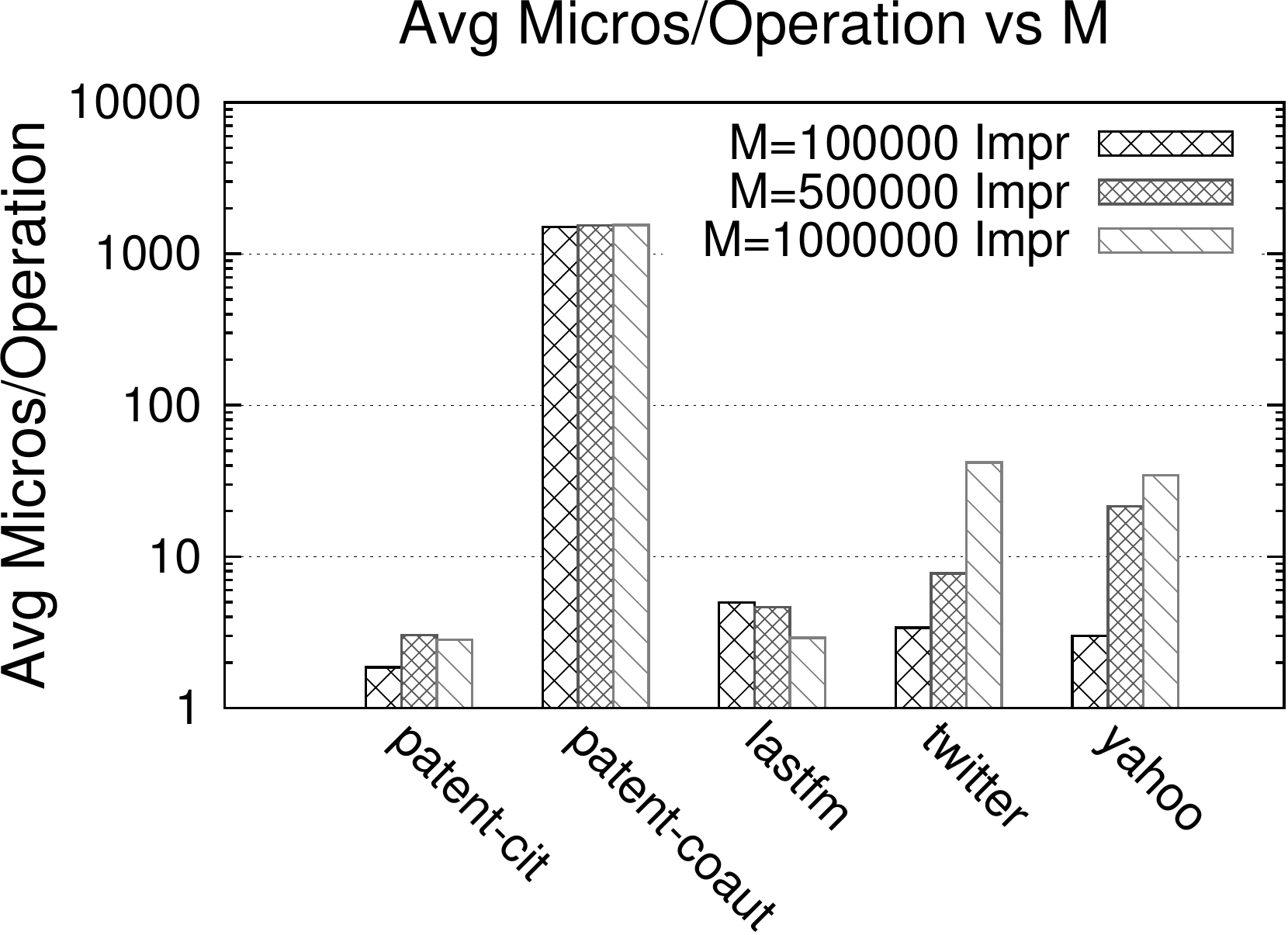}}
  \caption{Trade-offs between $M$ and MAPE and avg.~update time ($\mu$s) -- edge
  insertion only. Larger $M$ implies lower errors but generally higher update times.}
  \label{fig:mape-time-vs-m-add}
\end{figure}

\paragraph{Alternative edge orders}
In all previous experiments the edges are added in their natural order (i.e., in
order of their appearance).\footnote{Excluding Twitter for which we used the
random order, given the lack of timestamps.} While the natural order is the most
important use case, we have assessed the impact of other ordering on the
accuracy of the algorithms. We experiment with both the uniform-at-random
(u.a.r.)~order of the edges and the random BFS order: until all the graph is
explored a BFS is started from a u.a.r.~unvisited node and edges are added in
order of their visit (neighbors are explored in u.a.r.~order). The results for
the random BFS order and u.a.r.~order (Fig.~\ref{fig:mape-all-orders-add}) confirm
that \algo has the lowest error and is very scalable in every tested ordering.

\begin{figure}[ht]
  \centering
  \subfigure[BFS order]{\includegraphics[width=0.49\textwidth]{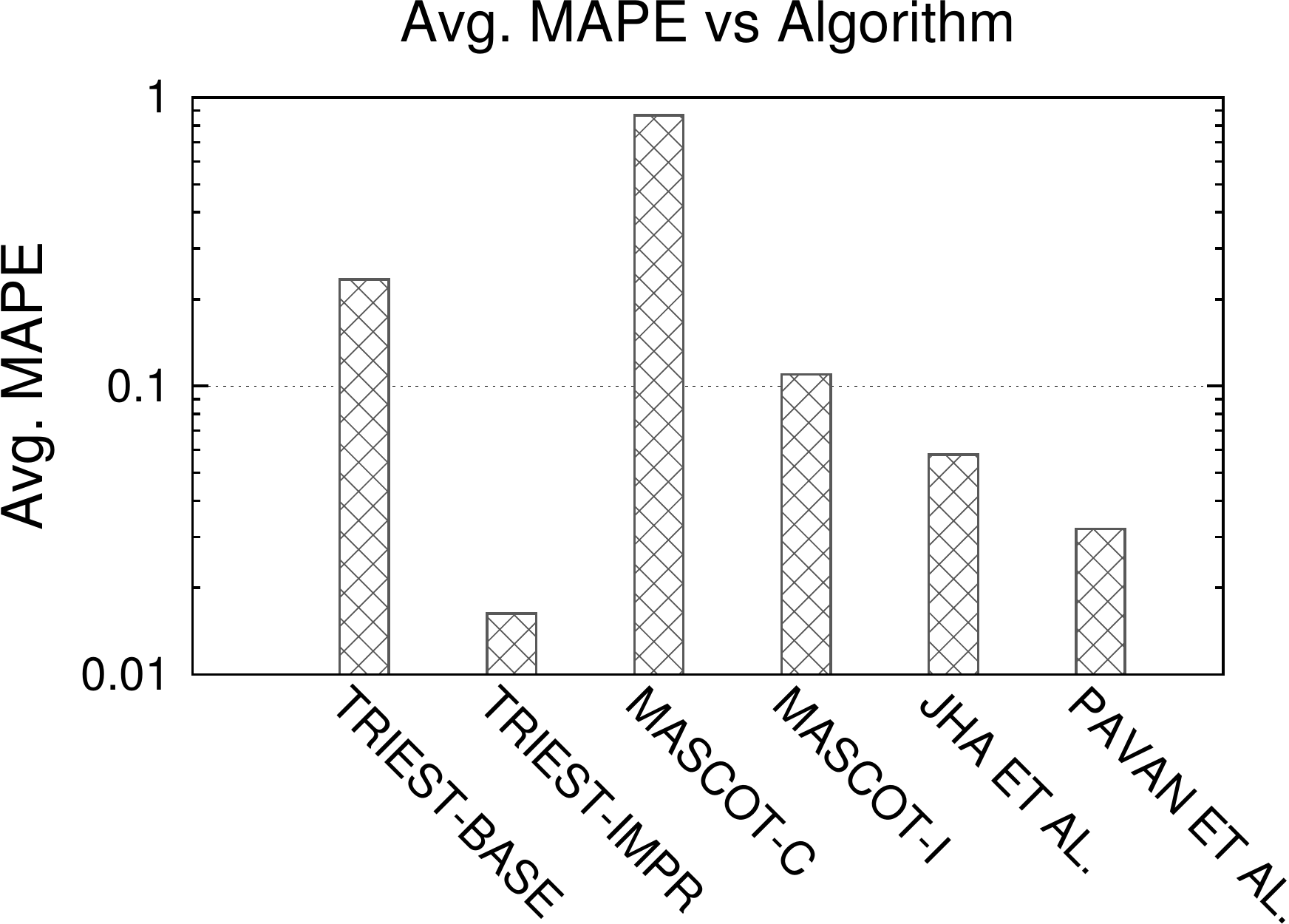}}
  \subfigure[U.a.r. order]{\includegraphics[width=0.49\textwidth]{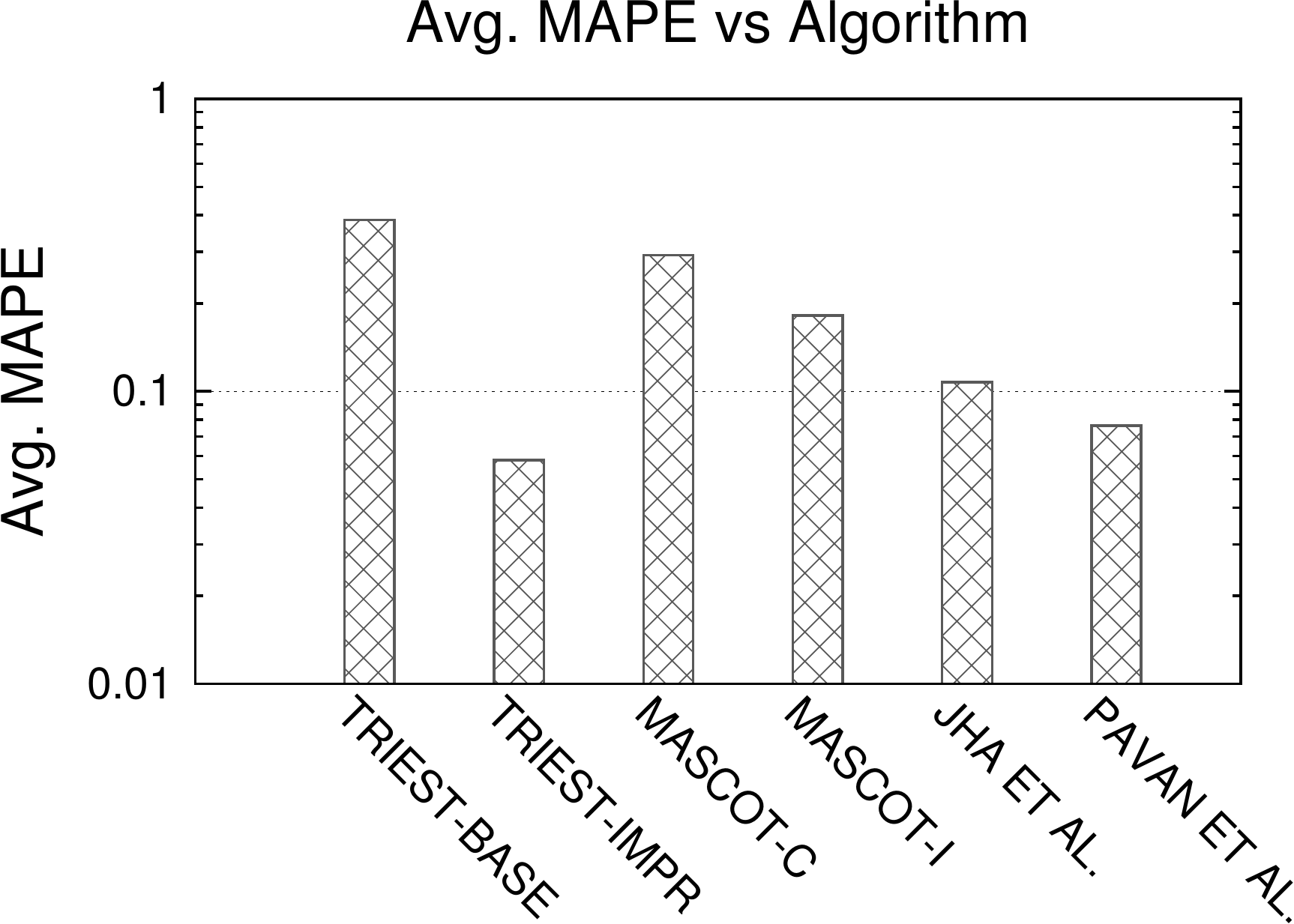}}
  \caption{Average MAPE on Patent (Co-Aut.), with $p=0.01$ (for \textsc{mascot},
    see the main text for how we computed the space used by the other
    algorithms) -- insertion only in Random BFS order and in u.a.r. order. \algoimproved has the
  lowest error.}
\label{fig:mape-all-orders-add}
\end{figure}

\subsection{Fully-dynamic case}
We evaluate \algofd on fully-dynamic streams. We
cannot compare \algofd with the algorithms previously
used~\citep{JhaSP15, PavanTTW13,LimK15} as they only handle insertion-only streams.

In the first set of experiments we model deletions using the widely used \textit{sliding window
model}, where a sliding window of the most recent edges defines the current
graph. The sliding window model is of practical interest as it allows to observe recent trends in the stream. For Patent (Co-Aut.) \& (Cit.) we keep in the sliding window the edges
generated in the last $5$ years, while for LastFm we keep the edges generated in
the last $30$ days. For Yahoo!\ Answers we keep the last $100$ millions edges in
the window\footnote{The sliding window model is not interesting for the Twitter
dataset as edges have random timestamps. We omit the results for Twitter but
\algofd is fast and
has low variance.}.

Figure~\ref{fig:add-rem} shows the evolution of the global number of triangles
in the sliding window model using \algofd using $M=200{,}000$ ($M=1{,}000{,}000$
for Yahoo!\ Answers). The sliding window scenario is significantly more
challenging than the addition-only case (very often the entire sample of edges
is flushed away) but \algofd maintains good variance and scalability even when,
as for LastFm and Yahoo!\ Answers, the global number of triangles
varies quickly.

\begin{figure}[ht]
\subfigure[Patent (Co-Aut.)]{\includegraphics[width=0.49\textwidth]{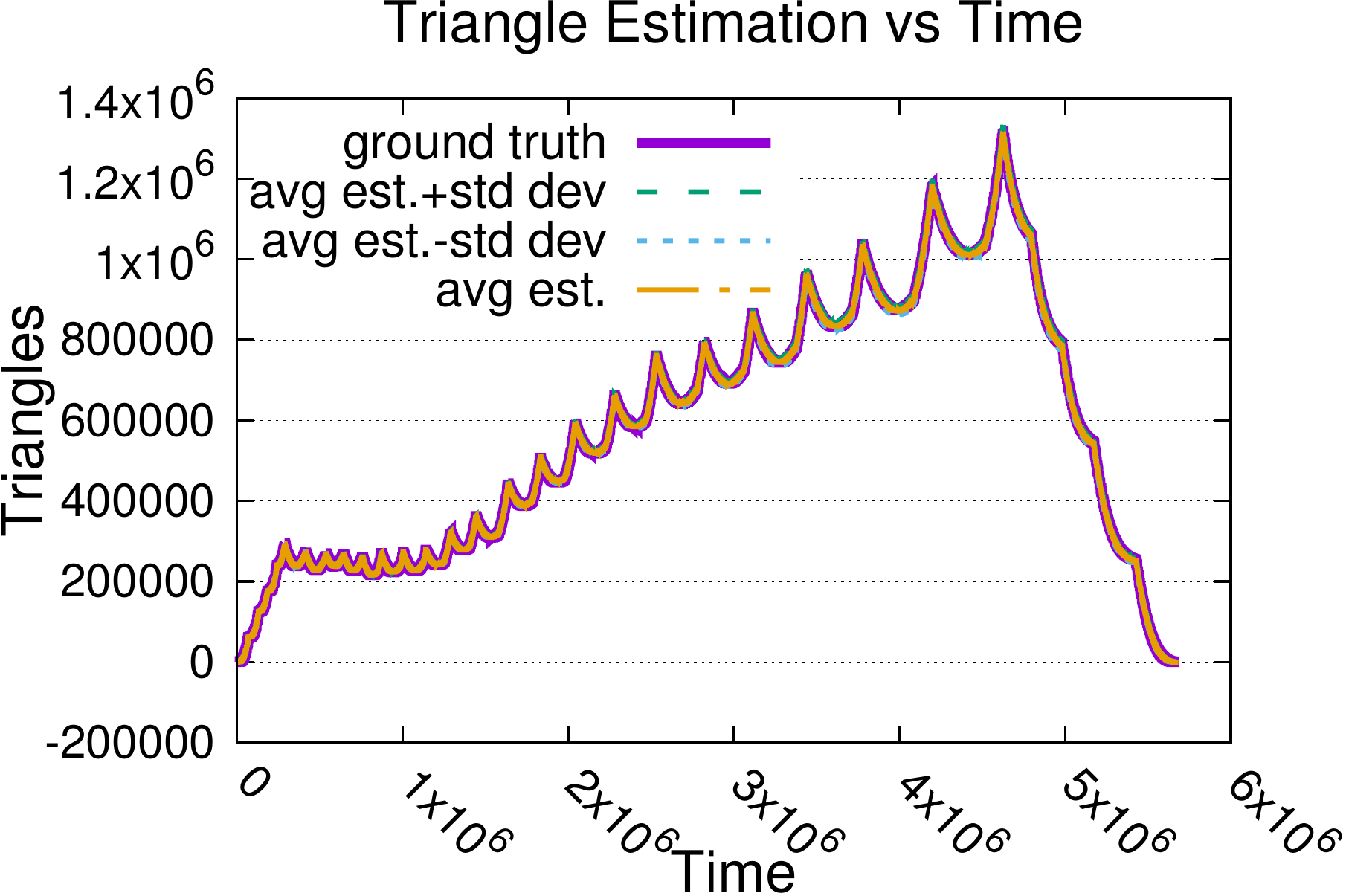}}
\subfigure[Patent (Cit.)]{\includegraphics[width=0.49\textwidth]{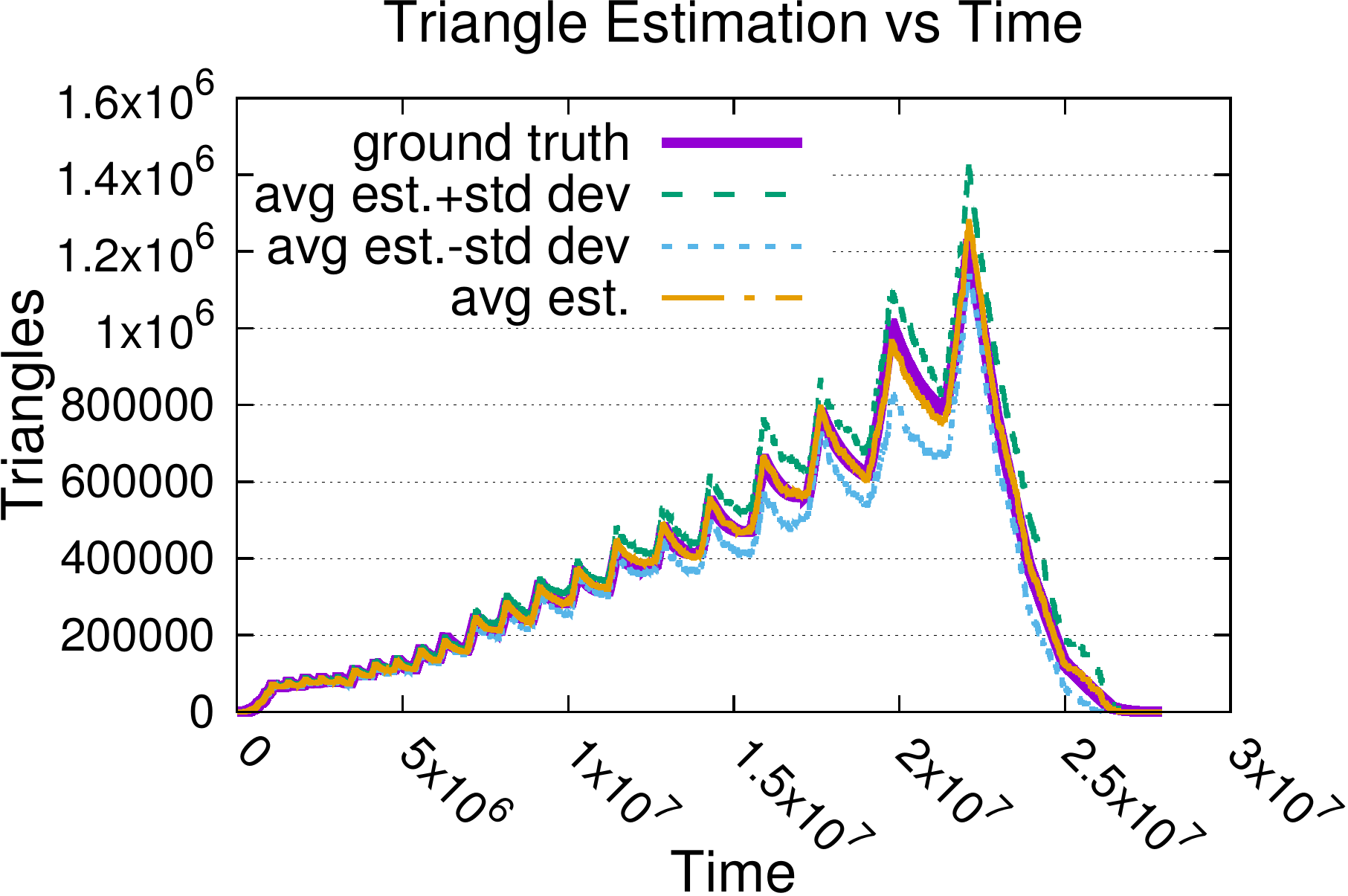}}
\subfigure[LastFm]{\includegraphics[width=0.49\textwidth]{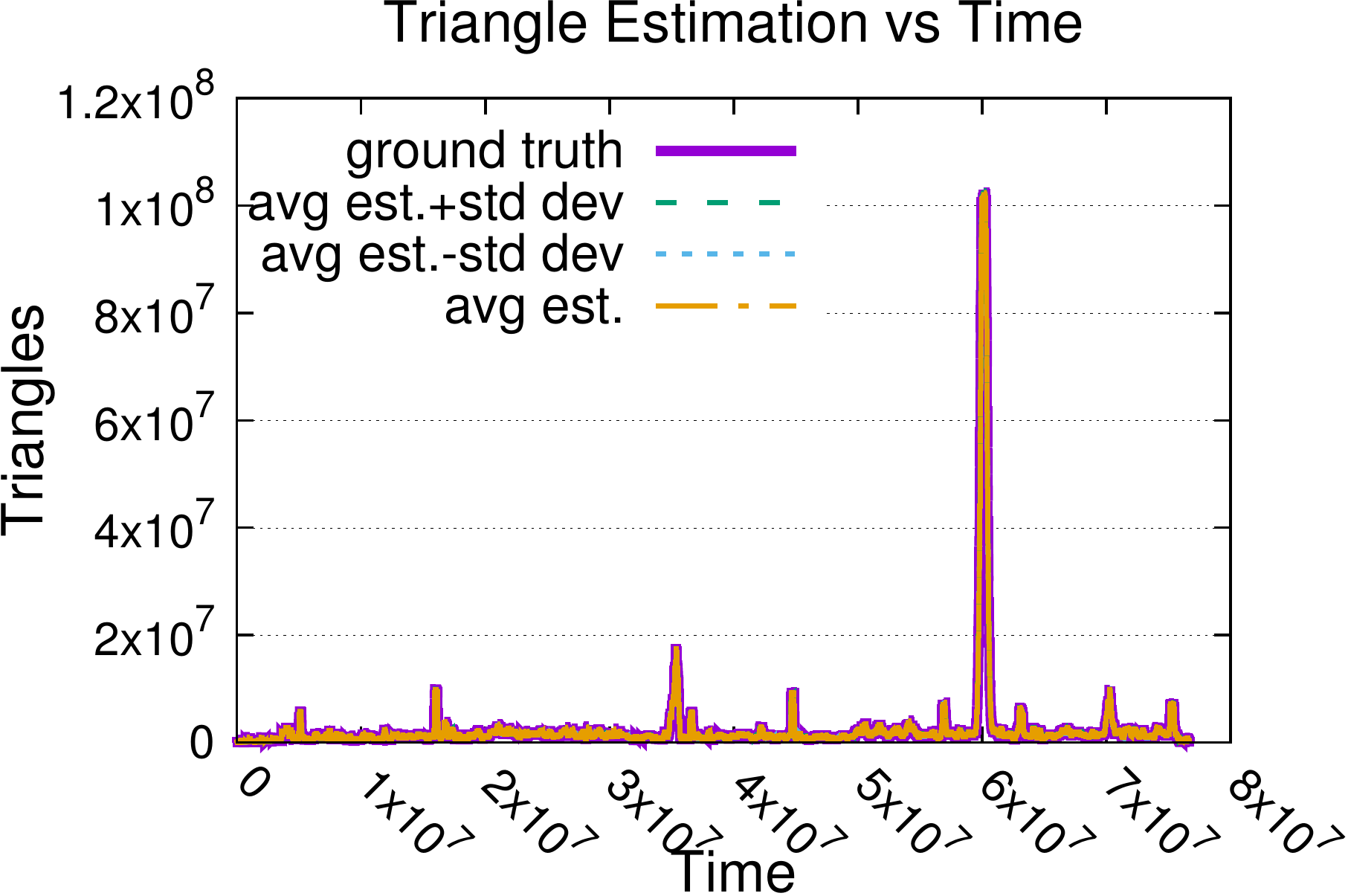}}
\subfigure[Yahoo!\ Answers]{\includegraphics[width=0.49\textwidth]{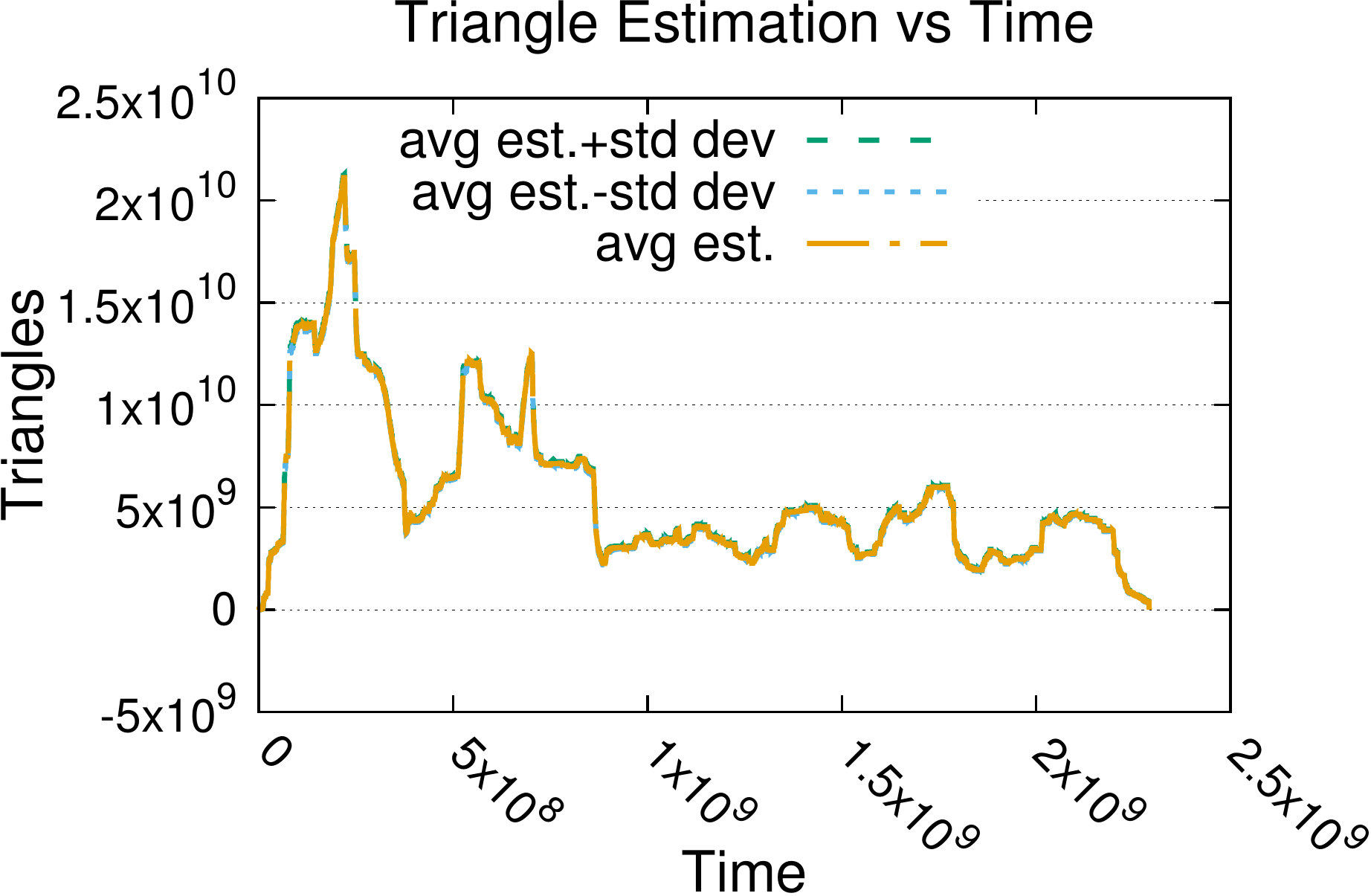}}
\caption{Evolution of the global number of triangles in the fully-dynamic case
  (sliding window model for edge deletion). The curves are
  \emph{indistinguishable on purpose}, to remark the fact that \algofd
  estimations are extremely accurate and consistent. We comment on the observed
  patterns in the text.}
\label{fig:add-rem}
\end{figure}

Continuous monitoring of triangle counts with \algofd allows to detect patterns
that would otherwise be difficult to notice. For LastFm
(Fig.~\ref{fig:add-rem}(c)) we observe a sudden
spike of several order of magnitudes. The dataset is anonymized so we cannot
establish which songs are responsible for this spike. In Yahoo!\ Answers
(Fig.~\ref{fig:add-rem}(d)) a
popular topic can create a sudden (and shortly lived) increase in the number of
triangles, while the evolution of the Patent co-authorship and co-citation
networks is slower, as the creation of an edge requires filing a patent
(Fig.~\ref{fig:add-rem}(a) and (b)). The almost constant increase over
time\footnote{The decline at the end is due to the removal of the last edges
from the sliding window after there are no more edge additions.} of the number
of triangles in Patent graphs is consistent with previous observations of {\it
densification} in collaboration networks as in the case of nodes'
degrees~\citep{leskovec2007graph} and the observations on the density of the
densest subgraph~\citep{epasto2015efficient}.

\begin{figure}[ht]
  \centering
  \includegraphics[width=0.6\textwidth]{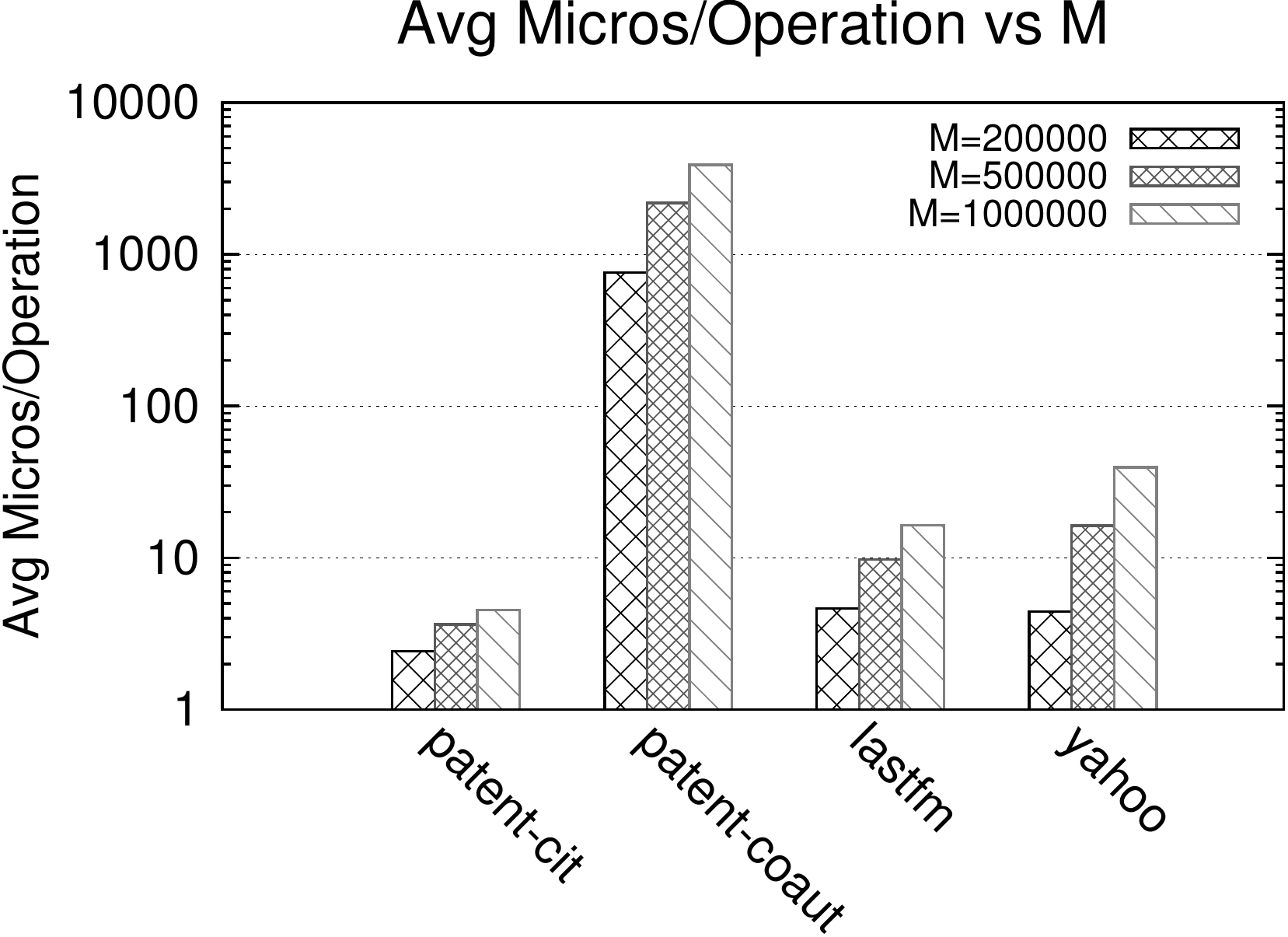}
  \caption{Trade-offs between the avg.~update time ($\mu$s) and $M$ for \algofd.}
  \label{fig:time-vs-m-add-rem}
\end{figure}

Table~\ref{table:res-global-local-add-rem} shows the results for both the local
and global triangle counting estimation provided by \algofd. In this case we
can not compare with previous works, as they only handle insertions. It is
evident that precision improves with $M$ values, and even relatively small $M$
values result in a low MAPE (global
estimation), high Pearson correlation and low $\varepsilon$ error (local
estimation). Figure~\ref{fig:time-vs-m-add-rem} shows the tradeoffs between
memory (i.e., accuracy) and time. In all cases our algorithm is very fast and it
presents update times in the order of hundreds of microseconds for datasets with
billions of updates (Yahoo!\ Answers).

\begin{table}[ht]
  \tbl{Estimation errors for \algofd.
  }{
    \begin{tabular}{crccc}
      \toprule
      \multicolumn{2}{c}{}&\multicolumn{1}{c}{Avg. Global}&\multicolumn{2}{c}{Avg. Local}\\
      \cmidrule(l{2pt}r{2pt}){3-3} \cmidrule(l{2pt}r{2pt}){4-5}
      Graph & \multicolumn{1}{c}{$M$} & MAPE & Pearson & $\varepsilon$ Err. \\
      \midrule
      \multirow{2}{*}{LastFM}&200000&0.005&0.980&0.020\\
      &1000000&0.002&0.999&0.001\\
      \midrule
      \multirow{2}{*}{Pat. (Co-Aut.)}&200000&0.010&0.660&0.300\\
      &1000000&0.001&0.990&0.006\\
      \midrule
      \multirow{2}{*}{Pat. (Cit.)}&200000&0.170&0.090&0.160\\
      &1000000&0.040&0.600&0.130\\
      \bottomrule
    \end{tabular}
  } 
  \label{table:res-global-local-add-rem}
\end{table}

\paragraph{Alternative models for deletion}
We evaluate \algofd using other models for deletions than the sliding
window model. To assess the resilience of the algorithm to massive deletions we
run the following experiments. We added edges in their natural order but each edge
addition is followed with probability $q$ by a mass deletion event where each
edge currently in the the graph is deleted with probability $d$ independently.
We run experiments with $q = 3{,}000{,}000^{-1}$ (i.e., a mass deletion expected
every $3$ millions edges) and $d=0.80$ (in expectation $80\%$ of edges are
deleted).
The results are shown in Table~\ref{table:res-global-local-add-rem-mass-deletion}.

\begin{table}[ht]
  \tbl{Estimation errors for \algofd -- mass deletion experiment, $q =
  3{,}000{,}000^{-1}$ and $d=0.80$.
  }{
    \begin{tabular}{crccc}
      \toprule
      \multicolumn{2}{c}{}&\multicolumn{1}{c}{Avg. Global}&\multicolumn{2}{c}{Avg. Local}\\
      \cmidrule(l{2pt}r{2pt}){3-3} \cmidrule(l{2pt}r{2pt}){4-5}
      Graph & \multicolumn{1}{c}{$M$} & MAPE & Pearson & $\varepsilon$ Err. \\
      \midrule
      \multirow{2}{*}{LastFM}&200000&0.040&0.620&0.53\\
      &1000000&0.006&0.950&0.33\\
      \midrule
      \multirow{2}{*}{Pat. (Co-Aut.)}&200000&0.060&0.278&0.50\\
      &1000000&0.006&0.790&0.21\\
      \midrule
      \multirow{2}{*}{Pat. (Cit.)}&200000&0.280&0.068&0.06\\
      &1000000&0.026&0.510&0.04\\
      \bottomrule
    \end{tabular}
  } 
  \label{table:res-global-local-add-rem-mass-deletion}
\end{table}
We observe that \algofd maintains a good accuracy and scalability even in
face of a massive (and unlikely) deletions of the vast majority of the edges:
e.g., for LastFM with $M=200000$ (resp. $M=1{,}000{,}000$) we observe
$0.04$ (resp. $0.006$) Avg. MAPE.

\section{Conclusions}\label{sec:concl}
We presented \algo, the first suite of algorithms that use reservoir sampling
and its variants to continuously maintain unbiased, low-variance estimates of
the local and global number of triangles in fully-dynamic graphs streams of
arbitrary edge/vertex insertions and deletions using a fixed, user-specified
amount of space. Our experimental evaluation shows that \algo outperforms
state-of-the-art approaches and achieves high accuracy on real-world datasets
with more than one billion of edges, with update times of hundreds of
microseconds.

\appendix
\section*{APPENDIX}

\section{Additional theoretical results}\label{sec:appendix}
In this section we present the theoretical results (statements and proofs) not
included in the main body.

\subsection{Theoretical results for \algobase}\label{app:algobase}

Before proving Lemma~\ref{lem:reservoirhighorder}, we need to introduce the
following lemma, which states a well known property of the reservoir sampling
scheme.
\begin{lemma}[{\citep[Sect.~2]{Vitter85}}]\label{lem:reservoir}
  For any $t>M$, let $A$ be any subset of $E^{(t)}$ of size $|A|=M$. Then, at
  the end of time step $t$,
  \[
    \Pr(\Sam=A)=\frac{1}{\binom{|E^{(t)}|}{M}}=\frac{1}{\binom{t}{M}},
  \]
  i.e., the set of edges in $\Sam$ at the end of time $t$
  is a subset of $E^{(t)}$ of size $M$ chosen \emph{uniformly} at random from
  all subsets of $E^{(t)}$ of the same size.
\end{lemma}

\begin{proof}[of Lemma~\ref{lem:reservoirhighorder}]
  If $k>\min\{M,t\}$, we have $\Pr(B\subseteq\Sam)=0$ because it is impossible
  for $B$ to be equal to $\Sam$ in these cases. From now on
  we then assume $k\le\min\{M,t\}$.

  If $t\le M$, then $E^{(t)}\subseteq\Sam$ and $\Pr(B\subseteq\Sam)=1=\xi_{k,t}^{-1}$.

  Assume instead that $t>M$, and let $\mathcal{B}$ be the family of subsets of
  $E^{(t)}$ that {\em 1.}~have size $M$, and {\em 2.}~contain
  $B$:
  \[
    \mathcal{B}=\{C\subset E^{(t)} ~:~ |C|=M, B\subseteq C\}\enspace.
  \]
  We have
  \begin{equation}
    |\mathcal{B}|=\binom{|E^{(t)}|-k}{M-k}=\binom{t-k}{M-k}\enspace.
  \end{equation}
  From this and and Lemma~\ref{lem:reservoir} we then have
  \begin{align*}
    \Pr(B\subseteq\Sam)&=\Pr(\Sam\in\mathcal{B})=\sum_{C\in\mathcal{B}}\Pr(\Sam=C)\\
    &=\frac{\binom{t-k}{M-k}}{\binom{t}{M}}=\frac{\binom{t-k}{M-k}}{\binom{t-k}{M-k}\prod_{i=0}^{k-1}\frac{t-i}{M-i}}=\prod_{i=0}^{k-1}\frac{M-i}{t-i}=\xi_{k,t}^{-1}\enspace.
  \end{align*}
\end{proof}

\subsubsection{Expectation}

\begin{proof}[of Lemma~\ref{lem:baseunbiasedaux}]
  We only show the proof for $\tau$, as the proof for the local counters
  follows the same steps.

  The proof proceeds by induction. The thesis is true after the first call to
  \textsc{UpdateCounters} at time $t=1$. Since only one edge is in $\Sam$ at
  this point, we have $\Delta^\Sam=0$, and $\neigh^\Sam_{u,v}=\emptyset$, so
  \textsc{UpdateCounters} does not modify $\tau$, which was initialized to
  $0$. Hence $\tau=0=\Delta^\Sam$.

  Assume now that the thesis is true for any subsequent call to
  \textsc{UpdateCounters} up to some point in the execution of the algorithm
  where an edge $(u,v)$ is inserted or removed from
  $\Sam$. We now show that the thesis is still true after the call to
  \textsc{UpdateCounters} that follows this change in $\Sam$. Assume
  that $(u,v)$ was \emph{inserted} in $\Sam$ (the proof for the case of an edge
  being removed from $\Sam$ follows the same steps). Let
  $\Sam^\mathrm{b}=\Sam\setminus\{(u,v)\}$ and $\tau^\mathrm{b}$ be the value of
  $\tau$ before the call to \textsc{UpdateCounters} and, for any $w\in
  V_{\Sam^\mathrm{b}}$, let $\tau_w^\mathrm{b}$ be the value of $\tau_w$
  before the call to \textsc{UpdateCounters}. Let $\Delta^{\Sam}_{u,v}$ be
  the set of triangles in $G_\Sam$ that have $u$ and $v$ as corners. We need
  to show that, after the call, $\tau = |\Delta^{\Sam}|$. Clearly we have
  $\Delta^{\Sam}=\Delta^{\Sam^{\mathrm{b}}} \cup \Delta^{\Sam}_{u,v}$ and
  $\Delta^{\Sam^{\mathrm{b}}} \cap \Delta^{\Sam}_{u,v} = \emptyset$, so
  \[
    |\Delta^{\Sam}|=|\Delta^{\Sam^{\mathrm{b}}}| + |\Delta^{\Sam}_{u,v}|
  \]
  We have $|\Delta^{\Sam}_{u,v}|=|\neigh^\Sam_{u,v,}|$ and, by the inductive
  hypothesis, we have that $\tau^\mathrm{b}=|\Delta^{\Sam^{\mathrm{b}}}|$.
  Since \textsc{UpdateCounters} increments $\tau$ by $|\neigh^\Sam_{u,v,}|$,
  the value of $\tau$ after \textsc{UpdateCounters} has completed is exactly
  $|\Delta^\Sam|$.
\end{proof}

We can now prove Thm.~\ref{thm:baseunbiased} on the unbiasedness of the
estimation computed by \algobase (and on their exactness for $t\le M$).

\begin{proof}[of Thm.~\ref{thm:baseunbiased}]
  We prove the statement for the estimation of global triangle count. The
  proof for the local triangle counts follows the same steps.

  If $t\le M$, we have $G_\Sam=G^{(t)}$ and from
  Lemma~\ref{lem:baseunbiasedaux} we have
  $\tau^{(t)}=|\Delta^{\Sam}|=|\Delta^{(t)}|$, hence the thesis holds.

  Assume now that $t > M$, and assume that $|\Delta^{(t)}|>0$, otherwise, from
  Lemma~\ref{lem:baseunbiasedaux}, we have $\tau^{(t)}=|\Delta^{\Sam}|=0$ and
  \algobase estimation is deterministically correct. Let
  $\lambda=(a,b,c)\in\Delta^{(t)}$, (where $a,b,c$ are edges in
  $E^{(t)}$) and let $\delta_\lambda^{(t)}$ be a random variable that takes
  value $\xi^{(t)}$ if $\lambda\in\Delta_\Sam$
  (i.e., $\{a,b,c\}\subseteq\Sam$) at the end of the step instant $t$, and $0$
  otherwise. From
  Lemma~\ref{lem:reservoirhighorder}, we have that
  \begin{equation}\label{eq:baseunbiasedexpectation}
    \mathbb{E}\left[\delta_\lambda^{(t)}\right]=\xi^{(t)}\Pr(\{a,b,c\}\subseteq\Sam)=\xi^{(t)}\frac{1}{\xi_{3,t}}=\xi^{(t)}\frac{1}{\xi^{(t)}}=1\enspace.
  \end{equation}
  We can write
  \[
    \xi^{(t)}\tau^{(t)}=\sum_{\lambda \in \Delta^{(t)}}
    \delta_{\lambda}^{(t)}
  \]
  and from this, \eqref{eq:baseunbiasedexpectation}, and linearity of expectation,
  we have
  \[
    \mathbb{E}\left[\xi^{(t)}\tau^{(t)}\right] = \sum_{\lambda
    \in \Delta^{(t)}}\mathbb{E}\left[\delta_{\lambda}^{(t)} \right] =
    |\Delta^{(t)}|\enspace.
  \]
\end{proof}

\subsubsection{Concentration}

\begin{proof}[of Lemma~\ref{lem:equivMp}]
  Using the law of total probability, we have
  \begin{align}
    \Pr\left(f(\Sam_\textsc{in}) = 1\right) &=\sum_{k=0}^t
    \Pr\left(f(\Sam_\textsc{in}) = 1~\left|~|\Sam_\textsc{in}|=k\right.\right)\Pr(|\Sam_\textsc{in}|=k)\nonumber\\
    &\ge \Pr\left(f(\Sam_\textsc{in}) = 1~\left|~
    |\Sam_\textsc{in}|=M\right.\right)\Pr(|\Sam_\textsc{in}|=M)\nonumber\\
    &\ge \Pr\left(f(\Sam_\textsc{mix}) = 1\right) \Pr\left(|\Sam_\textsc{in}|= M\right),\label{eq:boundrareevent}
  \end{align}
  where the last inequality comes from Lemma~\ref{lem:reservoir}: the set of
  edges included in $\Sam_\textsc{mix}$ is a uniformly-at-random subset of $M$
  edges from $E^{(t)}$, and the same holds for $\Sam_\textsc{in}$ when
  conditioning its size being $M$.

  Using the Stirling approximation $\sqrt{2\pi n}(\frac{n}{e})^n\le n!\le
  e\sqrt{n}(\frac{n}{e})^n$ for any positive integer $n$, we have
  \begin{align*}
    \Pr\left(|\Sam_\textsc{in}| = M\right) &= \binom{t}{M} \left(\frac{M}{t}\right)^M  \left(\frac{t-M}{t}\right)^{t-M} \\
    &\ge \frac{t^t \sqrt{t} \sqrt{2\pi} e^{-t}}{e^2 \sqrt{M} \sqrt{t-M} e^{-t} M^M (t-M)^{t-M}} \frac{M^M(t-M)^{t-M}}{t^t}\\
    &\ge \frac{1}{e\sqrt{M}}\enspace.
  \end{align*}
  Plugging this into~\eqref{eq:boundrareevent} concludes the proof.
\end{proof}

\begin{fact}\label{fact:first}
  For any $x>2$, we have
  \[
    \frac{x^2}{(x-1)(x-2)} \le 1+\frac{4}{x-2}\enspace.
  \]
\end{fact}

\begin{proof}[of Lemma~\ref{lem:estimatorsrelationship}]
  We start by looking at the ratio between
  $\frac{t(t-1)(t-2)}{M(M-1)(M-2)}$ and $(t/M)^3$. We have:
  \begin{align*}
    1\le \frac{t(t-1)(t-2)}{M(M-1)(M-2)}\left(\frac{M}{t}\right)^3 &=
    \frac{M^2}{(M-1)(M-2)} \frac{(t-1)(t-2)}{t^2}\\
    &\le \frac{M^2}{(M-1)(M-2)}\\
    &\le 1+ \frac{4}{M-2}
  \end{align*}
  where the last step follows from Fact~\ref{fact:first}. Using this, we
  obtain
  \begin{align*}
    \left|\phi^{(t)} - \phi_\textsc{mix}^{(t)}\right| &=
    \left|\tau^{(t)}\frac{t(t-1)(t-2)}{M(M-1)(M-2)} - \tau^{(t)} \left(\frac{t}{M}\right)^3\right|\\
    &=\left|\tau^{(t)} \left(\frac{t}{M}\right)^3 \left
    (\frac{t(t-1)(t-2)}{M(M-1)(M-2)}\left(\frac{M}{t}\right)^3  -
    1\right)\right|  \\
    &\le \tau^{(t)}\left(\frac{t}{M}\right)^3 \frac{4}{M-2}\\
    &= \phi_\textsc{mix}^{(t)}\frac{4}{M-2}\enspace.
  \end{align*}
\end{proof}

\subsubsection{Variance comparison}
We now prove Lemma~\ref{lem:variancecomparison}, about the fact that the
variance of the estimations computed by \algobase is smaller, for most of the
stream, than the variance of the estimations computed by
\textsc{mascot-c}~\citep{LimK15}. We first need the following technical fact.

\begin{fact}\label{fact:second}
  For any $x>42$, we have
  \[
    \frac{x^2}{(x-3)(x-4)} \le 1+\frac{8}{x}\enspace.
  \]
\end{fact}

\begin{proof}[of Lemma~\ref{lem:variancecomparison}]
  We focus on $t>M> 42$ otherwise the theorem is immediate.
  We show that for such conditions $f(M,t) < \bar{f}(M/T)$ and  $g(M,t) <
  \bar{g}(M/T)$. Using the fact that $t\le\alpha T$ and Fact~\ref{fact:first}, we have
  \begin{align}
    f(M,t) - \bar{f}(M/T) & = \frac{t(t-1)(t-2)}{M(M-1)(M-2)} -\frac{T^3}{M^3} \nonumber \\
    & < \frac{\alpha^3 T^3}{M^3}\frac{M^2}{(M-1)(M-2)} -\frac{T^3}{M^3} \nonumber \\
    & \le \frac{\alpha^3 T^3}{M^3} \left(1 + \frac{4}{M-2}\right) - \frac{T^3}{M^3}\nonumber\\
    & \le \frac{T^3}{M^3}\left(\alpha^3 +\frac{4\alpha^3 }{M-2}-1\right)\enspace.
    \label{eq:inequalityf}
  \end{align}
  Given that $T$ and $M$ are $\ge 42$, the r.h.s.~of~\eqref{eq:inequalityf} is
  non-positive iff
  \[
    \alpha^3+\frac{4\alpha^3}{M-2}-1\le 0\enspace.
  \]
  Solving for $M$ we have that the above is verified when
  $M\ge\frac{4\alpha^3}{1-\alpha^3}+2$. This is always true given our assumption
  that $M>\max(\frac{8\alpha}{1-\alpha}, 42)$: for any $0<\alpha <0.6$, we have
  $\frac{4\alpha^3}{1-\alpha^3}+2 < 42\le M$ and for any $0.6\le \alpha <1$ we have
  $\frac{4\alpha^3}{1-\alpha^3}+2 < \frac{8\alpha}{1-\alpha}\le M$. Hence the
  r.h.s.~of~\eqref{eq:inequalityf} is $\le 0$ and $f(M,t)< \bar{f}(M/T)$.

  We also have:
  \begin{align}
    g(M,t) -\bar{g}(M/T) & = \frac{t(t-1)(t-2)(M-3)(M-4)}{(t-3)(t-4)M(M-1)(M-2)} -\frac{T}{M} \nonumber \\
    & < \frac{t}{M}\frac{t^2}{(t-3)(t-4)} -\frac{T}{M} \nonumber \\
    & \le \frac{t}{M}\left(1+\frac{8}{t}\right) -\frac{T}{M},\label{eq:inequalityg}
  \end{align}
  where the last inequality follow from Fact~\ref{fact:second}, since $t>M>42$.
  Now, from~\eqref{eq:inequalityg} since $t\le \alpha T$ and $t> M$, we can write:
  \[
    g(M,t) - \bar{g}(M/T) <
    \frac{T}{M}\left(\alpha+\frac{8\alpha}{M}-1\right)\enspace.
  \]
  The r.h.s.~of this equation is non-positive given the assumption
  $M>\frac{8\alpha}{1-\alpha}$, hence $g(M,t)<\bar{g}(M/T)$.
\end{proof}

\subsection{Theoretical results for \algoimproved}\label{app:algoimproved}

\subsubsection{Expectation}
\begin{proof}[of Thm.~\ref{thm:improvedunbiased}]

  If $t\le M$ \algoimproved behaves exactly like \algobase, and the statement
  follows from Lemma~\ref{thm:baseunbiased}.

  Assume now $t>M$ and assume that $|\Delta^{(t)}|>0$, otherwise, the algorithm
  deterministically returns $0$ as an estimation and the thesis follows. Let
  $\lambda\in\Delta^{(t)}$ and denote with $a$, $b$, and $c$ the edges of
  $\lambda$ and assume, w.l.o.g., that they appear in this
  order (not necessarily consecutively) on the stream. Let $t_\lambda$ be the
  time step at which $c$ is on the stream. Let $\delta_\lambda$ be a random
  variable that takes value $\xi_{2,t_\lambda-1}$ if $a$ and $b$
  are in $\Sam$ at the end of time step $t_\lambda-1$, and $0$ otherwise.
  Since it must be $t_\lambda-1\ge 2$, from Lemma~\ref{lem:reservoirhighorder}
  we have that \begin{equation}\label{eq:improvedprob}
    \Pr\left(\delta_\lambda=\xi_{2,t_\lambda-1}\right)=\frac{1}{\xi_{2,t_\lambda-1}}\enspace.
  \end{equation}
  When $c=(u,v)$ is on the stream, i.e., at time $t_\lambda$, \algoimproved
  calls \textsc{UpdateCounters} and increments the counter $\tau$ by
  $|\mathcal{N}^\Sam_{u,v}|\xi_{2,t_\lambda-1}$, where
  $|\mathcal{N}^\Sam_{u,v}|$ is the number of triangles with $(u,v)$ as an
  edge in $\Delta^{\Sam\cup\{c\}}$. All these triangles have the corresponding
  random variables taking the same value $\xi_{2,t_\lambda-1}$. This
  means that the random variable $\tau^{(t)}$ can be expressed as
  \[
    \tau^{(t)}=\sum_{\lambda\in\Delta^{(t)}}\delta_\lambda\enspace.
  \]
  From this, linearity of expectation, and~\eqref{eq:improvedprob}, we get
  \[
    \mathbb{E}\left[\tau^{(t)}\right]=\sum_{\lambda\in\Delta^{(t)}}
    \mathbb{E}[\delta_\lambda]
    =\sum_{\lambda\in\Delta^{(t)}}
    \xi_{2,t_\lambda-1}\Pr\left(\delta_\lambda=\xi_{2,t_\lambda-1}\right)=\sum_{\lambda\in\Delta^{(t)}}\xi_{2,t_\lambda-1}\frac{1}{\xi_{2,t_\lambda-1}}=|\Delta^{(t)}|\enspace.
  \]
\end{proof}

\subsubsection{Variance}

\begin{proof}[of Lemma~\ref{lem:negativedependence}]
  Consider first the case where \emph{all} edges of $\lambda$ appear on the
  stream before \emph{any} edge of $\gamma$, i.e.,
  \[
    t_{\ell_1}<t_{\ell_2}<t_{\ell_3}<t_{g_1}<t_{g_2}<t_{g_3}\enspace.
  \]
  The presence or absence of either or both $\ell_1$ and $\ell_2$ in $\Sam$ at
  the beginning of time step $t_{\ell_3}$ (i.e., whether $D_\lambda$
  happens or not) has no effect whatsoever on the probability that $g_1$ and
  $g_2$ are in the sample $\Sam$ at the beginning of time step $t_{g_3}$.
  Hence in this case,
  \[
    \Pr(D_\gamma~|~D_\lambda)=\Pr(D_\gamma)\enspace.
  \]
  Consider now the case where, for any $i\in\{1,2\}$, the edges
  $g_1,\dotsc,g_i$ appear on the stream before $\ell_3$ does.
  Define now the events
  \begin{itemize}
    \item $A_i$: ``the edges $g_1,\dotsc,g_i$ are in the sample $\Sam$ at the
      beginning of time step $t_{\ell_3}$.'' \item $B_i$: if $i=1$, this is the
      event ``the edge $g_2$ is inserted in the sample $\Sam$ during time step
      $t_{g_2}$.'' If $i=2$, this event is the whole event space, i.e., the
      event that happens with probability $1$.
    \item $C$: ``neither $g_1$ nor $g_2$ were among the edges removed from
      $\Sam$ between the beginning of time step $t_{\ell_3}$ and the beginning
      of time step $t_{g_3}$.''
  \end{itemize}
  We can rewrite $D_\gamma$ as
  \[
    D_\gamma=A_i\cap B_i\cap C\enspace.
  \]
  Hence
  \begin{align}\label{eq:decomposition}
    \Pr(D_\gamma~|~D_\lambda)&=\Pr\left(A_i\cap
  B_i\cap C~|~D_\lambda\right)\nonumber\\
  &= \Pr\left(A_i~|~D_\lambda\right)\Pr\left(B_i\cap
C~|~A_i\cap D_\lambda\right)\enspace.
  \end{align}
  We now show that
  \[
    \Pr\left(A_i~|~D_\lambda\right)\le
    \Pr\left(A_i\right)\enspace.
  \]
  If we assume that $t_{\ell_3}\le M+1$, then all the edges that appeared on
  the stream up until the beginning of $t_{\ell_3}$ are in $\Sam$. Therefore,
  \[
    \Pr\left(A_i~|~D_\lambda\right)=\Pr(A_i)=1\enspace.
  \]
  Assume instead that $t_{\ell_3}> M+1$. Among the $\binom{t_{\ell_3}-1}{M}$
  subsets of $E^{(t_{\ell_3}-1)}$ of size $M$, there are
  $\binom{t_{\ell_3}-3}{M-2}$ that contain $\ell_1$ and $\ell_2$. From
  Lemma~\ref{lem:reservoir}, we have that at the beginning of time
  $t_{\ell_3}$, $\Sam$ is a subset of size $M$ of $E^{(t_{\ell_3}-1)}$ chosen
  uniformly at random. Hence, if we condition on the fact that
  $\{\ell_1,\ell_2\}\subset\Sam$, we have that $\Sam$ is chosen uniformly at
  random from the $\binom{t_{\ell_3}-3}{M-2}$ subsets of $E^{(t_{\ell_3}-1)}$
  of size $M$ that contain $\ell_1$ and $\ell_2$. Among these, there are
  $\binom{t_{\ell_3}-3-i}{M-2-i}$ that also contain $g_1,\dotsc,g_i$.
  Therefore,
  \[
    \Pr(A_i~|~D_\lambda) =
    \frac{\binom{t_{\ell_3}-3-i}{M-2-i}}{\binom{t_{\ell_3}-3}{M-2}}=\prod_{j=0}^{i-1}\frac{M-2-j}{t_{\ell_{3}}-3-j}\enspace.
  \]
  From Lemma~\ref{lem:reservoirhighorder} we have
  \[
    \Pr(A_i)=\frac{1}{\xi_{i,t_{\ell_3}-1}}=\prod_{j=0}^{i-1}\frac{M-j}{t_{\ell_3}-1-j},
  \]
  where the last equality comes from the assumption $t_{\ell_3}> M+1$. From
  the same assumption and from the fact that for any $j\ge 0$ and any $y\ge x> j$ it holds
  $\frac{x-j}{y-j}\le\frac{x}{y}$, then we have
  \[
    \Pr(A_i~|~D_\lambda)\le \Pr(A_i)\enspace.
  \]
  This implies, from~\eqref{eq:decomposition}, that
  \begin{equation}\label{eq:eventsinequal}
    \Pr(D_\gamma~|~D_\lambda)\le\Pr(A_i)\Pr(B_i\cap C~|~A_i\cap D_\lambda)\enspace.
  \end{equation}

  Consider now the events $B_i$ and $C$. When conditioned on
  $A_i$, these event are \emph{both independent} from
  $D_\lambda$: if the edges $g_1,\dotsc,g_i$ are in $\Sam$ at the beginning of
  time $t_{\ell_3}$, the fact that the edges $\ell_1$ and $\ell_2$ were also
  in $\Sam$ at the beginning of time $t_{\ell_3}$ has no influence whatsoever
  on the actions of the algorithm (i.e., whether an edge is inserted in or
  removed from $\Sam$). Thus,
  \[
    \Pr(A_i)\Pr(B_i\cap C~|~A_i\cap D_\lambda) =
    \Pr(A_i)\Pr(B_i\cap C~|~A_i)\enspace.
  \]
  Putting this together with~\eqref{eq:eventsinequal}, we obtain
  \[
    \Pr(D_\gamma~|~D_\lambda)\le\Pr(A_i)\Pr(B_i\cap
    C~|~A_i)\le\Pr(A_i\cap B_i\cap C)\le\Pr(D_\gamma)\enspace,
  \]
  where the last inequality follows from the fact that $D_\gamma=A_i\cap B_i\cap
  C$ by definition.
\end{proof}

\subsection{Theoretical results for \algofd}\label{app:algofd}

\subsubsection{Expectation}
Before proving Thm.~\ref{thm:fdunbiased} we need the following technical lemmas.

The following is a corollary of~\citep[Thm.~1]{GemullaLH08}.
\begin{lemma}\label{lem:gemulla}
  For any $t>0$, and any $j$, $0\le j\le s^{(t)}$,  let $\mathcal{B}^{(t)}$ be
  the collection of subsets of $E^{(t)}$ of size $j$. For any
  $B\in\mathcal{B}^{(t)}$ it holds
  \[
    \Pr\left(\Sam=B~|~M^{(t)}=j\right)=\frac{1}{\binom{|E^{(t)}|}{j}}\enspace.
  \]
  That is, conditioned on its size at the end of time step $t$, $\Sam$ is
  equally likely to be, at the end of time step $t$, any of the subsets of
  $E^{(t)}$ of that size.
\end{lemma}

The next lemma is an immediate corollary of~\citep[Thm.~2]{GemullaLH08}
\begin{lemma}\label{lem:kappa}
  Recall the definition of $\kappa^{(t)}$ from~\eqref{eq:kappa}. We have
  \[
    \kappa^{(t)}=\Pr(M^{(t)}\ge 3)\enspace.
  \]
\end{lemma}

The next lemma follows from Lemma~\ref{lem:gemulla} in the same way
as Lemma~\ref{lem:reservoirhighorder} follows from Lemma~\ref{lem:reservoir}.

\begin{lemma}\label{lem:gemullahighorder}
  For any time step $t$  and any $j$, $0\le j\le s^{(t)}$, let $B$ be any
  subset of $E^{(t)}$ of size $|B|=k\le s^{(t)}$.
  Then, at the end of time step $t$,
  \[
    \Pr\left(B\subseteq\Sam~|~M^{(t)}=j\right)=\left\{\begin{array}{cl}0 & \text{if } k>j \\
    \displaystyle\frac{1}{\psi_{k,j,s^{(t)}}} &
  \text{otherwise}\end{array}\right.\enspace.
\]
\end{lemma}

The next two lemmas discuss properties of \algofd for $t< t^*$, where $t^*$ is
the first time that $|E^{(t)}|$ had size $M+1$ ($t^*\ge M+1$).

\begin{lemma}\label{lem:equalsize}
  For all $t<t^*$, we have:
  \begin{enumerate}
    \item $d_\mathrm{o}^{(t)}=0$; and
    \item $\Sam=E^{(t)}$; and
    \item $M^{(t)}=s^{(t)}$.
  \end{enumerate}
\end{lemma}

\begin{proof}
  Since the third point in the thesis follows immediately from the second, we
  focus on the first two points.

  The proof is by induction on $t$. In the base base for $t=1$: the
  element on the stream must be an insertion, and the algorithm
  deterministically inserts the edge in $\Sam$.  Assume now that it is true for
  all time steps up to (but excluding) some $t\le t^*-1$. We now show that it is
  also true for $t$.

  Assume the element on the stream at time $t$ is a deletion. The
  corresponding edge must be in $\Sam$, from the inductive hypothesis. Hence
  \algofd removes it from $\Sam$ and increments the counter $d_\mathrm{i}$ by
  1. Thus it is still true that $\Sam=E^{(t)}$ and $d_\mathrm{o}^{(t)}=0$, and
  the thesis holds.

  Assume now that the element on the stream at time $t$ is an insertion. From
  the inductive hypothesis we have that the current value of the counter
  $d_\mathrm{o}$ is 0.

  If the counter $d_\mathrm{i}$ has currently value 0 as well, then, because of the hypothesis that $t<t^*$, it must
  be that $|\Sam|=M^{(t-1)}=s^{(t-1)}< M$. Therefore \algofd always inserts the edge in $\Sam$. Thus it is
  still true that $\Sam=E^{(t)}$ and $d_\mathrm{o}^{(t)}=0$, and the thesis
  holds.

  If otherwise $d_\mathrm{i}>0$, then \algofd flips a biased coin with
  probability of heads equal to
  \[
    \frac{d_\mathrm{i}}{d_\mathrm{i}+d_\mathrm{o}}=\frac{d_\mathrm{i}}{d_\mathrm{i}}=1,
  \]
  therefore \algofd always inserts the edge in $\Sam$ and decrements
  $d_\mathrm{i}$ by one. Thus it is still true that $\Sam=E^{(t)}$ and
  $d_\mathrm{o}^{(t)}=0$, and the thesis holds.
\end{proof}

The following result is an immediate consequence of Lemma~\ref{lem:kappa} and
Lemma~\ref{lem:equalsize}.
\begin{lemma}\label{lem:kappaone}
  For all $t<t^*$ such that $s^{(t)}\ge 3$, we have $\kappa^{(t)}=1$.
\end{lemma}

We can now prove Thm.~\ref{thm:fdunbiased}.

\begin{proof}[of Thm.~\ref{thm:fdunbiased}]
  Assume for now that $t<t^*$. From Lemma~\ref{lem:equalsize}, we
  have that $s^{(t)}=M^{(t)}$. If $M^{(t)}< 3$, then it must be $s^{(t)}<3$,
  hence $|\Delta^{(t)}|=0$ and indeed the algorithm returns $\rho^{(t)}=0$ in
  this case. If instead $M^{(t)}=s^{(t)}\ge3$, then we have
  \[
    \rho^{(t)}=\frac{\tau^{(t)}}{\kappa^{(t)}}\enspace.
  \]
  From Lemma~\ref{lem:kappaone} we have that $\kappa^{(t)}=1$ for all $t<t^*$,
  hence $\rho^{(t)}=\tau^{(t)}$ in these cases. Since (an identical version of) Lemma~\ref{lem:baseunbiasedaux} also
  holds for \algofd, we have $\tau^{(t)}=|\Delta^{\Sam}|=|\Delta^{(t)}|$,
  where the last equality comes from the fact that $\Sam=E^{(t)}$
  (Lemma~\ref{lem:equalsize}). Hence $\rho^{(t)}=|\Delta^{(t)}|$ for any
  $t\le t^*$, as in the thesis.

  Assume now that $t\ge t^*$. Using the law of total expectation, we can write
  \begin{equation}\label{eq:fdconditionalexpect}
    \mathbb{E}\left[\rho^{(t)}\right] = \sum_{j=0}^{\min\{s^{(t)},M\}}
    \mathbb{E}\left[\rho^{(t)}~|~M^{(t)}=j\right]\Pr\left(M^{(t)}=j\right)\enspace.
  \end{equation}

  Assume that $|\Delta^{(t)}|>0$, otherwise, the algorithm deterministically
  returns $0$ as an estimation and the thesis follows. Let $\lambda$ be a
  triangle in $\Delta^{(t)}$, and let $\delta_{\lambda}^{(t)}$ be a random
  variable that takes value
  \[
    \frac{\psi_{3,M^{(t)},s^{(t)}}}{\kappa^{(t)}}=\frac{s^{(t)}(s^{(t)}-2)(s^{(t)}-2)}{M^{(t)}(M^{(t)}-1)(M^{(t)}-2)}\frac{1}{\kappa^{(t)}}
  \]
  if all edges of
  $\lambda$ are in $\Sam$ at the end of the time instant $t$, and $0$
  otherwise. Since (an identical version of) Lemma~\ref{lem:baseunbiasedaux}
  also holds for \algofd, we can write
  \[
    \rho^{(t)}=\sum_{\lambda\in\Delta^{(t)}}\delta_{\lambda}^{(t)}\enspace.
  \]
  Then, using Lemma~\ref{lem:kappa} and Lemma~\ref{lem:gemullahighorder}, we have, for $3\le j\le\min\{M,s^{(t)}\}$,
  \begin{align}\label{eq:fdcondexpectation1}
    \mathbb{E}\left[\rho^{(t)}~|~M^{(t)}=j\right]&=
    \sum_{\lambda\in\Delta^{(t)}}\frac{\psi_{3,j,s^{(t)}}}{\kappa^{(t)}}\Pr\left(\delta_{\lambda}^{(t)}=\frac{\psi_{3,j,s^{(t)}}}{\kappa^{(t)}}~|~M^{(t)}=j\right)\nonumber\\
    &=|\Delta^{(t)}|\frac{\psi_{3,j,s^{(t)}}}{\kappa^{(t)}}\frac{1}{\psi_{3,j,s^{(t)}}}=\frac{1}{\kappa^{(t)}}|\Delta^{(t)}|,
  \end{align}
  and
  \begin{equation}\label{eq:fdcondexpectation2}
    \mathbb{E}\left[\rho^{(t)}~|~M^{(t)}=j\right]=0, \text{if }0\le j\le 2.
  \end{equation}

  Plugging this into~\eqref{eq:fdconditionalexpect}, and using Lemma~\ref{lem:kappa}, we have
  \[
    \mathbb{E}\left[\rho^{(t)}\right] =
    |\Delta^{(t)}|\frac{1}{\kappa^{(t)}}\sum_{j=3}^{\min\{s^{(t)},M\}}\Pr(M^{(t)}=j)=|\Delta^{(t)}|\enspace.
  \]
\end{proof}

\subsubsection{Variance}
We now move to prove Thm.~\ref{thm:fdvariance} about the variance of the \algofd
estimator. We first need some technical lemmas.

\begin{lemma}\label{lem:confdvariance}
  For any time $t\ge t^*$, and any $j$, $3\le j\le \min\{s^{(t)}, M\}$, we have:
  \begin{align}
    \mathrm{Var}\left[\rho^{(t)}|M^{(t)}=j\right]
    = (\kappa^{(t)})^{-2}&\left(|\Delta^{(t)}|\left(\psi_{3,j,s^{(t)}}-1 \right)
    + r^{(t)}\left(\psi_{3,j,s^{(t)}}^2\psi_{5,j,s^{(t)}}^{-1}-1\right)\right.\nonumber\\
    &\left.+w^{(t)}\left(\psi_{3,j,s^{(t)}}^2\psi_{6,j,s^{(t)}}^{-1}-1\right)\right)\label{eq:globvariancefdcond}
  \end{align}
  An analogous result holds for any $u\in V^{(t)}$, replacing the global
  quantities with the corresponding local ones.
\end{lemma}
\begin{proof}
  The proof is analogous to that of Theorem~\ref{thm:basevariance}, using $j$
  in place of $M$, $s^{(t)}$ in place of $t$, $\psi_{a,j,s^{(t)}}$ in
  place of $\xi_{a,t}$, and using Lemma~\ref{lem:gemullahighorder} instead of
  Lemma~\ref{lem:reservoirhighorder}. The additional $(k^{(t)})^{-2}$
  multiplicative term comes from the $(k^{(t)})^{-1}$ term used in
  the definition of $\rho^{(t)}$.
\end{proof}
The term $w^{(t)}\left(\psi_{3,j,s^{(t)}}^2\psi_{6,j,s^{(t)}}^{-1}-1\right)$ is
non-positive.
\begin{lemma}\label{lem:boundvarfd}
  For any time $t\ge t^*$, and any $j$, $6< j\le \min\{s^{(t)}, M\}$, if  $s^{(t)}\geq M$ we have:
  \begin{align*}
    \mathrm{Var}\left[\rho^{(t)}|M^{(t)} = i\right]
    &\leq (\kappa^{(t)})^{-2}\left(|\Delta^{(t)}|\left(\psi_{3,j,s^{(t)}}-1
  \right) + r^{(t)}\left(\psi_{3,j,s^{(t)}}^2\psi_{5,j,s^{(t)}}^{-1}-1\right)
\right),\text{ for }i\geq j \\
\mathrm{Var}\left[\rho^{(t)}|M^{(t)}= i\right]
&\leq (\kappa^{(t)})^{-2}\left(|\Delta^{(t)}|\left(\psi_{3,3,s^{(t)}}-1
\right) +
r^{(t)}\left(\psi_{3,5,s^{(t)}}^2\psi_{5,5,s^{(t)}}^{-1}-1\right)\right),\text{
for }i < j
  \end{align*}
  An analogous result holds for any $u\in V^{(t)}$, replacing the global
  quantities with the corresponding local ones.
\end{lemma}
\begin{proof}
  The proof follows by observing that the term
  $w^{(t)}\left(\psi_{3,j,s^{(t)}}^2\psi_{6,j,s^{(t)}}^{-1}-1\right)$ is
  non-positive, and that~\eqref{eq:globvariancefdcond} is a non-increasing
  function of the sample size.
\end{proof}

The following lemma deals with properties of the r.v.~$M^{(t)}$.
\begin{lemma}\label{lem:mtprops}
  Let $t>t^*$, with $s^{(t)}\geq M$. Let $d^{(t)}=d_o^{(t)}+d_i^{(t)}$ denote
  the total number of unpaired deletions at time $t$.\footnote{While both
  $d_o^{(t)}$ and $d_i^{(t)}$ are r.v.s, their sum is not.} The sample size
  $M^{(t)}$ follows the hypergeometric distribution:\footnote{We
    use here the convention that $\binom{0}{0}=1$, and
  $\binom{k}{0}=1$.}
  \begin{equation}\label{eq:samplesize}
    \Pr\left(M^{(t)} = j \right) =
    \left\{\begin{array}{ll}\binom{s^{(t)}}{j} \binom{d^{(t)}}{M-j}\big/
        \binom{s^{(t)}+d^{(t)}}{M} & \text{for }
        \max\{M-d^{(t)},0\}\leq j \leq M\\
      0 & \text{otherwise} \end{array}\right.\enspace.
      \end{equation}
      We have
      \begin{equation}
        \mathbb{E}\left[M^{(t)}\right] = M\frac{s^{(t)}}{s^{(t)}+d^{(t)}},\label{eq:expnumdeletion}
      \end{equation}
      and for any $0< c< 1$
      \begin{equation}
        \Pr\left(M^{(t)}  > \mathbb{E}\left[M^{(t)}\right]-cM\right)\geq 1 -
        \frac{1}{e^{2c^2M}}\label{eq:boundsamplesize}\enspace.
      \end{equation}
    \end{lemma}
    \begin{proof}
      Since $t>t*$, from the definition of $t^*$ we have that the
      $M^{(t)}$ has reached size $M$ at least once (at $t^*$). From this and the
      definition of $d^{(t)}$ (number of uncompensated deletion), we have that
      $M^{(t)}$ can not be less than $M-d^{(t)}$. The rest of the proof
      for~\eqref{eq:samplesize} and for~\eqref{eq:expnumdeletion} follows
      from~\citep[Thm.~2]{GemullaLH08}.

      The concentration bound in~\eqref{eq:boundsamplesize} follows from the
      properties of the hypergeometric distribution discussed by~\citet{Skala13}.
    \end{proof}

    The following is an immediate corollary from Lemma~\ref{lem:mtprops}.
    \begin{corollary} \label{corol:samplesizebound}
      Consider the execution of \algofd at time $t>t^*$. Suppose we have $d^{(t)}\leq \alpha s^{(t)}$, with $0\leq\alpha< 1$ and $s^{(t)}\geq M$. If $M\geq \frac{1}{2\sqrt{\alpha'-\alpha}}c'\ln s^{(t)} $ for $\alpha < \alpha'<1$, we have:
      \[
        \Pr\left(M^{(t)}  \geq M (1-\alpha') \right)>1 - \frac{1}{\left(s^{(t)}\right)^{c'}}.
      \]
    \end{corollary}

    We can now prove Thm.~\ref{thm:fdvariance}.
    \begin{proof}[of Thm.~\ref{thm:fdvariance}]
      From the law of total variance we have:
      \begin{align*}
        \mathrm{Var}\left[\rho^{(t)}\right] &= \sum_{j=0}^M \mathrm{Var}\left[\rho^{(t)}|M^{(t)}= j\right] \Pr\left(M^{(t)}= j\right) \\
        &+ \sum_{j=0}^M \mathbb{E}\left[\rho^{(t)}|M^{(t)}= j\right]^2(1-\Pr\left(M^{(t)}= j\right))\Pr\left(M^{(t)}= j\right)\\
        &-2\sum_{j=1}^M\sum_{i=0}^{j-1} \mathbb{E}\left[\rho^{(t)}|M^{(t)}= j\right]\Pr\left(M^{(t)}= j\right)\mathbb{E}\left[\rho^{(t)}|M^{(t)}= i\right]\Pr\left(M^{(t)}= i\right).
      \end{align*}
      As shown in~\eqref{eq:fdcondexpectation1} and~\eqref{eq:fdcondexpectation2}, for any $j=0,1,\ldots,M$ we have $\mathbb{E}\left[\rho^{(t)}|M^{(t)}= j\right]\geq 0$. This in turn implies:
      \begin{align}
        \mathrm{Var}\left[\rho^{(t)}\right] &\leq \sum_{j=0}^M \mathrm{Var}\left[\rho^{(t)}|M^{(t)}= j\right] \Pr\left(M^{(t)}= j\right)\nonumber \\
        &+ \sum_{j=0}^M \mathbb{E}\left[\rho^{(t)}|M^{(t)}= j\right]^2(1-\Pr\left(M^{(t)}= j\right))\Pr\left(M^{(t)}= j\right) \label{eq:finvarfd}.
      \end{align}
      Let us consider separately the two main components of \eqref{eq:finvarfd}. From Lemma~\ref{lem:boundvarfd} we have:
      \begin{align}
        &\sum_{j=0}^M \mathrm{Var}\left[\rho^{(t)}|M^{(t)}= j\right] \Pr\left(M^{(t)}= j\right)
        =\\
      &\sum_{j\geq M(1-\alpha')}^M \mathrm{Var}\left[\rho^{(t)}|M^{(t)}=j \right] \Pr\left(M^{(t)}=j)\right)
        + \sum_{j= 0}^{ M(1-\alpha')}\mathrm{Var}\left[\rho^{(t)}|M^{(t)}=j\right] \Pr\left(M^{(t)}=j\right)\nonumber\\
        &\leq (\kappa^{(t)})^{-2}\left(|\Delta^{(t)}|\left(\psi_{3,j,s^{(t)}}-1
        \right) + r^{(t)}\left(\psi_{3,j,s^{(t)}}^2\psi_{5,j,s^{(t)}}^{-1}-1\right)
        \right)\times\Pr\left(M^{(t)} > M(1-\alpha')\right)\nonumber\\
        &\leq (\kappa^{(t)})^{-2}\left(|\Delta^{(t)}|\frac{\left(s^{(t)}\right)^3}{6} + r^{(t)}\frac{s^{(t)}}{6}\right) \Pr\left(M^{(t)} \leq M(1-\alpha')\right)\label{eq:totvariance1}
      \end{align}

      According to our hypothesis $M\geq \frac{1}{2\sqrt{\alpha'-\alpha}}7\ln
      s^{(t)}$, thus we have, from Corollary~\ref{corol:samplesizebound}:
      \begin{equation*}
      \Pr\left(M^{(t)}\leq M (1-\alpha'))\right) \leq \frac{1}{(s^{(t)})^7}.
    \end{equation*}

    As $r^{(t)}<|\Delta^{(t)}|^2$ and $|\Delta^{(t)}|\leq (s^{(t)})^3$ we have:
    \begin{equation*}
      (\kappa^{(t)})^{-2}\left(|\Delta^{(t)}|\frac{\left(s^{(t)}\right)^3}{6} + r^{(t)}\frac{s^{(t)}}{6}\right) \Pr\left(M^{(t)} \leq M(1-\alpha')\right) \leq (\kappa^{(t)})^{-2}
    \end{equation*}
    We can therefore rewrite~\eqref{eq:totvariance1} as:
    \begin{align} \label{eq:variancefdpart1}
      \sum_{j=0}^M \mathrm{Var}\left[\rho^{(t)}|M^{(t)}= j\right] \Pr\left(M^{(t)}= j\right) &\leq (\kappa^{(t)})^{-2}\left(|\Delta^{(t)}|\left(\psi_{3,M(1-\alpha'),s^{(t)}}-1 \right)\right) \nonumber\\
      &+ (\kappa^{(t)})^{-2}\left(r^{(t)}\left(\psi_{3,M(1-\alpha'),s^{(t)}}^2\psi_{5,M(1-\alpha'),s^{(t)}}^{-1}-1\right)+1\right).
    \end{align}

    Let us now consider the term $\sum_{j=0}^M \mathbb{E}\left[\rho^{(t)}|M^{(t)}=
    j\right]^2(1-\Pr\left(M^{(t)}= j\right))\Pr\left(M^{(t)}= j\right) $. Recall
    that, from~\eqref{eq:fdcondexpectation1} and~\eqref{eq:fdcondexpectation2}, we
    have $\mathbb{E}\left[\rho^{(t)}|M^{(t)}= j\right] =
    |\Delta^{(t)}|(\kappa^{(t)})^{-1}$ for $j= 3,\ldots,M$, and
    $\mathbb{E}\left[\rho^{(t)}|M^{(t)}= j\right] = 0$ for $j=0,1,2$. From
    Corollary~\ref{corol:samplesizebound} we have that for $j\leq (1-\alpha')M$
    and $M\geq \frac{1}{2\sqrt{\alpha'-\alpha}}7\ln s^{(t)}$
    \begin{equation*}
      \Pr\left(M^{(t)}= j\right)\leq \Pr\left(M^{(t)}\leq (1-\alpha')M\right) \leq \frac{1}{\left(s^{(t)}\right)^7},
    \end{equation*}
    and thus:
    \begin{align}\label{eq:variancefdpart2}
      \sum_{j=0}^{(1-\alpha')M} \mathbb{E}\left[\rho^{(t)}|M^{(t)}= j\right]^2(1-\Pr\left(M^{(t)}= j\right))\Pr\left(M^{(t)}= j\right) &\leq \frac{(1-\alpha')M |\Delta^{(t)}|^2(\kappa^{(t)})^{-2}}{\left(s^{(t)}\right)^7}\nonumber\\
      &\leq (1-\alpha')(\kappa^{(t)})^{-2},
    \end{align}
    where the last passage follows since, by hypothesis, $M\leq s^{(t)}$.

    Let us now consider the values $j>(1-\alpha')M$, we have:
    \begin{align}\label{eq:variancefdpart3}
      &\sum_{j>(1-\alpha')M}^M \mathbb{E}\left[\rho^{(t)}|M^{(t)}= j\right]^2(1-\Pr\left(M^{(t)}= j\right))\Pr\left(M^{(t)}= j\right)&\ \ \ \ \ \ \ \ \ \ \ \ \ \ \ \ \ \ \nonumber \\
      &\ \ \ \ \ \ \ \ \ \ \ \ \ \ \ \ \ \ \leq \alpha'M |\Delta^{(t)}|^2(\kappa^{(t)})^{-2}\left(1- \sum_{j>(1-\alpha')M}^M \Pr\left(M^{(t)}= j\right)\right)\nonumber\\
      &\ \ \ \ \ \ \ \ \ \ \ \ \ \ \ \ \ \ \leq \alpha'M |\Delta^{(t)}|^2(\kappa^{(t)})^{-2}\left(1- \Pr\left(M^{(t)}> (1-\alpha')M\right)\right) \nonumber\\
      &\ \ \ \ \ \ \ \ \ \ \ \ \ \ \ \ \ \ \leq \frac{\alpha'M |\Delta^{(t)}|^2(\kappa^{(t)})^{-2}}{\left(s^{(t)}\right)^7}\nonumber\\
      &\ \ \ \ \ \ \ \ \ \ \ \ \ \ \ \ \ \ \leq \alpha'(\kappa^{(t)})^{-2},
    \end{align}
    where the last passage follows since, by hypothesis, $M\leq s^{(t)}$.

    The theorem follows from composing the upper bounds obtained
    in~\eqref{eq:variancefdpart1},~\eqref{eq:variancefdpart2}
    and~\eqref{eq:variancefdpart3} according to~\eqref{eq:finvarfd}.
  \end{proof}

\subsubsection{Concentration}
We now prove Thm.~\ref{thm:fdconcentration} about \algofd.

\begin{proof}[of Thm.~\ref{thm:fdconcentration}]
  By Chebyshev’s inequality it is sufficient to prove that
  $$\frac{\variance[\rho^{(t)}]}{\varepsilon^2|\Delta^{(t)}|^2} <\delta\enspace.$$
  From Lemma~\ref{thm:fdvariance}, for $M\geq \frac{1}{\sqrt{a'-\alpha}}7\ln s^{(t)}$ we have:
  \begin{align*}
    \mathrm{Var}\left[\rho^{(t)}\right] &\leq
    (\kappa^{(t)})^{-2}|\Delta^{(t)}|\left(\psi_{3,M(1-\alpha'),s^{(t)}}-1 \right) \\
    &+(\kappa^{(t)})^{-2}r^{(t)}\left(\psi_{3,M(1-\alpha'),s^{(t)}}^2\psi_{5,M(1-\alpha'),s^{(t)}}^{-1}-1\right)\\
    &+(\kappa^{(t)})^{-2}2
  \end{align*}
  Let $M'= (1-\alpha')M$. In order to verify that the lemma holds, it is sufficient to impose the following two conditions:
  \begin{description}
    \item[Condition (1)]
      \begin{equation*}
        \frac{\delta}{2} > \frac{(\kappa^{(t)})^{-2}\left(|\Delta^{(t)}|\left(\psi_{3,M(1-\alpha'),s^{(t)}}-1 \right)+2\right)}{\varepsilon^2|\Delta^{(t)}|^2}.
      \end{equation*}
      As by hypothesis $|\Delta^{(t)}|>0$, we can rewrite this condition as:
      \begin{equation*}
        \frac{\delta}{2} > \frac{(\kappa^{(t)})^{-2}\left(\psi_{3,M(1-\alpha'),s^{(t)}}-(\frac{|\Delta^{(t)}|-2}{|\Delta^{(t)}|} \right)}{\varepsilon^2|\Delta^{(t)}|}
        \end{equation*}
        which is verified for:
        \begin{align*}
          M'(M'-1)(M'-2) &> \frac{2s^{(t)}(s^{(t)}-1)(s^{(t)}-2)}{\delta \varepsilon^2|\Delta^{(t)}|(\kappa^{(t)})^{2}+2\frac{|\Delta^{(t)}|-2}{|\Delta^{(t)}|}},\\
          M'&> \sqrt[3]{\frac{2s^{(t)}(s^{(t)}-1)(s^{(t)}-2)}{\delta \varepsilon^2|\Delta^{(t)}|(\kappa^{(t)})^{2}+2\frac{|\Delta^{(t)}|-2}{|\Delta^{(t)}|}}}+2, \\
          M &> (1-\alpha')^{-1}\left(\sqrt[3]{\frac{2s^{(t)}(s^{(t)}-1)(s^{(t)}-2)}{\delta \varepsilon^2|\Delta^{(t)}|(\kappa^{(t)})^{2}+2\frac{|\Delta^{(t)}|-2}{|\Delta^{(t)}|}}}+2 \right).
        \end{align*}
      \item[Condition (2)]
        \begin{equation}
          \frac{\delta}{2} > \frac{(\kappa^{(t)})^{-2}}{\varepsilon^2|\Delta^{(t)}|^2} r^{(t)}\left(\psi_{3,M(1-\alpha'),s^{(t)}}^2\psi_{5,M(1-\alpha'),s^{(t)}}^{-1}-1\right). \label{eq:fdcond2}
        \end{equation}

        As we have:
        \begin{equation*}
          (\kappa^{(t)})^{-2}r^{(t)}\left(\psi_{3,M(1-\alpha'),s^{(t)}}^2\psi_{5,M(1-\alpha'),s^{(t)}}^{-1}-1\right) \leq (\kappa^{(t)})^{-2}r^{(t)}\left(\frac{s^{(t)}}{6M(1-\alpha')}-1\right)
        \end{equation*}
        The condition~\eqref{eq:fdcond2} is verified for:
        \begin{equation*}
          M > \frac{(1-\alpha')^{-1}}{3} \left( \frac{r^{(t)}s^{(t)}}{\delta \varepsilon^2 |\Delta^{(t)}|^2(\kappa^{(t)})^{-2}+2r^{(t)}}\right).
        \end{equation*}
    \end{description}
    The theorem follows.
  \end{proof}

\begin{acknowledgments}
  This work was supported in part by NSF grant IIS-1247581 and NIH grant
  R01-CA180776.
\end{acknowledgments}

\bibliographystyle{abbrvnat}
\bibliography{bibliography_full}

\begin{thebibliography}{42}
\providecommand{\natexlab}[1]{#1}
\providecommand{\url}[1]{\texttt{#1}}
\expandafter\ifx\csname urlstyle\endcsname\relax
  \providecommand{\doi}[1]{doi: #1}\else
  \providecommand{\doi}{doi: \begingroup \urlstyle{rm}\Url}\fi

\bibitem[Ahmed et~al.(2014)Ahmed, Duffield, Neville, and Kompella]{AhmedDKN14}
N.~K. Ahmed, N.~Duffield, J.~Neville, and R.~Kompella.
\newblock Graph {S}ample and {H}old: A framework for big-graph analytics.
\newblock In \emph{Proceedings of the 20th ACM SIGKDD International Conference
  on Knowledge Discovery and Data Mining}, KDD '14, pages 1446--1455. ACM,
  2014.

\bibitem[Bar-Yossef et~al.(2002)Bar-Yossef, Kumar, and
  Sivakumar]{BarYossefKS02}
Z.~Bar-Yossef, R.~Kumar, and D.~Sivakumar.
\newblock Reductions in streaming algorithms, with an application to counting
  triangles in graphs.
\newblock In \emph{Proceedings of the Thirteenth Annual ACM-SIAM Symposium on
  Discrete Algorithms}, SODA '02, pages 623--632. SIAM, 2002.

\bibitem[Becchetti et~al.(2010)Becchetti, Boldi, Castillo, and
  Gionis]{BecchettiBCG10}
L.~Becchetti, P.~Boldi, C.~Castillo, and A.~Gionis.
\newblock Efficient algorithms for large-scale local triangle counting.
\newblock \emph{ACM Transactions on Knowledge Discovery from Data}, 4\penalty0
  (3):\penalty0 13:1--13:28, 2010.

\bibitem[Berry et~al.(2011)Berry, Hendrickson, LaViolette, and
  Phillips]{BerryHLVP11}
J.~W. Berry, B.~Hendrickson, R.~A. LaViolette, and C.~A. Phillips.
\newblock Tolerating the community detection resolution limit with edge
  weighting.
\newblock \emph{Physical Review E}, 83\penalty0 (5):\penalty0 056119, 2011.

\bibitem[Boldi et~al.(2011)Boldi, Rosa, Santini, and Vigna]{brsllp}
P.~Boldi, M.~Rosa, M.~Santini, and S.~Vigna.
\newblock Layered label propagation: A multiresolution coordinate-free ordering
  for compressing social networks.
\newblock In \emph{Proceedings of the 20th International Conference on World
  Wide Web}, WWW '11, pages 587--596. ACM, 2011.

\bibitem[Bulteau et~al.(2015)Bulteau, Froese, Kutzkov, and Pagh]{BulteauFKP15}
L.~Bulteau, V.~Froese, K.~Kutzkov, and R.~Pagh.
\newblock Triangle counting in dynamic graph streams.
\newblock \emph{Algorithmica}, pages 1--20, 2015.
\newblock ISSN 1432-0541.
\newblock \doi{10.1007/s00453-015-0036-4}.

\bibitem[Buriol et~al.(2006)Buriol, Frahling, Leonardi, Marchetti-Spaccamela,
  and Sohler]{BuriolFSMSS06}
L.~S. Buriol, G.~Frahling, S.~Leonardi, A.~Marchetti-Spaccamela, and C.~Sohler.
\newblock Counting triangles in data streams.
\newblock In \emph{Proceedings of the 25th ACM SIGMOD-SIGACT-SIGART Symposium
  on Principles of Database Systems}, PODS '06, pages 253--262. ACM, 2006.

\bibitem[Celma~Herrada(2009)]{celma2009music}
{\`O}.~Celma~Herrada.
\newblock Music recommendation and discovery in the long tail.
\newblock Technical report, Universitat Pompeu Fabra, 2009.

\bibitem[Cohen et~al.(2012)Cohen, Cormode, and Duffield]{CohenCD12}
E.~Cohen, G.~Cormode, and N.~Duffield.
\newblock Don't let the negatives bring you down: sampling from streams of
  signed updates.
\newblock \emph{ACM SIGMETRICS Performance Evaluation Review}, 40\penalty0
  (1):\penalty0 343--354, 2012.

\bibitem[Datasets()]{yahoo-ans}
Y.~R.~W. Datasets.
\newblock Yahoo!\ {A}nswers browsing behavior version 1.0.
\newblock \url{http://webscope.sandbox.yahoo.com}.

\bibitem[De~Stefani et~al.(2016)De~Stefani, Epasto, Riondato, and
  Upfal]{DeStefaniERU16}
L.~De~Stefani, A.~Epasto, M.~Riondato, and E.~Upfal.
\newblock {TRI\`{E}ST}: Counting local and global triangles in fully-dynamic
  streams with fixed memory size.
\newblock In \emph{Proceedings of the 22nd ACM SIGKDD International Conference
  on Knowledge Discovery and Data Mining}, KDD '16. ACM, 2016.

\bibitem[Deville and Till{\'e}(2004)]{DevilleT04}
J.-C. Deville and Y.~Till{\'e}.
\newblock Efficient balanced sampling: The cube method.
\newblock \emph{Biometrika}, 91\penalty0 (4):\penalty0 893--912, 2004.
\newblock \doi{10.1093/biomet/91.4.893}.

\bibitem[Eckmann and Moses(2002)]{EckmannM02}
J.-P. Eckmann and E.~Moses.
\newblock Curvature of co-links uncovers hidden thematic layers in the {W}orld
  {W}ide {W}eb.
\newblock \emph{Proceedings of the National Academy of Sciences}, 99\penalty0
  (9):\penalty0 5825--5829, 2002.

\bibitem[Epasto et~al.(2015{\natexlab{a}})Epasto, Lattanzi, Mirrokni, Sebe,
  Taei, and Verma]{epasto2015ego}
A.~Epasto, S.~Lattanzi, V.~Mirrokni, I.~O. Sebe, A.~Taei, and S.~Verma.
\newblock Ego-net community mining applied to friend suggestion.
\newblock \emph{Proceedings of the VLDB Endowment}, 9\penalty0 (4):\penalty0
  324--335, 2015{\natexlab{a}}.

\bibitem[Epasto et~al.(2015{\natexlab{b}})Epasto, Lattanzi, and
  Sozio]{epasto2015efficient}
A.~Epasto, S.~Lattanzi, and M.~Sozio.
\newblock Efficient densest subgraph computation in evolving graph.
\newblock In \emph{Proceedings of the 24th International Conference on World
  Wide Web}, WWW '15, pages 300--310. ACM, 2015{\natexlab{b}}.

\bibitem[Gemulla et~al.(2008)Gemulla, Lehner, and Haas]{GemullaLH08}
R.~Gemulla, W.~Lehner, and P.~J. Haas.
\newblock Maintaining bounded-size sample synopses of evolving datasets.
\newblock \emph{The VLDB Journal}, 17\penalty0 (2):\penalty0 173--201, 2008.

\bibitem[Hajnal and Szemer\'{e}di(1970)]{hajnal1970proof}
A.~Hajnal and E.~Szemer\'{e}di.
\newblock Proof of a conjecture of {P}.~{E}rd{\H{o}}s.
\newblock In \emph{Combinatorial theory and its applications, II (Proc.
  Colloq., Balatonf{\"u}red, 1969)}, pages 601--623, 1970.

\bibitem[Hall et~al.(2001)Hall, Jaffe, and Trajtenberg]{hjt01}
B.~H. Hall, A.~B. Jaffe, and M.~Trajtenberg.
\newblock The {NBER} patent citation data file: Lessons, insights and
  methodological tools.
\newblock Technical report, National Bureau of Economic Research, 2001.

\bibitem[Hyndman and Koehler(2006)]{hk06}
R.~J. Hyndman and A.~B. Koehler.
\newblock Another look at measures of forecast accuracy.
\newblock \emph{International Journal of Forecasting}, 22\penalty0
  (4):\penalty0 679--688, 2006.

\bibitem[Jha et~al.(2015)Jha, Seshadhri, and Pinar]{JhaSP15}
M.~Jha, C.~Seshadhri, and A.~Pinar.
\newblock A space-efficient streaming algorithm for estimating transitivity and
  triangle counts using the birthday paradox.
\newblock \emph{ACM Transactions on Knowledge Discovery from Data}, 9\penalty0
  (3):\penalty0 15:1--15:21, 2015.

\bibitem[Jowhari and Ghodsi(2005)]{JowharyG05}
H.~Jowhari and M.~Ghodsi.
\newblock New streaming algorithms for counting triangles in graphs.
\newblock In \emph{Computing and Combinatorics: 11th Annual International
  Conference}, COCOON '05, pages 710--716. Springer, 2005.

\bibitem[Kane et~al.(2012)Kane, Mehlhorn, Sauerwald, and Sun]{KaneMSS12}
D.~M. Kane, K.~Mehlhorn, T.~Sauerwald, and H.~Sun.
\newblock Counting arbitrary subgraphs in data streams.
\newblock In A.~Czumaj, K.~Mehlhorn, A.~Pitts, and R.~Wattenhofer, editors,
  \emph{Automata, Languages, and Programming}, volume 7392 of \emph{Lecture
  Notes in Computer Science}, pages 598--609. Springer, 2012.

\bibitem[Kolountzakis et~al.(2012)Kolountzakis, Miller, Peng, and
  Tsourakakis]{KolountzakisMPT12}
M.~N. Kolountzakis, G.~L. Miller, R.~Peng, and C.~E. Tsourakakis.
\newblock Efficient triangle counting in large graphs via degree-based vertex
  partitioning.
\newblock \emph{Internet Mathematics}, 8\penalty0 (1--2):\penalty0 161--185,
  2012.

\bibitem[Kutzkov and Pagh(2013)]{KutzkovP13}
K.~Kutzkov and R.~Pagh.
\newblock On the streaming complexity of computing local clustering
  coefficients.
\newblock In \emph{Proceedings of the 6th ACM International Conference on Web
  Search and Data Mining}, WSDM '13, pages 677--686. ACM, 2013.

\bibitem[Kwak et~al.(2010)Kwak, Lee, Park, and Moon]{kwak2010twitter}
H.~Kwak, C.~Lee, H.~Park, and S.~Moon.
\newblock What is {T}witter, a social network or a news media?
\newblock In \emph{Proceedings of the 19th International Conference on World
  Wide Web}, WWW '10, pages 591--600. ACM, 2010.

\bibitem[Latapy(2008)]{Latapy08}
M.~Latapy.
\newblock Main-memory triangle computations for very large (sparse (power-law))
  graphs.
\newblock \emph{Theoretical Computer Science}, 407\penalty0 (1):\penalty0
  458--473, 2008.

\bibitem[Leskovec et~al.(2007)Leskovec, Kleinberg, and
  Faloutsos]{leskovec2007graph}
J.~Leskovec, J.~Kleinberg, and C.~Faloutsos.
\newblock Graph evolution: Densification and shrinking diameters.
\newblock \emph{ACM Transactions on Knowledge Discovery from Data}, 1\penalty0
  (1):\penalty0 2, 2007.

\bibitem[Lim and Kang(2015)]{LimK15}
Y.~Lim and U.~Kang.
\newblock {MASCOT}: {M}emory-efficient and accurate sampling for counting local
  triangles in graph streams.
\newblock In \emph{Proceedings of the 21st ACM SIGKDD International Conference
  on Knowledge Discovery and Data Mining}, KDD '15, pages 685--694. ACM, 2015.

\bibitem[Manjunath et~al.(2011)Manjunath, Mehlhorn, Panagiotou, and
  Sun]{ManjunathMPS11}
M.~Manjunath, K.~Mehlhorn, K.~Panagiotou, and H.~Sun.
\newblock Approximate counting of cycles in streams.
\newblock In \emph{European Symposium on Algorithms}, ESA '11, pages 677--688.
  Springer, 2011.

\bibitem[Milo et~al.(2002)Milo, Shen-Orr, Itzkovitz, Kashtan, Chklovskii, and
  Alon]{MiloSOIKA02}
R.~Milo, S.~Shen-Orr, S.~Itzkovitz, N.~Kashtan, D.~Chklovskii, and U.~Alon.
\newblock Network motifs: simple building blocks of complex networks.
\newblock \emph{Science}, 298\penalty0 (5594):\penalty0 824--827, 2002.

\bibitem[Mitzenmacher and Upfal(2005)]{mitzenmacher2005probability}
M.~Mitzenmacher and E.~Upfal.
\newblock \emph{Probability and computing: Randomized algorithms and
  probabilistic analysis}.
\newblock Cambridge University Press, 2nd edition, 2005.

\bibitem[Pagh and Tsourakakis(2012)]{PaghT12}
R.~Pagh and C.~E. Tsourakakis.
\newblock Colorful triangle counting and a {M}ap{R}educe implementation.
\newblock \emph{Information Processing Letters}, 112\penalty0 (7):\penalty0
  277--281, Mar. 2012.

\bibitem[Park et~al.(2014)Park, Silvestri, Kang, and Pagh]{ParkSKP14}
H.~Park, F.~Silvestri, U.~Kang, and R.~Pagh.
\newblock {M}ap{R}educe triangle enumeration with guarantees.
\newblock In \emph{Proceedings of the 23rd {ACM} International Conference on
  Conference on Information and Knowledge Management}, CIKM '14, pages
  1739--1748. ACM, 2014.
\newblock \doi{10.1145/2661829.2662017}.

\bibitem[Park and Chung(2013)]{ParkC13}
H.-M. Park and C.-W. Chung.
\newblock An efficient {M}ap{R}educe algorithm for counting triangles in a very
  large graph.
\newblock In \emph{Proceedings of the 22nd ACM International Conference on
  Conference on Information \& Knowledge Management}, CIKM '13, pages 539--548.
  ACM, 2013.

\bibitem[Park et~al.(2016)Park, Myaeng, and Kang]{ParkMK16}
H.-M. Park, S.-H. Myaeng, and U.~Kang.
\newblock {PTE}: Enumerating trillion triangles on distributed systems.
\newblock In \emph{Proceedings of the 22nd ACM SIGKDD International Conference
  on Knowledge Discovery and Data Mining}, KDD '16. ACM, 2016.

\bibitem[Pavan et~al.(2013)Pavan, Tangwongsan, Tirthapura, and Wu]{PavanTTW13}
A.~Pavan, K.~Tangwongsan, S.~Tirthapura, and K.-L. Wu.
\newblock Counting and sampling triangles from a graph stream.
\newblock \emph{Proceedings of the VLDB Endowment}, 6\penalty0 (14):\penalty0
  1870--1881, 2013.

\bibitem[Skala(2013)]{Skala13}
M.~Skala.
\newblock Hypergeometric tail inequalities: ending the insanity.
\newblock \emph{arXiv preprint}, 1311.5939, 2013.

\bibitem[Suri and Vassilvitskii(2011)]{SuriV11}
S.~Suri and S.~Vassilvitskii.
\newblock Counting triangles and the curse of the last reducer.
\newblock In \emph{Proceedings of the 20th International Conference on World
  Wide Web}, WWW '11, pages 607--614. ACM, 2011.

\bibitem[{The Koblenz Network Collection (KONECT)}()]{koblenz}
{The Koblenz Network Collection (KONECT)}.
\newblock {Last.fm} song network dataset.
\newblock \url{http://konect.uni-koblenz.de/networks/lastfm_song}.

\bibitem[Tsourakakis et~al.(2009)Tsourakakis, Kang, Miller, and
  Faloutsos]{TsourakakisKMF09}
C.~E. Tsourakakis, U.~Kang, G.~L. Miller, and C.~Faloutsos.
\newblock Doulion: counting triangles in massive graphs with a coin.
\newblock In \emph{Proceedings of the 15th ACM SIGKDD International Conference
  on Knowledge Discovery and Data Mining}, KDD '09, pages 837--846. ACM, 2009.

\bibitem[Tsourakakis et~al.(2011)Tsourakakis, Kolountzakis, and
  Miller]{TsourakakisKM11}
C.~E. Tsourakakis, M.~N. Kolountzakis, and G.~L. Miller.
\newblock Triangle sparsifiers.
\newblock \emph{Journal of Graph Algorithms and Applications}, 15\penalty0
  (6):\penalty0 703--726, 2011.

\bibitem[Vitter(1985)]{Vitter85}
J.~S. Vitter.
\newblock Random sampling with a reservoir.
\newblock \emph{ACM Transactions on Mathematical Software}, 11\penalty0
  (1):\penalty0 37--57, 1985.

\end{thebibliography}

\received{Month Year}{Month Year}{Month Year}
\end{document}